\documentclass[12pt]{article}
\usepackage{algorithm,algorithmic,amsmath,amssymb,epsf,epsfig,amsthm,times,ifthen,epic,psfrag,array,etoolbox}
\usepackage{enumerate,subcaption,multirow,makecell,enumitem}
\usepackage{printtime}
\usepackage{cite}     
\usepackage{fancyhdr}
\usepackage{setspace} 
\usepackage{lastpage} 
\usepackage{url}
\usepackage{hyperref}

\def\submitteddate{January 14, 2016}
\def\reviseddate{}

\renewcommand{\baselinestretch}{1.0}
\DeclareFontFamily{U}{matha}{\hyphenchar\font45}
\DeclareFontShape{U}{matha}{m}{n}{
      <5> <6> <7> <8> <9> <10> gen * matha
      <10.95> matha10 <12> <14.4> <17.28> <20.74> <24.88> matha12
      }{}
\DeclareSymbolFont{matha}{U}{matha}{m}{n}
\DeclareMathSymbol{\notdivides}{3}{matha}{"1F}
\DeclareMathSymbol{\divides}{3}{matha}{"17}

\begin{document}

\newcommand{\creationtime}{\today \ \ @ \theampmtime}

\pagestyle{fancy}
\renewcommand{\headrulewidth}{0cm}
\chead{\footnotesize{Connelly-Zeger}}
\rhead{\footnotesize{\reviseddate}}
\rhead{\footnotesize{\submitteddate}}
\lhead{}
\cfoot{Page \arabic{page} of \pageref{LastPage}} 

\renewcommand{\qedsymbol}{$\blacksquare$} 

\newtheorem{theorem}              {Theorem}     [section]
\newtheorem{lemma}      [theorem] {Lemma}
\newtheorem{corollary}  [theorem] {Corollary}
\newtheorem{proposition}[theorem] {Proposition}
\newtheorem{remark}     [theorem] {Remark}
\newtheorem{conjecture} [theorem] {Conjecture}
\newtheorem{example}    [theorem] {Example}

\theoremstyle{definition}         
\newtheorem{definition} [theorem] {Definition}
\newtheorem*{claim} {Claim}
\newtheorem*{notation}  {Notation}

\newcommand{\alphabet}{\mathcal{A}}
\newcommand{\netI}{\mathcal{N}}
\newcommand{\netII}{\mathcal{N}'}

\newcolumntype{L}[1]{>{\raggedright\let\newline\\\arraybackslash\hspace{0pt}}m{#1}}
\newcolumntype{C}[1]{>{\centering\let\newline\\\arraybackslash\hspace{0pt}}m{#1}}
\newcolumntype{R}[1]{>{\raggedleft\let\newline\\\arraybackslash\hspace{0pt}}m{#1}}

\newcommand{\Z}{\mathbf{Z}}
\newcommand{\Q}{\mathbf{Q}}
\newcommand{\N}{\mathbf{N}}
\newcommand{\R}{\mathbf{R}}
\newcommand{\A}{\mathcal{A}}
\newcommand{\I}{\mathcal{I}}
\newcommand{\V}[2]{\mathrm{V}(#1,#2)}
\newcommand{\F}{\mathbb{F}}
\newcommand{\Capacity}{\mathcal{C}}

\newcommand{\B}{B}

\newcommand{\yields}{\longrightarrow}

\newcommand{\dub}[2]{\left\langle #1, #2 \right\rangle}
\newcommand{\GF}[1]{\mathrm{GF}\!\left(#1\right)}
\newcommand{\theranksymbol}{\mathsf{rank}}
\newcommand{\Rank}[1]{\mathsf{rank}\left(#1\right)}
\newcommand{\Char}[1]{\mathsf{char}\!\left(#1\right)}
\newcommand{\Comment}[1]{ \left[\mbox{from  #1} \right]}
\newcommand{\Div}[2]{\ensuremath{ #1  \! \bigm| \!  #2  } }
\newcommand{\NDiv}[2]{\ensuremath{ #1   \notdivides   #2  } }
\newcommand{\GCD}[2]{\ensuremath{ \mathsf{gcd} \!\left( #1  ,  #2 \right)  } }
\newcommand{\GGGCD}[3]{\ensuremath{ \mathsf{gcd} \!\left( #1  ,  #2 , #3 \right)  } }
\newcommand{\Order}[1]{\ensuremath{ \mathcal{O}\left( #1 \right)}}
\newcommand{\PrimeFact}[1]{\ensuremath{ p_1^{\gamma_1} \cdots p_{\omega(#1)}^{\gamma_{\omega(#1)}}}}

\newcommand{ \cc} [2]{c_{#1,#2}}
\newcommand{\ccc} [3]{c^{(#1)}_{#2,#3}}
\newcommand{\cccinv} [3]{\left(c^{(#1)}_{#2,#3}\right)^{-1}}
\renewcommand{\d} [1]{d_{#1}}
\newcommand{ \dd} [2]{d_{#1,#2}}
\newcommand{\ddd} [3]{d^{(#1)}_{#2,#3}}
\newcommand{\dddinv} [3]{\left(d^{(#1)}_{#2,#3}\right)^{-1}}
\newcommand{ \ee} [2]{e^{(#1)}_{#2}}
\newcommand{\xx} [2]{x^{(#1)}_{#2}}
\newcommand{\xxh} [2]{\hat{x}^{(#1)}_{#2}}
\newcommand{\RR} [2]{R^{(#1)}_{#2}}
\newcommand{\Ss} [2]{S^{(#1)}_{#2}}
\newcommand{  \pipi} [2]{\pi^{(#1)}_{#2}}
\newcommand{  \pipiinv} [2]{\pi^{(#1)^{-1}}_{#2}}
\newcommand{\sigsig} [2]{\sigma^{(#1)}_{#2}}
\newcommand{ \MM} [2]{M_{#1,#2}}
\newcommand{\MMM} [3]{M^{(#1)}_{#2,#3}}
\newcommand{ \DD} [2]{D_{#1,#2}}
\newcommand{\DDD} [3]{D^{(#1)}_{#2,#3}}
\newcommand{\QQ}  [2]{Q_{#1,#2}}
\newcommand{\QQQ} [3]{Q^{(#1)}_{#2,#3}}
\newcommand{\n}{w}
\newcommand{\m}{m}
\newcommand{\Radds} [1] {\left( #1 \right)_R}

\def\?#1{}
\newcommand{\niceRModule}{\text{standard $R$-module}}
\newcommand{\niceRModules}{\text{standard $R$-modules}}
\newcommand{\niceMkFModule}{\text{standard $M_k(\F)$-module}}

\newcommand{\mybullet}{\ensuremath{\cdot}  }

\newcommand{\Stack} [2] {\! \! \! \begin{array}{l} #1 \\ #2 \end{array}}

\newcommand{\Mod}[1] {\left(\text{mod } #1 \right)}
\newcommand{\act}{\cdot}
\newcommand{\dsum}{\displaystyle\sum}
\newcommand{\vcomp} [2]{{\left[ #1 \right]}_{#2}}

\newcommand{\RepAdd}[3]{\ensuremath{  \underbrace{#1 #2 \cdots #2 #1}_{#3 \text{ adds}} }}
\makeatletter
\newcommand{\pushright}[1]{\ifmeasuring@#1\else\omit\hfill$\displaystyle#1$\fi\ignorespaces}
\makeatother

\renewcommand{\emptyset}{\varnothing} 
\renewcommand{\subset}{\subseteq}     
\newcommand{\Network}{\mathcal{N}}
\newcommand{\TBA}{*** To Be Added ***}

\newcommand{\LHS}{\mathrm{LHS}}
\newcommand{\RHS}{\mathrm{RHS}}

\newcommand{\Fig}[2]{
 \medskip
  \epsfysize=#2 
  \epsffile{#1.eps}
 \medskip
}

\let\bbordermatrix\bordermatrix
\patchcmd{\bbordermatrix}{8.75}{4.75}{}{}
\patchcmd{\bbordermatrix}{\left(}{\left[}{}{}
\patchcmd{\bbordermatrix}{\right)}{\right]}{}{}

\setcounter{page}{0}

\title{A Class of Non-Linearly Solvable Networks
\thanks{This work was supported by the 
National Science Foundation.\newline
\indent \textbf{J. Connelly and K. Zeger} are with the 
Department of Electrical and Computer Engineering, 
University of California, San Diego, 
La Jolla, CA 92093-0407 
(j2connelly@ucsd.edu and zeger@ucsd.edu).
}}

\author{Joseph Connelly and Kenneth Zeger\\}

\date{
\textit{
IEEE Transactions on Information Theory\\
Submitted: \submitteddate\\
}}

\maketitle
\begin{abstract}
For each integer $\m \geq 2,$
a network is constructed which 
is solvable over an alphabet of size $\m$ but
is not solvable over any smaller alphabets.
If $\m$ is composite,
then the network has no vector linear solution
over any $R$-module alphabet
and is not asymptotically linear solvable over any finite-field alphabet.
The network's capacity is shown to equal one, and
when $\m$ is composite, its linear capacity
is shown to be bounded away from one for all finite-field alphabets.
\end{abstract}

\thispagestyle{empty}

\clearpage

\centerline{*** Table Of Contents Provided During Manuscript Review Only ***}
\tableofcontents

\clearpage
\section{Introduction} \label{sec:intro}

A \textit{network} will refer to a finite, directed, acyclic multigraph,
some of whose nodes are \textit{sources} or \textit{receivers}.
Source nodes generate $k$-dimensional vectors of \textit{messages}, 
where each of the $k$ messages is an arbitrary element 
of a fixed, finite set of size at least $2$,
called an \textit{alphabet}.
The elements of an alphabet are called \textit{symbols}.
The \textit{inputs} to a node are the messages, if any, originating at the node
and the symbols on the incoming edges of the node.
Each outgoing edge of a network node
carries a vector of $n$ alphabet symbols, called \textit{edge symbols}.
If a node has at most $n$ input symbols,
then we will assume, without loss of generality, that each of its out-edges carries all $n$ of such symbols.
Each outgoing edge of a node has associated with it an \textit{edge function}
which maps the node's inputs
to the output vector carried by the edge.
Each receiver node has \textit{demands},
which are $k$-dimensional message vectors the receiver wishes to obtain.
Each receiver also has \textit{decoding functions}
which map the receiver's inputs
to $k$-dimensional vectors of alphabet symbols in an attempt to satisfy the receiver's demands.

A \textit{$(k,n)$ fractional code over an alphabet $\A$}
(or, more briefly, a \textit{$(k,n)$ code over $\A$})
is an assignment of edge functions to all of the edges in a network and
an assignment of decoding functions to all of the receiver nodes in the network.

A \textit{$(k,n)$ solution over $\A$} is a $(k,n)$ code over $\A$ such that
each receiver's decoding functions can recover all $k$ components 
of each of its demands
from its inputs.

An edge function 
$$f:   \underbrace{\A^k \times \dots \times \A^k}_{i\ \text{message vectors}} 
\times \underbrace{\A^n \times \dots \times \A^n}_{j\ \text{in-edges}} 
    \longrightarrow \A^n$$
is \textit{linear over $\A$} if it can be written in the form
\begin{align}
f(x_1, \dots, x_i, y_1, \dots, y_j) = M_1 x_1 + \dots + M_i x_i + M'_1 y_1 + \dots + M'_j y_j \label{eq:0}
\end{align}
where 
$M_1, \dots, M_i$ are $n\times k$ matrices 
and $M'_1, \dots, M'_j$ are $n\times n$ matrices
whose entries are constant values.
Similarly, a decoding function is linear
if it has a form analogous to \eqref{eq:0}.
A $(k,n)$ code is said to be \textit{linear over $\A$}
if each edge function and each decoding function is linear over $\A$.
We will focus attention on linear codes in a very general setting where the alphabets are $R$-modules
(discussed in in Section~\ref{ssec:def}).
If the network alphabet is an $R$-module,
then, in \eqref{eq:0}, 
$\A$ is an Abelian group,
the elements of the matrices are from the ring $R$,
and multiplication of ring elements by elements of $\A$ is the action of the module.
Special cases of linear codes over $R$-modules include
linear codes over groups, rings, and fields.
A network is defined to be
\begin{itemize}[noitemsep,topsep=0pt]
  \item[--] \textit{solvable over $\A$}
if there exists 
a $(1,1)$ solution over $\A$,
  \item[--] \textit{scalar linear solvable over $\A$}
if there exists a $(1,1)$ linear solution over $\A$,
  \item[--] \textit{vector linear solvable over $\A$}
if there exists a $(k,k)$ linear solution over $\A$, for some $k \geq 1$,
  \item[--] \textit{asymptotically linear solvable over $\A$}
if for any $\epsilon > 0$, 
there exists a $(k,n)$ linear solution over $\A$ 
for some $k$ and $n$ satisfying $k/n > 1 -\epsilon$.
\end{itemize}
We say that a network is 
\textit{solvable}, 
(respectively, \textit{vector linear solvable} or \textit{scalar linear solvable})
if it is 
solvable
(respectively, vector linear solvable or scalar linear solvable) 
over some alphabet.

  The \textit{capacity}\footnote{In the literature,
  this is sometimes referred to as the ``coding capacity''
  (as opposed to the routing capacity).
  For brevity, we will simply use the term ``capacity,''
  as we do not discuss routing capacity in this paper.}
  of a network is:
  $$\text{sup}\{k/n \, : \, \exists \text{ a $(k,n)$ solution over some $\A$}\}.$$
  The \textit{linear capacity} of a network
  with respect to an alphabet $\A$ is:
  $$ \text{sup}\{k/n \, : \, \exists \text{ a $(k,n)$ linear solution over $\A$}\} .$$
It was shown in \cite{Cannons-Routing}
that the capacity of a network is independent 
of alphabet size,
and it was noted that linear capacity
can depend on alphabet size.

\subsection{Previous work} \label{ssec:previous}
One decade ago,
it was demonstrated in \cite{Dougherty-Freiling-Zeger04-Insufficiency}
that there can exist a network which is solvable,
but not vector linear solvable over any finite-field alphabet and any vector dimension.
To date, 
the network given in \cite{Dougherty-Freiling-Zeger04-Insufficiency}
is the only known example of such a network published in the literature.
In fact, the network given in \cite{Dougherty-Freiling-Zeger04-Insufficiency}
was shown to not be vector linear solvable over very general algebraic types of alphabets,
such as
finite rings and
modules, 
and was shown not to even be asymptotically linear solvable over finite-field alphabets,
and, as a result,
the network has been described as ``diabolical'' 
by 
Kschischang \cite{Kschischang-chapter}%
\footnote{The terminology was apparently attributed by F. Kschischang to M. Sudan.}
and
Koetter\cite{Koetter-WiOpt08}.

The diabolical network has been utilized 
in numerous extensions and applications of network coding,
such as
by Krishnan and Rajan \cite{Krishnan-ErrorCorrecting}
for network error correction,
and by Rai and Dey \cite{Rai-Dey-Sum}
for multicasting the sum of messages 
to construct networks with equivalent solvability properties
hence showing that linear codes 
are insufficient for each problem.
El Rouayheb, Sprintson, and Georghiades \cite{Rouayheb-Index}
reduced the index coding problem
to a network coding problem,
thereby using the diabolical network 
to show that linear index codes
are not necessarily sufficient.
Blasiak, Kleinberg, and Lubetzky \cite{Blasiak-Lexicopgraphic} used index codes
to create networks where there is a polynomial separation
between linear and non-linear network coding rates.
Chan and Grant \cite{ChanGrant-EntropyFunctions}
showed a duality between entropy functions
and network coding problems,
which allowed for an alternative proof of 
the insufficiency of linear network codes.


We now summarize some of the existing results 
regarding the solvability and linear solvability
of \textit{multicast networks} (in which each receiver demands all of the messages) 
and \textit{general networks} (in which each receiver demands a subset of the messages).
Network codes were first presented by Ahlswede, Ning, Li, and Yeung \cite{Ahlswede-NIF}
as a method of improving the throughput of a network;
they presented the butterfly network,
a variant of which is scalar linear solvable but
not solvable via routing.
Li, Young, and Cai \cite{Li-Linear} showed that if a multicast network is solvable,
then it is scalar linear solvable over all sufficiently large finite-field alphabets.
In addition, Riis \cite{Riis-LinearVsNonlinear} showed that
every solvable multicast network 
has a binary linear solution in some vector dimension.
Feder, Ron, and Tavory \cite{Feder-Bounds} 
and Rasala Lehman and Lehman \cite{Lehman-Complexity}
both independently showed that some solvable multicast networks
asymptotically require finite-field alphabets to be at least
as large as twice the square root of the number of receiver nodes.

Non-linear coding in multicast networks can offer advantages
such as reducing the alphabet size required for solvability;
Rasala Lehman and Lehman \cite{Lehman-Complexity} presented a network which is solvable 
over a ternary alphabet but has no scalar linear solution
over any alphabet whose size is less than five,
and Riis \cite{Riis-LinearVsNonlinear}
and also \cite{DFZ-LinearitySolvability} 
demonstrated general and multicast networks, respectively,
which have scalar non-linear binary solutions but no scalar linear binary solutions.
A multicast network was presented in \cite{DFZ-LinearitySolvability} which is solvable precisely over those alphabets
whose size is neither $2$ nor $6$,
and Sun, Yin, Li, and Long \cite{Sun-FieldSize}
presented families of multicast networks 
which are scalar linear solvable over certain finite-field alphabets
but not over all larger finite-field alphabets.

Unlike multicast networks,
general networks that are solvable are not necessarily vector linear solvable,
as demonstrated in \cite{Dougherty-Freiling-Zeger04-Insufficiency}.
M\'{e}dard, Effros, Ho, and Karger \cite{Medard-NonMulticast}
showed that there can exist a network which is vector linear solvable
but not scalar linear solvable.
Shenvi and Dey \cite{Shenvi-TwoPair} 
showed that for networks with $2$ source-receiver pairs
the following are equivalent:
the network is solvable,
the network is vector linear solvable,
the network satisfies a simple cut condition.
Cai and Han \cite{Cai-ThreePair} 
showed that for a particular class of networks with $3$ source-receiver pairs:
the solvability can be determined in polynomial time,
being solvable is equivalent to being scalar linear solvable,
and finite-field alphabets of size $2$ or $3$ are sufficient
to construct scalar linear solutions.
In \cite{DFZ-Unachievability},
the Fano and non-Fano networks were shown to be
solvable precisely over even and odd alphabets, respectively.
For each integer $\m \geq 2,$
Rasala Lehman and Lehman \cite{Lehman-Complexity} 
demonstrated a class of networks which
are not solvable over any alphabet 
whose size is less than $\m$
and are solvable over all alphabets whose size 
is a prime power greater than or equal to $\m$.
For each integer $\m \geq 3$,
Chen and HaiBin\cite{ChenHaiBin-Characterization}
demonstrated a class of networks which
are not solvable over any alphabet 
whose size is less than $\m$
and are solvable over all alphabets whose size
is not divisible by $2,3,\dots, \m-1$.

Koetter and M\'{e}dard \cite{Koetter-Algebraic} showed
for every finite field $\F$ and every network,
the network is scalar linear solvable 
over $\F$ if and only if
a corresponding system of polynomials 
has a common root in $\F$,
and in \cite{DFZ-Polynomials} it was shown
that for every finite field $\F$ and any system of polynomials
there exists a corresponding network which is scalar linear solvable
over $\F$ if and only if 
the system of polynomials has a common root
in $\F$.
Subramanian and Thangaraj \cite{Subramanian-Path} showed an alternate method of deriving 
a system of polynomials which corresponds to the scalar linear solvability of a network,
such that the degree of each polynomial equation is at most $2$.
Presently,
there are no known algorithms 
for determining whether a general network is solvable.

While vector linear solvable networks are solvable networks,
the converse need not be true.
This paper demonstrates infinitely many such counterexamples.

There remain numerous open questions regarding the existence 
of solvable networks which are not vector linear solvable.
Are many/most solvable networks not vector/scalar linearly solvable?
Can such networks be efficiently characterized?
Can such networks be algorithmically recognized?
We leave these questions for future research.

\subsection{Our contributions}
In this paper, we present an infinite class of solvable networks
which are not linear solvable 
over any $R$-module alphabet and any vector dimension.
We denote each such network as $\Network_4$, and we construct $\Network_4$ from 
several intermediate networks denoted by
$\Network_0, \Network_1, \Network_2, \text{ and } \Network_3$,
all of which are constructed from a fundamental network building block $\B$.
Specifically,
for each positive composite number $\m$,
we describe how to construct a network $\Network_4$ which has a non-linear
solution over an alphabet of size $\m$,
yet has no vector linear solution over 
any vector dimension
and any
finite field, 
commutative ring with identity, 
or $R$-module alphabet.
In addition, such a network is not solvable over
any alphabet whose size is less than $\m$.
The diabolical network in \cite{Dougherty-Freiling-Zeger04-Insufficiency}
was shown to be non-linear solvable over an alphabet of size $4$.

We will now summarize the main results of this paper, which all appear in Section~\ref{sec:N4}.
The network $\Network_4$ is parameterized by an arbitrary integer $m \ge 2$.
Theorem~\ref{thm:N4_m} shows that $\Network_4$ is solvable over an alphabet of size $m$.
Theorem~\ref{thm:N4_solv} shows, however, that $\Network_4$ is never solvable over alphabets smaller than $m$.
Theorem~\ref{thm:N4_prime} shows that when $m$ is prime, $\Network_4$ has a scalar linear solution
over a field of size $m$.
In fact, for all non-prime integers $m$, the network $\Network_4$ has no linear solution,
as demonstrated by Theorems~\ref{thm:N4_R} and \ref{thm:N4_cap}.
In particular, Theorem~\ref{thm:N4_R} shows that when $m$ is composite,
no vector linear solution for $\Network_4$ exists over any $R$-module, and
Corollary~\ref{cor:N4_asymp} shows that in such case,
$\Network_4$ is not even asymptotically linear solvable over any finite-field alphabet.
In the special case of $\m = 4$, 
the demonstrated network $\Network_4$ exhibits properties 
similar to the network presented in \cite{Dougherty-Freiling-Zeger04-Insufficiency}.

The diabolical network was shown in \cite{Dougherty-Freiling-Zeger04-Insufficiency} 
to have capacity equal to one, whereas its linear capacity is bounded away from one for any 
finite-field alphabet.
Analogously, we show in Theorem~\ref{thm:N4_cap} that for all $\m$, the capacity of $\Network_4$
equals one, whereas for all composite $\m$, its linear capacity
over any finite-field alphabet is bounded away from one.
Related capacity results are given for the constituent networks
$\Network_0$ (in Lemma~\ref{lem:N0_cap}),
$\Network_1$ (in Lemma~\ref{lem:N1_cap}),
$\Network_2$ (in Lemma~\ref{lem:N2_cap}),
and
$\Network_3$ (in Lemma~\ref{lem:N3_cap}).

The rest of the paper is organized as follows.
Table~\ref{tab:T1} summarizes the networks created 
and the results in this paper.
Section~\ref{ssec:def} provides mathematical background and definitions. 
Sections~\ref{sec:N0}-\ref{sec:N3} present the building block networks
which are used to construct the main class of networks.
Section~\ref{sec:N4} details the properties and construction
of the main class of networks.
For each network family,
we will discuss the solvability properties,
the linear solvability properties,
and the capacity.
The Appendix contains the proofs of every lemma in this paper.
All other proofs are given in the main body of the paper.

Section~\ref{sec:Q} poses some open questions
regarding solvability of networks.

\clearpage
\begin{table}[h]
\small
 \begin{center}
  \begin{tabular}{|L{13.25cm}r|}
  \hline
   \textbf{Networks and Their Main Properties}
   & \textbf{Location} \\ \Xhline{3\arrayrulewidth}
\textbf{Network }$\Network_0(\m)$ 
    & Section~\ref{sec:N0} \\ \hline
    \mybullet Consists of a block $\B(\m)$ together with source nodes.
    & Figure~\ref{fig:N0}\\ 
    \mybullet $4\m+6$ nodes.
    & Remark~\ref{rem:N0_nodes} \\
    \mybullet If a $(1,1)$ code over $\A$ is a solution, then the code has an Abelian group structure. 
    & Lemma~\ref{lem:N0_P}  \\ \Xhline{2\arrayrulewidth}
%
\textbf{Network }$\Network_1(\m)$ 
     & Section~\ref{sec:N1} \\ \hline
    \mybullet Consists of a block $\B(\m)$ together with source nodes and 
    an additional receiver. & Figure~\ref{fig:N1}\\
    \mybullet $4\m + 7$ nodes.
    &  Remark~\ref{rem:N1_nodes}
    \\
    \mybullet If solvable over $\A$, then $\GCD{\vert\A\vert}{\m} = 1$.
    & Lemma~\ref{lem:N1_solv} \\ 
    \mybullet Scalar linear solvable over \niceRModule{} $G$ iff $\GCD{\Char{R} \!}{\m} = 1$.
    & Lemma~\ref{lem:N1_lin} \\ 
    \mybullet If asymptotically linear solvable over finite field $\F$, then $\NDiv{\Char{\F}}{\m}$.
    & Lemma~\ref{lem:N1_cap} 
    \\ \Xhline{2\arrayrulewidth}
\textbf{Network }$\Network_2(\m,\n)$
     & Section~\ref{sec:N2} \\ \hline
     \mybullet Consists of $\n$ blocks $\B(\m+1)$ together with source nodes and 
    &  \\ 
    \ \  an additional receiver. & Figure~\ref{fig:N2}\\
    \mybullet $4 \m \n  + 9 \n + 2$ nodes. 
    & Remark~\ref{rem:N2_nodes}
    \\
     \mybullet If $\n \geq 2,$ then non-linear solvable over an alphabet of size $\m\n$. 
    & Lemma~\ref{lem:N2_non} \\ 
    \mybullet If solvable over $\A$, then $\GCD{\vert \A \vert}{\m} \ne 1.$
    & Lemma~\ref{lem:N2_solv} \\ 
    \mybullet Scalar linear solvable over \niceRModule{} $G$ iff $\Div{\Char{R}}{\m}$.
    & Lemma~\ref{lem:N2_lin} \\ 
    \mybullet If asymptotically linear solvable over finite field $\F$, then $\Div{\Char{\F}}{\m}$.
    & Lemma~\ref{lem:N2_cap} \\  \Xhline{2\arrayrulewidth}
\textbf{Network }$\Network_3(\m_1,\m_2)$
    & Section~\ref{sec:N3} \\ \hline
    \mybullet Consists of blocks $\B(\m_1)$ and $\B(\m_2)$ together with source nodes and
    & \\ 
    \ \  an additional receiver. & Figure~\ref{fig:N3}\\
    \mybullet $4 \m_1 + 4 \m_2 + 12$ nodes. 
    & Remark~\ref{rem:N3_nodes} 
    \\
    \mybullet For each $s,t \geq 1$ relatively prime to $\m_1$,
    if $\m_2 = s \m_1^{\alpha}$ for some $\alpha >0$,
    & Corollary~\ref{cor:N3_non}  \\
    \ \ then non-linear solvable over an alphabet of size $t \m_1^{\alpha+1}$.&\\
    \mybullet If solvable over $\A$, then $\GCD{\vert\A\vert}{\m_1} = 1$ or $\NDiv{ \vert \A \vert}{ \m_2 }$.
    & Lemma~\ref{lem:N3_solv} \\ 
    \mybullet Scalar linear solvable over \niceRModule{} $G$ iff $\GGGCD{\Char{R} \!}{\m_1}{\m_2} = 1$.
    & Lemma~\ref{lem:N3_lin} \\ 
    \mybullet If asymptotically linear solvable over finite field $\F$,
      then $\Char{\F}$ is  &\\ 
    \ \ relatively prime to $\m_1$ or $\m_2$. 
    & Lemma~\ref{lem:N3_cap}\\ \Xhline{2\arrayrulewidth}
\textbf{Network }$\Network_4(\m)$
    & Section~\ref{sec:N4} \\ \hline
    \mybullet Consists of a disjoint union of various networks $\Network_1, \Network_2,$ and $\Network_3$.
    & Equation~\eqref{eq:N4_N4} \\ 
    \mybullet Solvable over an alphabet of size $\m$. 
    & Theorem~\ref{thm:N4_m} \\
    \mybullet If $\vert\A\vert < \m$, then not solvable over $\A$.
    & Theorem~\ref{thm:N4_solv} \\ 
    \mybullet If $\m$ is prime, then scalar linear solvable over $\GF{\m}$. 
    & Theorem~\ref{thm:N4_prime} \\ 
    \mybullet If $\m$ is composite, then: (1) not vector linear solvable over any $R$-module.
    & Theorem~\ref{thm:N4_R} \\ 
     \hspace*{3.92cm}  (2) not asymptotically linear solvable over any finite field. 
    & Corollary~\ref{cor:N4_asymp} \\ 
     \mybullet Number of nodes is $O\left( \m^{\frac{\log{\m}}{\log{\log{\m}}}} \right)$
    and $\Omega(m)$.  
    & Theorem~\ref{thm:N4_nodes}  \\ \hline
  \end{tabular}
 \end{center}
  \caption{Summary of the networks constructed in this paper,
  where $\m,\m_1,\m_2,$ and $\n$ are integers such that $\m,\m_1,\m_2 \geq 2$ and $\n \geq 1$.}
  \label{tab:T1}
\end{table}

\clearpage

\subsection{Preliminaries} \label{ssec:def}

The following definitions and results regarding linear network codes over $R$-modules
are from \cite{Dougherty-Freiling-Zeger04-Insufficiency}
and \cite{Dummit-Algebra}.

\begin{definition}
Let $(R,+,*)$ be a ring with additive identity $0_R$.
An \textit{$R$-module} (specifically a left $R$-module) is
an Abelian group $(G,\oplus)$ with identity $0_G$
and an action 
$$\act : R \times G \to G $$
such that
for all $r,s \in R$ and all $g,h \in G$ the following hold:
\begin{align*}
  r \act (g\oplus h) &= (r \act g) \oplus (r \act h) \\
  (r + s) \act g &= (r \act g) \oplus (s \act g) \\
  (r*s) \act g &= r \act (s \act g) \\
  0_R \act g &= 0_G .
\end{align*}
The ring multiplication symbol $*$ will generally be omitted for brevity.
If the ring $R$ has a multiplicative identity $1_R$,
then we also require $1_R \act g= g$ for all $g \in G$.
For brevity, we say that $G$ is an $R$-module.
$\ominus$ will denote adding the inverse of an element (subtraction) within the group.
\label{def:mod}
\end{definition}

The following definition describes a class of $R$-modules
which we will use to discuss linear solvability in this paper.
%

\begin{definition}
  Let $G$ be an $R$-module.
  We will say that $G$
  is a \textit{\niceRModule{}} if
  \begin{enumerate}
    \itemsep0.05em 
    \item $R$ acts faithfully on $G$;
  that is if $r,s \in R$ are such that $r\act g = s\act g$
  for all $g \in G$, then $r = s$.
    \item $R$ has a multiplicative identity $1_R$.
    \item $R$ is finite.
  
    \item If $r \in R$ has a multiplicative left (respectively, right) inverse,
  then it has a two-sided inverse,
  which will be denoted $r^{-1}$.
  \end{enumerate}
  \label{def:nice_mod}
\end{definition}

This enables us to characterize over which \niceRModules{}
the networks in this paper are scalar linear solvable.
Lemmas~\ref{lem:mod_2} and \ref{lem:mod_0}
show that if a network is not scalar linear solvable over any \niceRModule{},
then the network is not vector linear solvable over any $R$-module.

A finite ring $R$, with a multiplicative identity,
acting on itself is a \niceRModule.
For any finite field $\F$ and positive integer $k$,
the set $M_k(\F)$ of $k \times k$ matrices over $\F$
with matrix addition and multiplication 
is a ring and $\F^k$ is a \niceMkFModule. 

\begin{lemma}
  If a network $\Network$ is not scalar linear solvable over any \niceRModule{},
  then it is not scalar linear solvable over any $R$-module.
  \label{lem:mod_2}
\end{lemma}

\newcommand{\ProofOflemModTwo}{
\begin{proof}[Proof of Lemma \ref{lem:mod_2}]
  This follows from the proof of
  \cite[Theorem III.4]{Dougherty-Freiling-Zeger04-Insufficiency}.
\end{proof}
}

\begin{lemma}
  If a network is not scalar linear solvable over any $R$-module,
  then it is not vector linear solvable over any $R$-module.
  \label{lem:mod_0}
\end{lemma}

\newcommand{\ProofOflemModZero}{
\begin{proof}[Proof of Lemma \ref{lem:mod_0}]
If $R$ is a ring and $G$ is an $R$-module, then
the set $M_k(R)$ of $k \times k$ matrices over $R$
with matrix addition and multiplication defined in the usual way,
is a ring and $G^k$ is an $M_k(R)$-module.
So any vector linear solution over an $R$-module
is also a scalar linear solution over some other $R$-module.
Thus if no scalar linear solutions exist,
no vector linear solutions exist.
\end{proof}
}


Vector linear solutions over rings
are special cases of
vector linear solutions over $R$-modules
where $R$ acts on itself.
A field is a special case of a commutative ring with identity
where all elements have multiplicative inverses,
and scalar linear solutions are special cases of vector linear solutions
where $k = 1$.
Thus if a network is not vector linear solvable over $R$-modules,
it is also not vector (or scalar) linear solvable over rings with identity (or fields). 

For any ring $R$ with multiplicative identity,
the \textit{characteristic of $R$}
is denoted $\Char{R}$ and
is the smallest positive integer $\m$ such that 
$1_R$ added to itself $\m$ times equals $0_R$.
The characteristic of a finite field
is always a prime number.
We say that a positive integer $\m$
is \textit{invertible in} $R$
if there exists $\m^{-1} \in R$
such that 
$\m^{-1} \, (\m 1_R) = 1_R$,
where $(\m 1_R)$ denotes $1_R$ added to itself $\m$ times.
Specifically, 
$$\m^{-1} = \left(\RepAdd{1_R}{+}{\m} \right)^{-1}.$$

The following lemmas discuss properties of multiplicative inverses in rings
and will be used to more easily characterize the classes of $R$-modules 
over which $\Network_1$ and $\Network_3$ are scalar linear solvable.
\begin{lemma}
  For each finite ring $R$ with a multiplicative identity 
  and each positive integer $\m$,
  the integer $\m$ is invertible in $R$ 
  if and only if 
  there does not exist $s \in R\backslash\{0_R\}$ such that
  $ \m s = 0_R$.
  \label{lem:RingInv1}
\end{lemma}

\newcommand{\ProofOflemRingInvOne}{
\begin{proof}[Proof of Lemma \ref{lem:RingInv1}]
Assume $\m$ is invertible in $R$.
Then for all $s \in R$
such that $\m s = 0_R$,
if we multiply both sides of the equation by $\m^{-1}$,
we have $s = 0_R$.

To prove the converse,
assume $ \m s = 0_R$ only if $s = 0_R$.
Let $T = \{ \m s \; : \; s \in R\}$.
For each $s,s' \in R$,
we have $\m s = \m s'$ if and only if $\m (s - s') = 0_R$,
which implies $s = s'$,
so, by assumption, $\vert T \vert = \vert R \vert$.
Thus $1_R \in T$,
which implies $\m$ is invertible.
\end{proof}
}

\begin{lemma}
  For each finite ring $R$ with a multiplicative identity
  and each positive integer $\m$,
  the integer $\m$ is invertible in $R$
  if and only if 
  $\Char{R}$ and $\m$ are relatively prime.
  \label{lem:RingInv2}
\end{lemma}

\newcommand{\ProofOflemRingInvTwo}{
\begin{proof}[Proof of Lemma \ref{lem:RingInv2}]
Assume $\Char{R}$ and $\m$ are not relatively prime,
so they share a common factor $a > 1$.
Let $c$ and $\m'$ be integers such that 
$\Char{R} = a c $ and
$\m = a \m'$.
Then we have
$$ 0_R = \Char{R} \, 1_R = \m'  \, \Char{R} \, 1_R
  =  \m' \, a \, c \, 1_R
  = \m \, c  \, 1_R
  = \m \, \left( \RepAdd{1_R}{+}{c} \right). $$
Since $a > 1,$ 
we have $\RepAdd{1_R}{+}{c} \ne 0_R$,
so by Lemma~\ref{lem:RingInv1},
$\m$ is not invertible in $R$.

Conversely, 
assume $\m$ is not invertible in $R$.
Then by Lemma~\ref{lem:RingInv1},
there exists $s \in R\backslash\{0_R\}$
such that 
$$ 0_R = \m \, s = \RepAdd{s}{+}{\m}  $$
which implies the additive order of $s$ divides $\m$.
We also have
$$\RepAdd{s}{+}{\Char{R}} = \Char{R} \, s = 0_R, $$
which implies the additive order of $s$ divides $\Char{R}$.
Since $s \ne 0_R$,
the additive order of $s$ is greater than $1$,
and the additive order of $s$ divides both $\m$ and $\Char{R}$,
so they are not relatively prime.
\end{proof}
}

The following definition is called
Property $P'$
in \cite{ChenHaiBin-Characterization},
and will be utilized throughout.
\begin{definition}
  Let $\m \geq 2$. 
  A $(1,1)$ code for a network $\Network$ over an alphabet $\A$,
  containing messages $x_0,x_1,\dots,x_{\m}$
  and edge symbols $e_0,e_1,\dots,e_{\m}, \, e,$
  is said to have \textit{Property $P(\m)$} if
  there exists a binary operation
  $\oplus: \, \A \times \A \to \A$
  and 
  permutations $\pi_0,\pi_1,\dots,\pi_{\m}$
  and $\sigma_0, \sigma_1, \dots,\sigma_{\m}$ of $\A$,
  such that $(\A,\oplus)$ is an Abelian group
  and
  the edge symbols can be written as
  \begin{align*}
    e_i & = \sigma_i \left( \bigoplus_{ \substack{j = 0 \\ j \ne i} }^{\m} \pi_j(x_j)  \right) 
        && (i = 0,1,\dots,\m) \\
    e &= \bigoplus_{j = 0}^{\m} \pi_j(x_j) .
  \end{align*}
  \label{def:P}
\end{definition}

\clearpage

\section{The network \texorpdfstring{$\Network_0(\m)$}{N0(m)}} \label{sec:N0}
\psfrag{ex}{$e$}
\psfrag{e0}{$e_{0}$}
\psfrag{e1}{$e_{1}$}
\psfrag{en}{$e_{\m}$}
\psfrag{u0}{$u_{0}$}
\psfrag{u1}{$u_{1}$}
\psfrag{un}{$u_{\m}$}
\psfrag{ux}{$u$}

\psfrag{v0}{$v_{0}$}
\psfrag{v1}{$v_{1}$}
\psfrag{vn}{$v_{\m}$}
\psfrag{vx}{$v$}
\psfrag{R0}{$R_{0}$}
\psfrag{R1}{$R_{1}$}
\psfrag{Rn}{$R_{\m}$}
\psfrag{x0}{$x_0$}
\psfrag{x1}{$x_1$}
\psfrag{xn}{$x_{\m}$}
\psfrag{y0}{$y_0$}
\psfrag{y1}{$y_1$}
\psfrag{yn}{$y_{\m}$}

\psfrag{N0(n)'}{\large$\B(\m)$}

\psfrag{S0'}{\small$S_0$}
\psfrag{S1'}{\small$S_1$}
\psfrag{Sn'}{\small$S_{\m}$}
\psfrag{e0'}{\small$e_{0}$}
\psfrag{e1'}{\small$e_{1}$}
\psfrag{en'}{\small$e_{\m}$}
\psfrag{x0'}{\small$x_0$}
\psfrag{x1'}{\small$x_1$}
\psfrag{xn'}{\small$x_{\m}$}
\begin{figure}[h]
  \begin{center}
    \leavevmode
    \hbox{\epsfxsize=.6\textwidth\epsffile{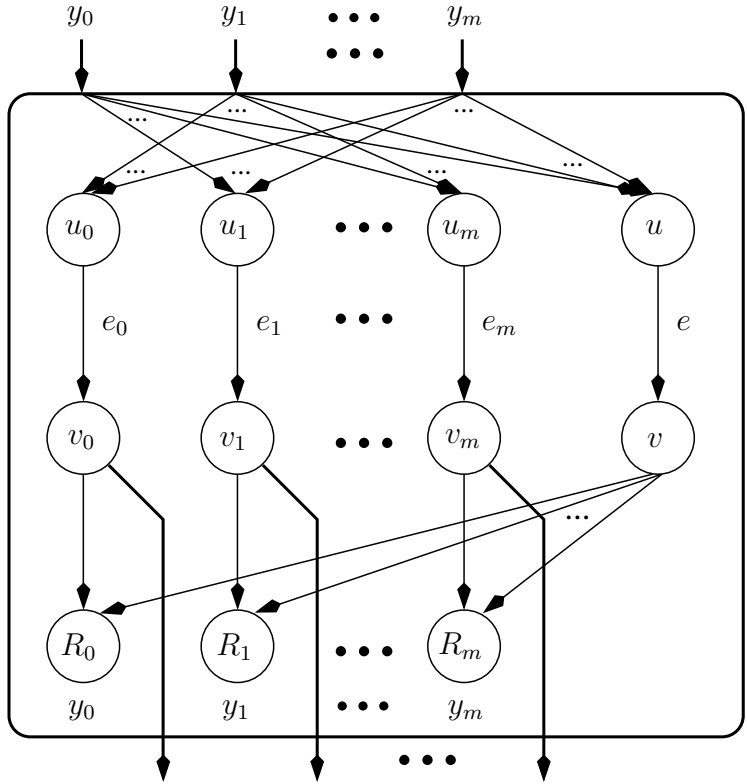}}
  \end{center}
  \caption{Network building block $\B(\m)$ has message inputs $y_0,y_1,\dots,y_{\m}$ 
      (from unspecified source nodes) and $\m+1$ output edges.
      For each $i$, the node $u_i$ receives each of the inputs except $y_i$
      and has a single outgoing edge to the node $v_i$, 
      which carries the edge symbol $e_i$.
      The node $u$ receives each of the inputs
      and has a single outgoing edge to the node $v$, 
      which carries the edge symbol $e$.
      For each $i$, the receiver node $R_i$ has an incoming edge from $v_i$ and an incoming edge from $v$
      and demands the $i$th message $y_i$.
      The $i$th output edge of $\B(\m)$ is an outgoing edge of node $v_i$.
  }
\label{fig:B}
\end{figure}

\begin{figure}[h]
  \begin{center}
    \leavevmode
    \hbox{\epsfxsize=.25\textwidth\epsffile{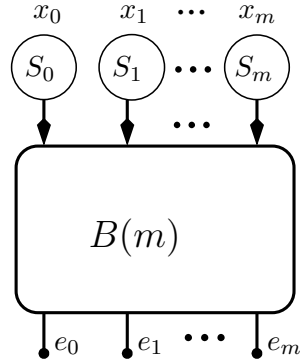}}
  \end{center}
  \caption{Network $\Network_0(\m)$ consists of a block $\B(\m)$
      together with source nodes $S_0,S_1,\dots,S_{\m}$,
      which generate messages $x_0,x_1,\dots,x_{\m}$, respectively.
      The output edges of $\B(\m)$ are unused.
  }
\label{fig:N0}
\end{figure}


For each $\m \ge 2,$ 
the network building block $\B(\m)$ is defined in Figure~\ref{fig:B}
and is used to build network $\Network_0(\m),$ 
which is defined in Figure~\ref{fig:N0}.
For each $i$, the node $v_i$ within $\B(\m)$ 
has a single incoming edge from node $u_i$,
so without loss of generality, we may assume both outgoing edges of $v_i$
carry the symbol $e_i$.  
Similarly, we may assume each of the outgoing edges of the node $v$
carries the symbol $e$.
Lemma~\ref{lem:N0_P} demonstrates that
for each $\m \geq 2$,
the $(1,1)$ solutions of network $\Network_0(\m)$
are precisely those codes which satisfy Property $P(\m)$, 
defined in Definition~\ref{def:P}.
In particular, the solution alphabets have to be permutations of Abelian groups.

\begin{remark}
  Network $\Network_0(\m)$ has $\m+1$ source nodes,
  $2 (\m + 2)$ intermediate nodes,
  and $\m+1$ receiver nodes,
  so the total number of nodes in $\Network_0(\m)$ is $4\m + 6$.
  \label{rem:N0_nodes}
\end{remark}

Lemma~\ref{lem:N0_P} characterizes the solvability of $\Network_0(\m)$
and will be used in the proofs of the solvability conditions of $\Network_1, \Network_2,$ and $\Network_3$.
\begin{lemma}
  Let $\m \geq 2$.
  A $(1,1)$ code over an alphabet $\A$
  is a scalar solution for
  network $\Network_0(\m)$
  if and only if 
  the code satisfies Property $P(\m)$.
  \label{lem:N0_P}
\end{lemma}

\newcommand{\ProofOflemNZeroP}{
\begin{proof}[Proof of Lemma \ref{lem:N0_P}]
  This lemma follows directly from 
  \cite[Proposition~3.2]{ChenHaiBin-Characterization}.
\end{proof}
}

The following result regarding the scalar linear solvability of $\Network_0(\m)$ 
will be used in later proofs.

\begin{lemma}
  Let $\m \geq 2$ and let $G$ be a \niceRModule.
  Suppose a scalar linear solution for network $\Network_0(\m)$
  over $G$ has edge symbols
\begin{align*}
e_i &= \bigoplus_{\substack{j= 0\\ j\ne i}}^{\m} \left( \cc{i}{j} \act x_j  \right)
  & &   (i = 0, 1, \dots, \m) \\
e &= \bigoplus_{j = 0}^{\m} \left( \cc{\?}{j} \act x_j \right)
\end{align*}
and decoding functions
\begin{align*}
R_i : \ \ 
  x_i & = \left( \dd{i}{e} \act e \right) \oplus  \left( \dd{\?}{i} \act e_i  \right)
    & &  (i = 0,1, \dots,\m)
\end{align*}    
where $\cc{i}{j}, \cc{\?}{j}, \dd{i}{e}, \dd{\?}{i} \in R$.
Then each $\dd{\?}{i}$ and $\cc{\?}{i}$ is invertible in $R$, 
and
\begin{align*}
  \cc{i}{j} & = - \dd{\?}{i}^{-1} \, \dd{i}{e} \, \cc{\?}{j}
    & & (i,j = 0,1,\dots,\m \text{ and } j \ne i) .
\end{align*}
  \label{lem:N0_lin}
\end{lemma}

\newcommand{\ProofOflemNZeroLin}{
\begin{proof}[Proof of Lemma \ref{lem:N0_lin}]
Equating message components at $R_{i}$ yields
\begin{align*}
    1_R &=  \dd{i}{e} \, \cc{\?}{i} 
    & & (i = 0,1,\dots, \m)\\
    0_R &=  \dd{i}{e} \, \cc{\?}{j} + \dd{\?}{i} \, \cc{i}{j}
    & & (i,j = 0,1,\dots, \m \text{ and } j \ne i)
\end{align*}
which implies the following elements of $R$ are invertible:
\begin{align*}
  & \dd{i}{e} \text{ and } \cc{\?}{i} 
    & & (i = 0,1,\dots, \m)\\
  & \dd{\?}{i} \text{ and } \cc{i}{j} 
    & & (i,j = 0,1,\dots, \m \text{ and } j \ne i) .
\end{align*}
The result then follows by solving for $\cc{i}{j}$.
\end{proof}
}

\begin{lemma}
  The network $\Network_0(\m)$ has capacity and linear capacity, for any finite-field alphabet, equal to $1$.
  \label{lem:N0_cap}
\end{lemma}

\newcommand{\ProofOflemNZeroCap}{
\begin{proof}[Proof of Lemma \ref{lem:N0_cap}]
  Let $G$ be a \niceRModule.
  The network $\Network_0(\m)$ has the following scalar linear solution over $G$:
   \begin{align*}
     e_i & = \bigoplus_{\substack{ j=0 \\ j\ne i}}^{\m} x_j
      & & (i = 0,1,\dots,\m) 
      \\
      e & = \bigoplus_{j=0}^{\m} x_j
   \end{align*}
  and decoding at each receiver as follows:
  \begin{align*}
    R_i: \ \ e \ominus e_i &= x_i & & (i = 0,1,\dots,\m) .
  \end{align*}

  A scalar linear solution over a finite-field alphabet is a special case of a scalar linear solution over a \niceRModule.
  Therefore $\Network_0(\m)$ is scalar linear solvable over any finite-field alphabet,
  so the linear capacity of $\Network_0(\m)$ for any finite-field alphabet is at least $1$.
  The only path for message $x_0$ to reach the receiver $R_0$
  is through the edge connecting nodes $u$ and $v$,
  so its capacity is at most $1$.
  Thus, both the capacity of $\Network_0(\m)$ and its linear capacity
  for any finite-field alphabet are equal to $1$.
\end{proof}
}

\clearpage

\section{The network \texorpdfstring{$\Network_1(\m)$}{N1(m)}} \label{sec:N1}
\psfrag{N0(n)}{\large $\B(\m)$}
\psfrag{Rx}{\small$R_x$}

\psfrag{S0}{$S_0$}
\psfrag{S1}{$S_1$}
\psfrag{Sn}{$S_{\m}$}
\psfrag{x0}{\large$x_0$}
\psfrag{x1}{\large$x_1$}
\psfrag{xn}{\large$x_{\m}$}
\psfrag{e0}{\large$e_{0}$}
\psfrag{e1}{\large$e_{1}$}
\psfrag{en}{\large$e_{\m}$}
\begin{figure}[h]
  \begin{center}
    \leavevmode
    \hbox{\epsfxsize=.3\textwidth\epsffile{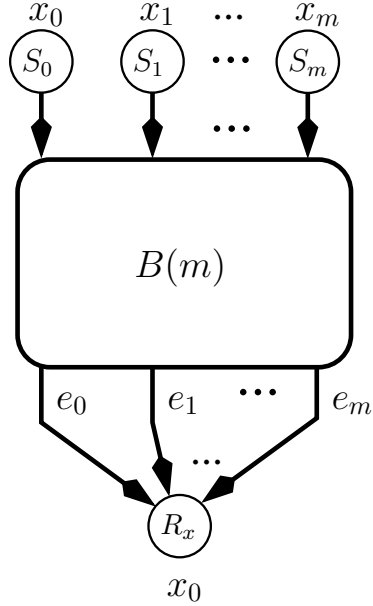}}
  \end{center}
  \caption{The network $\Network_1(\m)$
    is constructed from a $\B(\m)$ block
    together with
    source nodes $S_0,S_1,\dots,S_{\m}$
    and an additional receiver $R_x$.
    For each $i$, the source node $S_i$ generates the message $x_i$
    and is the $i$th input to $\B(\m)$.
    The additional receiver $R_x$ receives all of the output edges of $\B(\m)$
    and demands the message $x_0$.
  }
\label{fig:N1}
\end{figure}

For each $\m \geq 2,$ network $\Network_1(\m)$ is
defined in Figure~\ref{fig:N1}.
The special case $\m = 2$ corresponds to the non-Fano network 
from \cite{DFZ-Matroids}, \cite{DFZ-Unachievability}, 
with a relabeling of messages and nodes.
Lemmas~\ref{lem:N1_solv}, \ref{lem:N1_lin}, and \ref{lem:N1_cap}, respectively,
demonstrate that network $\Network_1(\m)$ is
\begin{enumerate}
\itemsep0em 
  \item solvable over alphabet $\A$ only if $|\A|$ is relatively prime to $\m$,

  \item scalar linear solvable over \niceRModule{} $G$ if and only if $\Char{R}$ is relatively prime to $\m$,

  \item asymptotically linear solvable over finite field $\F$ if and only if $\Char{\F}$ does not divide $\m$.
\end{enumerate}

\begin{remark}
Network $\Network_1(\m)$ is a network $\Network_0(\m)$ with one additional receiver node,
so the total number of nodes in $\Network_1(\m)$ is 
$4 \m +7$.
\label{rem:N1_nodes}
\end{remark}

\subsection{Solvability conditions of \texorpdfstring{$\Network_1(\m)$}{N1(m)}} \label{ssec:N1_solv}

The following lemma also follows from \cite[Proposition~4.1]{ChenHaiBin-Characterization}
and characterizes a condition on the alphabet size
necessary for the solvability of $\Network_1(\m)$.
\begin{lemma}
  For each $\m \geq 2,$ 
  if network $\Network_1(\m)$ 
  is solvable over alphabet $\A$,
  then $\m$ and $\vert \A \vert$ are relatively prime.
\label{lem:N1_solv}
\end{lemma}

\newcommand{\ProofOflemNOneSolv}{
\begin{proof}[Proof of Lemma \ref{lem:N1_solv}]
Assume $\Network_1(\m)$ is solvable over $\A$.
Network $\Network_1(\m)$ consists of a network $\Network_0(\m)$
with the additional receiver $R_x$,
so by Lemma~\ref{lem:N0_P}, 
the edge functions within $\B(\m)$ must satisfy Property $P(\m)$.
Thus, there exists an Abelian group $(\A,\oplus)$
and permutations
$\pi_{0},\pi_{1},\dots,\pi_{\m}$
and $\sigma_{0},\sigma_{1},\dots,\sigma_{\m}$ of $\A$,
such that the edges carry the symbols:
\begin{align}
  e_{i} & = \sigma_{i} \left( \bigoplus_{\substack{j = 0 \\ j \ne i}}^{\m} \pi_{j}(x_j)  \right) 
    && (i = 0,1,\dots,\m) \label{eq:N1_solv_1}  \\
  e &= \bigoplus_{j = 0}^{\m} \pi_{j}(x_j)
    \notag .
\end{align}

Now suppose to the contrary that 
$\m$ and $\vert\A\vert$ share a prime factor $p$.
By Cauchy's Theorem of Finite Groups \cite[p. 93]{Dummit-Algebra},
there exists a nonzero element $a$ in the group $\A$
whose order is $p$.
Since $\Div{p}{\m}$,
we have $\RepAdd{a}{\oplus}{\m}= 0$.

Define two collections of messages as follows:
\begin{align*}
   x_j  &= \pi_{j}^{-1}(0) & & (j = 0,1,\dots, \m)\\
   \hat{x}_j &= \pi_{j}^{-1}(a) & & (j = 0,1,\dots, \m) .
\end{align*}
Since $a \ne 0$ and each $\pi_j$ is bijective, it follows that $x_j \ne \hat{x}_j$ for all $j$.
By Property $P(\m)$, we have 
\begin{align*}
        e_{i} &= \sigma_{i} \left( \RepAdd{0}{\oplus}{\m} \right) = \sigma_{i}(0)
    && (i = 0,1,\dots,\m)
    & & \Comment{\eqref{eq:N1_solv_1}}
\end{align*}
for the messages $x_0,x_1 \dots, x_{\m}$,
and 
\begin{align*}
        e_{i} &= \sigma_{i} \left( \RepAdd{a}{\oplus}{\m} \right) = \sigma_{i}(0)
    && (i = 0,1,\dots,\m)
    & & \Comment{\eqref{eq:N1_solv_1}}
\end{align*}
for the messages $\hat{x}_0,\hat{x}_1 \dots, \hat{x}_{\m}$.
For both collections of messages,
the edge symbols $e_{0}, e_{1},\dots, e_{\m}$ are the same,
and therefore the decoded value $x_0$ at $R_x$ must be the same.
However, this contradicts the fact that $x_0 \ne \hat{x}_0$.
\end{proof}
}

\subsection{Linear solvability conditions of \texorpdfstring{$\Network_1(\m)$}{N1(m)}} \label{ssec:N1_lin}

Lemma~\ref{lem:N1_lin} presents a necessary and sufficient condition 
for the scalar linear solvability of $\Network_1(\m)$ over \niceRModules.
\begin{lemma}
  Let $\m \geq 2$,
  and let $G$ be a \niceRModule.
  Then network $\Network_1(\m)$ is scalar linear solvable over $G$
  if and only if
  $\Char{R}$ is relatively prime to $\m$.
  \label{lem:N1_lin}
\end{lemma}

\newcommand{\ProofOflemNOneLin}{
\begin{proof}[Proof of Lemma \ref{lem:N1_lin}]
By Lemma~\ref{lem:RingInv2},
$\m$ is invertible in $R$ if and only if $\Char{R}$ is relatively prime to $\m$,
so it suffices to show that for each $\m$ and each \niceRModule{} $G$, 
network $\Network_1(\m)$ is scalar linear solvable over $G$
if and only if $\m$ is invertible in $R$.

Assume network $\Network_1(\m)$ is scalar linear solvable 
over \niceRModule{} $G$.
The messages are drawn from $G$, and
there exist $\cc{i}{j}, \cc{\?}{j} \in R$,
such that the edge symbols can be written as:
\begin{align}
  e_{i} &= \bigoplus_{\substack{j=0\\ j\ne i}}^{\m} \left( \cc{i}{j} \act x_j \right)
    & &  (i = 0,1, \dots, \m) 
    \label{eq:N1_lin_1} \\
  e &= \bigoplus_{j = 0}^{\m} \left( \cc{\?}{j} \act x_j \right)
    \label{eq:N1_lin_2}
\end{align}
and there exist $\dd{i}{e}, \dd{\?}{i}, \dd{x}{i} \in R$,
such that each receiver can linearly recover its respective message
from its inputs by:
\begin{align}
R_{i} : \ \ 
  x_i & = \left(\dd{i}{e} \act e \right) \oplus  \left( \dd{\?}{i} \act e_{i}  \right)
    & &  (i = 0,1, \dots,\m) 
    \label{eq:N1_lin_3} \\
R_x : \ \ 
  x_{0} & = \bigoplus_{i = 0}^{\m} \left( \dd{x}{i} \act e_{i} \right).
    \label{eq:N1_lin_4}
\end{align}

Since $\Network_1(\m)$ contains $\Network_0(\m)$,
by Lemma~\ref{lem:N0_lin}
and \eqref{eq:N1_lin_1} -- \eqref{eq:N1_lin_3},
each $\cc{\?}{i}$ and each $\dd{\?}{i}$
is invertible in $R$,
and
\begin{align}
  \cc{i}{j} &= - \dd{\?}{i}^{-1} \, \dd{i}{e} \, \cc{\?}{j}
      & & (i,j = 0,1,\dots, \m \text{ and } j \ne i) .
   \label{eq:N1_lin_5} 
\end{align}

Equating message components at $R_x$ yields:
\begin{align}
1_R &= \sum_{i = 1}^{\m} \dd{x}{i} \, \cc{i}{0}
    & &\Comment{\eqref{eq:N1_lin_1}, \eqref{eq:N1_lin_4}} \notag \\
  & = - \sum_{i = 1}^{\m} \dd{x}{i} \, \dd{\?}{i}^{-1} \, \dd{i}{e} \, \cc{\?}{0}
    & & \Comment{\eqref{eq:N1_lin_5}}
    \label{eq:N1_lin_6} 
\end{align}
and for each $j = 1,2,\dots, \m$,
\begin{align}
0_R &= \sum_{\substack{ i = 0 \\ i \ne j}}^\m \dd{x}{i} \, \cc{i}{j}
    & &\Comment{\eqref{eq:N1_lin_1}, \eqref{eq:N1_lin_4}} \notag \\
  & = -\left( \sum_{\substack{i = 0 \\ i \ne j}}^{\m} \dd{x}{i} \, \dd{\?}{i}^{-1} \, \dd{i}{e} \right) \, \cc{\?}{j}
    & & \Comment{\eqref{eq:N1_lin_5}} .
      \label{eq:N1_lin_07}
\end{align}
For each $j = 1,2,\dots, \m$,
multiplying \eqref{eq:N1_lin_07} on the right by $\cc{\?}{j}^{-1} \, \cc{\?}{0}$ yields
\begin{align}
   0_R &= \sum_{\substack{i = 0 \\ i \ne j}}^{\m} \dd{x}{i} \, \dd{\?}{i}^{-1} \, \dd{i}{e} \, \cc{\?}{0}.
     & & \Comment{\eqref{eq:N1_lin_07}} . 
     \label{eq:N1_lin_7} 
\end{align}

By summing \eqref{eq:N1_lin_7} over $j=1,2,\dots, \m$ and subtracting \eqref{eq:N1_lin_6}, we get
\begin{align}
  -1_R &=  \sum_{j = 0}^{\m} \sum_{\substack{i = 0 \\ i \ne j}}^{\m} \dd{x}{i} \, \dd{\?}{i}^{-1} \, \dd{i}{e} \, \cc{\?}{0}
      & & \Comment{\eqref{eq:N1_lin_6}, \eqref{eq:N1_lin_7}} \notag \\
    &= \m \, \sum_{i = 0}^{\m} \dd{x}{i} \, \dd{\?}{i}^{-1} \, \dd{i}{e} \, \cc{\?}{0} . \notag 
\end{align}
Therefore, $\m$ is invertible in $R$.

To prove the converse,
let $G$ be a \niceRModule{}
such that
$\m$ is invertible in $R$.
Define a scalar linear code
over $G$ by:
\begin{align*}
  e_{i} &= \bigoplus_{\substack{j=0\\ j \ne  i}}^{\m} x_j 
  & &  (i = 0,1,\dots,\m) \\
  e &=\bigoplus_{j = 0}^{\m} x_j.
\end{align*}
Receiver $R_i$ can linearly recover $x_i$ from its received edge symbols $e$ and $e_{i}$ by:
\begin{align*}
  R_i: \ \ &
    e \ominus e_{i} =  x_i 
      & &  (i=0,1,\dots,\m)  
\end{align*}
and
receiver $R_x$ can linearly recover $x_{0}$ from its received edge symbols $e_{0},e_{1},\dots,e_{\m}$ by:
\begin{align*}
  R_x: \ \ &
    \left(\m^{-1} \act \bigoplus_{i = 0}^{\m} e_{i} \right) \ominus e_{0} \\
      & = \left( \m^{-1} \act \bigoplus_{i = 0}^{\m} \bigoplus_{\substack{j = 0\\ j \ne i}}^{\m} x_j \right) 
        \ominus \bigoplus_{j = 1}^{\m} x_j \\
      &= \bigoplus_{j = 0}^{\m} x_j \ominus \bigoplus_{j = 1}^{\m} x_j 
          = x_{0} .
\end{align*}
Thus the code is a scalar linear solution for $\Network_1(\m)$.
\end{proof}
}

\subsection{Capacity and linear capacity of \texorpdfstring{$\Network_1(\m)$}{N1(m)}} \label{ssec:N1_cap}
\begin{definition}
Let $\F$ be a finite field and 
suppose $a_1, \dots, a_q \in \F^{s_i}$ and $b_1, \dots, b_r \in \F^{t_j}$
are functions of variables $x_1, \dots, x_w$.
We write $a_1,\dots,a_q \; \longrightarrow \; b_1,\dots,b_r$ to mean
that there exist $t_j \times s_i$ matrices $M_{j,i}$ over $\F$ such that 
for all choices of the variables $x_1, \dots, x_w$,
\begin{align*}
  b_j &= \sum_{i=1}^q M_{j,i} \, a_i & & (j = 1,\dots, r). 
\end{align*}
\label{def:yields}
\end{definition}
In the context of network coding,
the variables $x_1, \dots, x_w$ will always be taken as the network messages. 
In what follows, 
the transitive relation $\longrightarrow$ 
will be used to describe linear coding functions at network nodes.
Lemma~\ref{lem:mat_1}
is known from linear algebra
\cite[p. 124]{Satake-LinearAlgebra},
and will be used in later proofs.
In particular, Lemmas~\ref{lem:mat_1}, \ref{lem:mat_2}, and \ref{lem:mat_3}
will be used in bounding the linear capacities of $\Network_1, \Network_2$, and $\Network_3$.

\begin{lemma}
Let $\F$ be a finite field.
If $A: \,\F^m \, \to \, \F^n$ and
$B: \, \F^k \, \to \, \F^m$ are linear maps, then
\begin{align}
  \Rank{A} + \Rank{B} - m &\leq \Rank{A \, B}  \label{eq:mat_1_1}\\
        & \leq \min (\Rank{A} , \Rank{B}). \label{eq:mat_1_2}
\end{align}
\label{lem:mat_1}
\end{lemma}

\begin{lemma}
If $A$ is an $n \times k$ matrix of rank $k$ over finite field $\F$, 
then there exists a nonsingular $n\times n$ matrix $B$ such that
\begin{align*}
  B \, A = \left[ \begin{array}{cc}
        I_k \\
        0
      \end{array} \right] .
\end{align*}
\label{lem:mat_2}
\end{lemma}

\newcommand{\ProofOflemMatTwo}{
\begin{proof}[Proof of Lemma \ref{lem:mat_2}]
It follows immediately from Gaussian elimination.
\end{proof}
}

\begin{lemma}
If $A$ is an $m \times n$ matrix of rank $k$ over finite field $\F$, 
then there exists an $(n-k) \times n$ matrix $Q$ over $\F$ of rank $n-k$
such that for all $x \in \F^n$
  $$ Ax, \, Qx \, \longrightarrow \, x. $$
  \vspace{-.75cm}
\label{lem:mat_3}
\end{lemma}

\newcommand{\ProofOflemMatThree}{
\begin{proof}[Proof of Lemma \ref{lem:mat_3}]
Choose $k$ independent rows of $A$,
find $n - k$ members of $\F^n$
which together with the $k$ rows of $A$
form a basis of $\F^n$,
and let the $n-k$ members be the rows of $Q$.
Since the rows of $A$ together with the rows of $Q$ form a basis of $\F^n$,
there exists an $n\times m$ matrix $C_1$ and an $n \times (n-k)$ matrix $C_2$ such that
for all $x \in \F^n$
$$x = C_1 A  x  + C_2 Q x. $$
The results follow immediately.
\end{proof}
}

The following lemma characterizes the capacity and the linear capacity over finite-field alphabets of $\Network_1(\m)$.
\begin{lemma}
  For each $\m \geq 2,$ 
  network $\Network_1(\m)$ has:
  \begin{itemize}
  \itemsep0em 
    \item[(a)] capacity equal to $1$,
    \item[(b)] linear capacity equal to $1$
      for any finite-field alphabet whose characteristic 
      does not divide $\m$,
    \item[(c)] linear capacity equal to
      $1 - \frac{1}{2\m + 2}$ 
      for any finite-field alphabet whose characteristic divides $\m$.
  \end{itemize}
\label{lem:N1_cap}
\end{lemma}
\newcommand{\ProofOflemNOnecap}{
\begin{proof}[Proof of Lemma \ref{lem:N1_cap}]
Since a scalar linear solution over a finite-field alphabet
is a special case of a scalar linear solution over a \niceRModule,
by Lemma~\ref{lem:N1_lin},
$\Network_1(\m)$ is scalar linear solvable 
over any finite-field alphabet whose characteristic does not divide $\m$,
so the network's linear capacity for such finite-field alphabets is at least $1$.
By Lemma~\ref{lem:N0_cap}, 
network $\Network_0(\m)$ has capacity equal to $1$,
and since $\Network_1(\m)$ contains $\Network_0(\m)$,
the capacity of $\Network_1(\m)$ is at most $1$.
Thus, both
the capacity of $\Network_1(\m)$ 
and its linear capacity for finite-field alphabets whose characteristic does not divide $\m$
are equal to $1$.

To prove part (c),
consider a $(k,n)$ fractional linear solution for $\Network_1(\m)$
over a finite field $\F$
whose characteristic divides $\m$.
Since $\Div{\Char{\F}}{\m}$, 
we have $\m = 0$ in $\F$.

We have $x_i \in \F^k$ and $e, e_i \in \F^n$,
with $n \geq k$, since the capacity is one.
There exist $n \times k$ coding matrices
$\MM{\?}{j},\MM{i}{j}$ with entries in $\F$,
such that the edge vectors can be written as:
\begin{align}
  e_i &= \sum_{\substack{j = 0 \\ j \ne i }}^{\m} \MM{i}{j} \, x_j 
    && (i = 0,1, \dots, \m)
    \label{eq:N1_cap_1} \\
  e &= \sum_{j = 0}^{\m} \MM{\?}{j} \, x_j 
    \label{eq:N1_cap_2}
\end{align}
and there exist $k \times n$ decoding matrices $\DD{i}{e}, \,  \DD{\?}{i}$
with entries in $\F$, such that
each $x_i$ can be linearly decoded at $R_i$
from the two $n$-vectors $e$ and $e_i$ by:
\begin{align}
  R_i: \ \ 
    x_i &= \DD{i}{e} \, e + \DD{\?}{i} \, e_i 
    & & (i = 0,1,\dots, \m).
    \label{eq:N1_cap_3}
\end{align}
Since receiver $R_x$ linearly recovers $x_0$ 
from $e_0,e_1,\dots,e_{\m}$, 
we can write
\begin{align}
  e_0,e_1,\dots, e_{\m} \; \longrightarrow \; x_0 .
\label{eq:N1_cap_4}
\end{align}

For each $i=0,1 \dots, \m$, if we set $x_i = 0$ in \eqref{eq:N1_cap_3},
then we get the following relationship among the remaining $\m$ messages
(since $e_i$ does not depend on $x_i$):
\begin{align}
0 &= \DD{i}{e} \, \sum_{\substack{j=0 \\ j\ne i}}^{\m} \MM{\?}{j} \, x_j   + \DD{\?}{i} \, e_i
  & & (i = 0,1,\dots,\m)
  &   \Comment{\eqref{eq:N1_cap_1}, \eqref{eq:N1_cap_2}, \eqref{eq:N1_cap_3}} ,
\label{eq:N1_cap_5} 
\end{align}
and thus
\begin{align}
    e_i & \longrightarrow \;  \DD{i}{e} \, \sum_{\substack{j=0\\ j\ne i}}^{\m} \MM{\?}{j} \, x_j  
       & & (i = 1,2,\dots,\m)
       & &\Comment{\eqref{eq:N1_cap_5}} 
        \label{eq:N1_cap_6} \\
    \sum_{j=1}^{\m} \MM{\?}{j} \, x_j  & \longrightarrow \; \DD{\?}{0} \, e_0
  & & 
  & &\Comment{\eqref{eq:N1_cap_5}} .
        \label{eq:N1_cap_7}
\end{align} 

For each $i=1,\dots,\m$, 
let $\QQ{i}{e}$ be the matrix $Q$ in Lemma~\ref{lem:mat_3}
corresponding to when $\DD{i}{e}$ is the matrix $A$ in Lemma~\ref{lem:mat_3}.
Similarly, let $\QQ{\?}{0}$ be the matrix $Q$ in Lemma~\ref{lem:mat_3} corresponding to taking $A$ to be $\DD{\?}{0}$.
Let $L$ be the following list of $2\m +1$ vector functions of $x_0,x_1, \dots, x_{\m}$:
\begin{align*}
  &  \QQ{\?}{0} \, e_0,  \\  
  & e_i,  
    & & (i = 1,2,\dots,\m)  \\
  & \QQ{i}{e} \, \sum_{\substack{j=0\\ j\ne i}}^{\m} \MM{\?}{j} \, x_j 
    & & (i = 1,2,\dots,\m) .  
\end{align*}


We have
\begin{align}
L & \longrightarrow  \DD{i}{e} \, \sum_{\substack{j=0\\ j\ne i}}^{\m} \MM{\?}{j} \, x_j 
    & & (i = 1,2,\dots,\m)
    & & \Comment{\eqref{eq:N1_cap_6}}
       \label{eq:N1_cap_8} \\
L & \longrightarrow \sum_{\substack{j=0\\ j\ne i}}^{\m} \MM{\?}{j} \, x_j 
    & & (i = 1,2,\dots,\m) 
    & & \Comment{Lemma~\ref{lem:mat_3}, \eqref{eq:N1_cap_8}}
        \label{eq:N1_cap_9} ,
\end{align}
and
\begin{align}
& \left\{ \sum_{\substack{j=0\\ j\ne i}}^{\m} \MM{\?}{j} \, x_j :
  \  i = 1,2,\dots,\m \right\} \notag \\
&\; \; \longrightarrow 
  \sum_{i=1}^{\m} \sum_{\substack{j=0\\ j\ne i}}^{\m} \MM{\?}{j} \, x_j 
      \notag\\
  &\; \; = \m \, \MM{\?}{0} \, x_{0} + (\m-1) \, \sum_{j = 1}^{\m} \MM{\?}{j} \, x_j 
      \notag \\
  &\; \; =  - \sum_{j = 1}^{\m} \MM{\?}{j} \, x_j  
    & & \Comment{$\Div{\Char{\F}}{\m}$} .
     \label{eq:N1_cap_10}  
\end{align}

Thus we have
\begin{align}
L & \longrightarrow \sum_{j=1}^{\m} \MM{\?}{j} \, x_j   
  & & \Comment{\eqref{eq:N1_cap_9}, \eqref{eq:N1_cap_10}}     \label{eq:N1_cap_11} \\
L & \longrightarrow \DD{\?}{0} \, e_0
  & & \Comment{\eqref{eq:N1_cap_7}, \eqref{eq:N1_cap_11}}     \label{eq:N1_cap_12} \\
L & \longrightarrow  e_0
  & & \Comment{Lemma~\ref{lem:mat_3}, \eqref{eq:N1_cap_12}}    \label{eq:N1_cap_13} \\
L & \longrightarrow x_{0}  
  & & \Comment{\eqref{eq:N1_cap_4}, \eqref{eq:N1_cap_13}}     \label{eq:N1_cap_14} \\
x_{0},\ \ \ 
   \sum_{j=1}^{\m} \MM{\?}{j} \, x_j 
  &\longrightarrow  e
  & & \Comment{\eqref{eq:N1_cap_2}}  
  \label{eq:N1_cap_15} \\
L &\longrightarrow  e 
  & & \Comment{\eqref{eq:N1_cap_11}, \eqref{eq:N1_cap_14}, \eqref{eq:N1_cap_15}}
  \label{eq:N1_cap_16} \\
L &\longrightarrow x_i 
    \; \; \quad \quad \quad \quad  (i = 1,2,\dots, \m) 
    & & \Comment{\eqref{eq:N1_cap_3}, \eqref{eq:N1_cap_16}} .
    \label{eq:N1_cap_17} 
\end{align}

We will now bound the number of independent entries in the list $L$. 
By equating message components in equation \eqref{eq:N1_cap_3},
we have:
\begin{align}
I_k &= \DD{i}{e} \, \MM{\?}{i} 
  \ \ \ \ \ \ \ \ \ (i = 0,1,\dots,\m)
  && \Comment{\eqref{eq:N1_cap_1}, \eqref{eq:N1_cap_2}, \eqref{eq:N1_cap_3}} 
  \label{eq:N1_cap_18}.
\end{align}

Since each $\DD{i}{e}$ and $\MM{\?}{i}$ are $k \times n$ and $n \times k$, respectively,
and $k \leq n$,
the rank of each matrix is at most $k$,
but we also have
\begin{align*}
\min\left(\Rank{\DD{i}{e}} \!,\, \Rank{\MM{\?}{i}} \right) &\geq  \Rank{\DD{i}{e} \, \MM{\?}{i}}  
  & & \Comment{\eqref{eq:mat_1_2}}  \\
& = \Rank{I_k} = k
  & & \Comment{\eqref{eq:N1_cap_18}}  ,
\end{align*}
and so $\Rank{\DD{i}{e}} = \Rank{\MM{\?}{i}} = k$,
which, by Lemma~\ref{lem:mat_3}, implies 
\begin{align}
  \Rank{\QQ{i}{e}} = n - k \ \ \ \ \ \ \ \ \ \ 
    (i = 1,2,\dots, \m). 
    \label{eq:N1_cap_19}
\end{align}

Since $\Rank{\MM{\?}{0}} = k$,
by Lemma~\ref{lem:mat_2}, there exists an $n\times n$ 
nonsingular matrix $W$ over $\F$ such that
\begin{align}
W \MM{\?}{0} 
  = \left[
    \begin{array}{c}
      I_k \\
      0_{(n-k) \times k}
    \end{array}
  \right] . \label{eq:N1_cap_20}
\end{align}
Partition each of the $k \times n$ matrix products $\DD{i}{e}W^{-1}$ 
into a $k \times k$ block $T_{i}$ to the left of a $k \times (n-k)$ block $U_{i}$:
\begin{align}
  \DD{i}{e} W^{-1} = \left[T_{i} \ \ \ \ U_{i} \right] 
    \label{eq:N1_cap_21}
\end{align}
and then let $V$ be the following $n\times n$ matrix over $\F$:
\begin{align}
V = \left[
  \begin{array}{cc}
    I_k & U_{0} \\
    0_{(n-k)\times k} & I_{n-k}
  \end{array}
  \right]. 
    \label{eq:N1_cap_22}
\end{align}
It is easy to verify that
\begin{align}
V^{-1} = \left[
  \begin{array}{cc}
    I_k & -U_{0} \\
    0_{(n-k)\times k} & I_{n-k}
  \end{array}
  \right] . 
    \label{eq:N1_cap_23}
\end{align}
For each $i = 0,1,\dots, \m$,
change the network encoding and decoding matrices from $\MM{\?}{i}$ and $\DD{i}{e}$, respectively, to
\begin{align}
  \MM{\?}{i}' &= V W \MM{\?}{i} 
    \label{eq:N1_cap_24}\\
  \DD{i}{e}' &= \DD{i}{e} W^{-1}V^{-1}. 
    \label{eq:N1_cap_25}
\end{align}
We have 
\begin{align}
  T_{0} &= \DD{0}{e} W^{-1} W \MM{\?}{0} = I_k 
    & & \Comment{\eqref{eq:N1_cap_18}, \eqref{eq:N1_cap_20}, \eqref{eq:N1_cap_21}}
    \label{eq:N1_cap_26}
\end{align}
and therefore
\begin{align}
  \MM{\?}{0}' &= \left[
    \begin{array}{c}
      I_k \\
      0
    \end{array}
  \right] 
    & & \Comment{\eqref{eq:N1_cap_20}, \eqref{eq:N1_cap_22}, \eqref{eq:N1_cap_24}} 
      \notag \\
  \DD{0}{e}' & = \left[I_k \ \ \ \ 0\right] 
    & & \Comment{\eqref{eq:N1_cap_21}, \eqref{eq:N1_cap_23}, \eqref{eq:N1_cap_25}, \eqref{eq:N1_cap_26}} .
      \label{eq:N1_cap_27}
\end{align}
In this case,
\begin{align*}
  e' &= \sum_{j = 0}^{\m} \MM{\?}{j}' \, x_j
\end{align*}
and for each $i = 0,1,\dots,\m$,
the messages can be recovered by:
\begin{align*}
  \DD{i}{e}' e' + \DD{\?}{i} e_i 
    &= \DD{i}{e} W^{-1} V^{-1} \sum_{j=0}^\m V W \MM{\?}{j} \, x_j + \DD{\?}{i} e_i 
      & & \Comment{\eqref{eq:N1_cap_24}, \eqref{eq:N1_cap_25}} \\
    &= \DD{i}{e} e + \DD{\?}{i} e_i = x_i  
      & & \Comment{\eqref{eq:N1_cap_2}, \eqref{eq:N1_cap_3}} .
\end{align*}
Thus, this linear code still provides a $(k,n)$ solution.

Partition each of the matrices $\MM{\?}{i}$ 
into a $k\times k$ block $R_{i}$ on top of a $(n-k) \times k$ block $S_{i}$:
\begin{align}
\def\arraystretch{2.2}
\MM{\?}{i} 
  &= \left[ \begin{array}{cc}
    R_{i}\\
    S_{i}
  \end{array}
 \right]  
  \label{eq:N1_cap_28}
\end{align}
and let 
$$\rho = \Rank{\left[R_{1}  \ \ \ \dots \ \ \ R_{\m} \right]}$$
where $\left[R_{1}  \ \ \ \dots \ \ \ R_{\m} \right]$ is the concatenation of the matrices
$R_{i}$ into a $k\times \m k$ matrix.
Clearly $\rho \le k$.
We have
\begin{align*}
  \DD{\?}{0} \, \sum_{j = 1}^{\m} \MM{0}{j} \, x_j 
    = \DD{\?}{0} \, e_0 
    &=- \DD{0}{e} \, \sum_{j = 1}^{\m} \MM{\?}{j} \, x_j 
      && \Comment{\eqref{eq:N1_cap_1}, \eqref{eq:N1_cap_5}}\\  
    & = - \sum_{j = 1}^{\m} R_{j} \, x_j
      && \Comment{\eqref{eq:N1_cap_27}, \eqref{eq:N1_cap_28}}.  
\end{align*}
This gives us
\begin{align*}
  \DD{\?}{0} \, [\MM{0}{1}  \ \ \ \dots \ \ \ \MM{0}{\m}] 
    & = -\left[R_{1}  \ \ \ \dots \ \ \ R_{\m} \right], 
\end{align*}
which implies 
\begin{align}
  \Rank{\DD{\?}{0}}
    &\geq \Rank{[R_{1} \ \ \ \dots \ \ \ R_{\m} ]} = \rho
      & & \Comment{\eqref{eq:mat_1_2}}  \notag \\
  &\therefore \; \Rank{\QQ{\?}{0}} = n - \Rank{\DD{\?}{0}} \leq n - \rho.
    \label{eq:N1_cap_29} 
\end{align}

Since the matrix $\left[R_{1}  \ \ \ \dots \ \ \ R_{\m} \right]$ 
has rank $\rho$,
there exists a $k\times k$ permutation matrix $P$ such that
the first $\rho$ rows of $P \, \left[R_{1}  \ \ \ \dots \ \ \ R_{\m} \right]$ are linearly independent
and the remaining $k - \rho$ rows are linear combinations of those first $\rho$ rows. 
Thus, there exists a $(k - \rho) \times k$ matrix $X$,
whose right-most $k-\rho$ columns form $I_{k-\rho}$, and such that
\begin{align}
  X P \, \left[R_{1}  \ \ \ \dots \ \ \ R_{\m} \right]  
    &= 0_{(k-\rho) \times \m  k}
    \label{eq:N1_cap_30}.
\end{align}

$X$ and $P$ are $(k - \rho) \times k$ and $k\times k$ respectively,
thus the rank of $X$ is at most $(k - \rho)$ and the rank of $P$ is at most $k$.
Since the right-most columns of $X$ form $I_{k -\rho}$, 
we have $\Rank{X} = k - \rho$,
and since $P$ is a permutation matrix,
we have $\Rank{P} = k$.
Since $XP$ is $(k-\rho)\times k$, we have 
\begin{align*}
  k - \rho &\geq \Rank{X P}   \\
        & \geq  \Rank{X} + \Rank{P} - k &\Comment{\eqref{eq:mat_1_1}}  \\
        & = (k - \rho) + k - k = k - \rho 
\end{align*}
and thus $\Rank{X P} = k - \rho$. \\
Define a $(k - \rho) \times n$ matrix $Y$ by
concatenating the product $XP$ with an all-zero matrix as follows:
$Y = \left[X P \ \ \ \ \ 0_{(k-\rho)\times(n-k)} \right]$. 
For each $i = 1,2,\dots, \m$ we have 
\begin{align}
Y \MM{\?}{i} &= \left[XP \ \ \ \ \ 0_{(k-\rho)\times(n-k)} \right] 
  \, \left[ \begin{array}{cc}
      R_{i}\\
      S_{i}
    \end{array}
  \right] 
  = 0_{(k-\rho) \times k} 
  & & \Comment{\eqref{eq:N1_cap_28}, \eqref{eq:N1_cap_30}}.
 \label{eq:N1_cap_31}
\end{align}

Since, for each $i=1,2,\dots,\m$, 
we have $Y \MM{\?}{i} = 0_{(k-\rho) \times k}$
and by \eqref{eq:N1_cap_18},
$\DD{i}{e} \MM{\?}{i} = I_k$,
the rows of $Y$ and the rows of $\DD{i}{e}$ 
are linearly independent.
(If $v$ is a nontrivial linear combination of rows of $\DD{i}{e}$,
then $v \MM{\?}{i} \ne 0$;
if $v'$ is a nontrivial linear combination of rows of $Y$,
then $v' \MM{\?}{i} = 0$,
so $v \ne v'$).
Therefore, by Lemma~\ref{lem:mat_3},
we may choose $\QQ{i}{e}$ 
such that its first $k- \rho$ rows 
are the rows of $Y$. 
By \eqref{eq:N1_cap_19}, each vector function
\[\QQ{i}{e} \, \sum_{\substack{j=0 \\ j\ne i}}^{\m} \MM{\?}{j} \, x_j \] 
in the list $L$ has dimension $n-k$,
but the first $k- \rho$ components of each such vector function can be written as 
\begin{align}
Y \, \sum_{\substack{j=0 \\ j\ne i}}^{\m} \MM{\?}{j} \, x_j
  &=  Y \MM{\?}{0} \, x_{0} & \Comment{\eqref{eq:N1_cap_31}}
  \label{eq:N1_cap_32}.
\end{align}

If we view the message vectors $x_0, x_1, \dots, x_{\m}$ as random variables,
each of whose $k$ components are independent and uniformly distributed over the field $\F$,
then we have the following entropy (using logarithms with base $|\F|$) upper bounds:
\begin{align*}
  H\left(\QQ{\?}{0} e_0 \right) 
    &\le n-\rho 
      & & \Comment{\eqref{eq:N1_cap_29}} \\
  H\left (e_{1}, \dots, e_{\m} \right) 
    &\le  \m n
      & & \Comment{$e_i \in \F^n$} \\
  H \left( \QQ{i}{e} \, \sum_{\substack{j=0 \\ j\ne i}}^{\m} \MM{\?}{j} \, x_j  \; : \;  i = 1,2,\dots, \m \right)
    &\le \m \, (n-k)  - (\m - 1) \,  (k - \rho) 
      & & \Comment{\eqref{eq:N1_cap_19}, \eqref{eq:N1_cap_32}}  .
\end{align*}

Therefore, the entropy of all of the vector functions in the list $L$ is bounded 
by summing these bounds:
\begin{align}
  H(L) &\le  ( \m (n-k) - (\m-1)(k-\rho) ) + (n-\rho) + \m n 
      \notag \\
    &= (2\m +1) n - (\m+1)k -(k-\rho)(\m-2) 
      \notag \\
    &\le (2\m + 1)n - (\m+1)k   
      & & \Comment{$\rho \leq k$ and $\m \geq 2$}. 
      \label{eq:N1_cap_33}
\end{align}
But then we have:
\begin{align*}
  (\m+1)k &= H(x_0, x_1, \dots, x_{\m}) 
    & & \Comment{$x_i \in \F^k$} 
    \\
  &\le H(L) 
    & & \Comment{\eqref{eq:N1_cap_14}, \eqref{eq:N1_cap_17}} 
    \\
  &\le (2\m + 1) \, n - (\m + 1) \, k 
    & & \Comment{\eqref{eq:N1_cap_33}} 
    \\
  \therefore \dfrac{k}{n} &\le \dfrac{2\m +1}{2\m + 2}. 
\end{align*}
Thus the linear capacity of $\Network_1(\m)$
for any finite-field alphabet whose characteristic divides $\m$
is upper bounded by
$$1 - \frac{1}{2 \m + 2}. $$

For each $y \in \F^m$, 
let $[y]_i$ denote the $i$th component 
of $y$. 
To show the upper bound on the linear capacity is tight,
consider a ($2 \m + 1, 2 \m + 2 $) fractional linear code 
for $\Network_1(\m)$ over any finite-field alphabet whose characteristic divides $\m$,
given by:

\begin{align*}
[e_0]_l &= \left\{ \begin{array}{ll}
      \dsum_{\substack{ j = 1 \\ j \ne l}}^{\m} [x_j]_l \; \; \; \; \; \; \; \; \; \; \; \; \; \; \; \;
        &  (l = 1,2, \dots, \m)  \\[1cm]
      \dsum_{j = 1}^{\m} [x_j]_l 
         &  (l = \m+1,\dots,2\m +1) \\[.75cm]
      \dsum_{j =2}^{\m} [x_j]_{j}
        &  (l = 2\m +2)
\end{array} \right. \\
\\
[e_i]_l  &= \left\{ \begin{array}{ll}
    \dsum_{\substack{ j = 0 \\ j \ne i \\ j \ne l}}^{\m} [x_j]_l 
      &  (l = 1,2,\dots, \m \text{ and } l \ne i ) \\[1.25cm]
    \displaystyle [x_0]_{\m+1} + \sum_{ \substack{ j  = 1 \\ j \ne i }}^{\m} [x_j]_{j} 
      &  (l = i) \\[1cm]
    \dsum_{\substack{ j = 0 \\ j \ne i}}^{\m} [x_j]_l 
      &  (l = \m+1,\dots,2\m +1) \\[1cm]
    [x_{0}]_{\m + 1 + i} 
      &  (l  = 2 \m +2)
    \end{array} \right.
      & (i = 1,2,\dots, \m) \\
\\
[e]_l &= \left\{ \begin{array}{ll}
      \dsum_{\substack{j = 0 \\ j \ne l}}^{\m} [x_j]_l 
         &  (l = 1,2,\dots, \m)  \\[1cm]
      \dsum_{j= 0}^{\m} [x_j]_l    
        &  (l = \m+1,\dots,2\m +1) \\[.75cm]
      [x_0]_{\m+1} + \dsum_{j = 1}^{\m} [x_j]_{j} 
        & (l = 2\m +2) .
    \end{array} \right.
\end{align*}
For each $l = 1,2,\dots,\m$,
we have
\begin{align}
  \sum_{\substack{i = 0 \\ i \ne l}}^{\m} [e_i]_l  
    & = \sum_{\substack{i = 0 \\ i \ne l}}^{\m} \sum_{\substack{ j = 0 \\ j \ne i \\ j \ne l }}^{\m} [x_j]_l 
    = (\m-1) \, \sum_{\substack{j = 0 \\ j \ne l}}^{\m} [x_j]_l 
    = - \sum_{\substack{j = 0 \\ j \ne l}}^{\m} [x_j]_l 
        & & \Comment{$\Div{\Char{\F}} {\m}$}.
      \label{eq:N1_cap_34}
\end{align}

For each $i = 1,2,\dots,\m$,
the receivers within $\B(\m)$ can linearly recover all $2\m+1$ components of their respective demands by:
\begin{align*}
R_{0}: \ \ &  
    [e]_{l} - [e_0]_{l}  = [x_{0}]_{l} 
       & & (l = 1,2,\dots, 2\m +1 ) \\
      \\
  R_i: \ \ & 
      {[e]}_l - {[e_i]}_l = [x_i]_l 
        & & (l = 1,2,\dots, 2\m + 1 \text{ and } l \ne i )\\
     & {[e]}_{2\m +2} - {[e_i]}_i = [x_i]_i
\end{align*}
and the additional receiver can linearly recover all components of $x_0$ by:
\begin{align*}
R_x: \ \ &  
   - [e_0]_l - \sum_{\substack{i = 0 \\ i \ne l}}^{\m} [e_i]_l
     = [x_{0}]_l
      && (l = 1,2,\dots,\m)
      & & \Comment{\eqref{eq:N1_cap_34}}\\ 
     & [e_1]_{1} - [e_0]_{2\m +2}  = [x_0]_{\m+1} \\
     &  [e_{l-\m-1}]_{2\m +2} = [x_{0}]_l
      && (l = \m+2,\dots, 2\m +1) .
\end{align*}

Thus, the code is in fact a solution for $\Network_1(\m)$.
\end{proof}
}

\clearpage

\section{The network \texorpdfstring{$\Network_2(\m,\n)$}{N2(m,w)}} \label{sec:N2}
\psfrag{N01(N+1)}{$\B^{(1)}(\m +1)$}
\psfrag{N02(N+1)}{$\B^{(2)}(\m +1)$}
\psfrag{N0M(N+1)}{$\B^{(\n)}(\m +1)$}

\psfrag{S11}{\tiny$\Ss{1}{1}$}
\psfrag{S12}{\tiny$\Ss{1}{2}$}
\psfrag{S1N+1}{\tiny$\Ss{1}{\m+1}$}

\psfrag{S21}{\tiny$\Ss{2}{1}$}
\psfrag{S22}{\tiny$\Ss{2}{2}$}
\psfrag{S2N+1}{\tiny$\Ss{2}{\m+1}$}

\psfrag{SM1}{\tiny$\Ss{\n}{1}$}
\psfrag{SM2}{\tiny$\Ss{\n}{2}$}
\psfrag{SMN+1}{\tiny$\Ss{\n}{\m+1}$}

\psfrag{Sz}{$S_z$}

\psfrag{e10}{$\ee{1}{0}$}
\psfrag{e11}{$\ee{1}{1}$}
\psfrag{e12}{$\ee{1}{2}$}
\psfrag{e1N+1}{$\ee{1}{\m +1}$}

\psfrag{e20}{$\ee{2}{0}$}
\psfrag{e21}{$\ee{2}{1}$}
\psfrag{e22}{$\ee{2}{2}$}
\psfrag{e2N+1}{$\ee{2}{\m +1}$}

\psfrag{eM0}{$\ee{\n}{0}$}
\psfrag{eM1}{$\ee{\n}{1}$}
\psfrag{eM2}{$\ee{\n}{2}$}
\psfrag{eMN+1}{$\ee{\n}{\m +1}$}

\psfrag{x11}{$\xx{1}{1}$}
\psfrag{x12}{$\xx{1}{2}$}
\psfrag{x1N+1}{$\xx{1}{\m+1}$}

\psfrag{x21}{$\xx{2}{1}$}
\psfrag{x22}{$\xx{2}{2}$}
\psfrag{x2N+1}{$\xx{2}{\m+1}$}

\psfrag{xM1}{$\xx{\n}{1}$}
\psfrag{xM2}{$\xx{\n}{2}$}
\psfrag{xMN+1}{$\xx{\n}{\m+1}$}

\psfrag{z}{$z$}

\psfrag{Rz}{\small$R_{z}$}
\begin{figure}[h]
  \centering
  \leavevmode
  \hbox{\epsfxsize=0.975\textwidth\epsffile{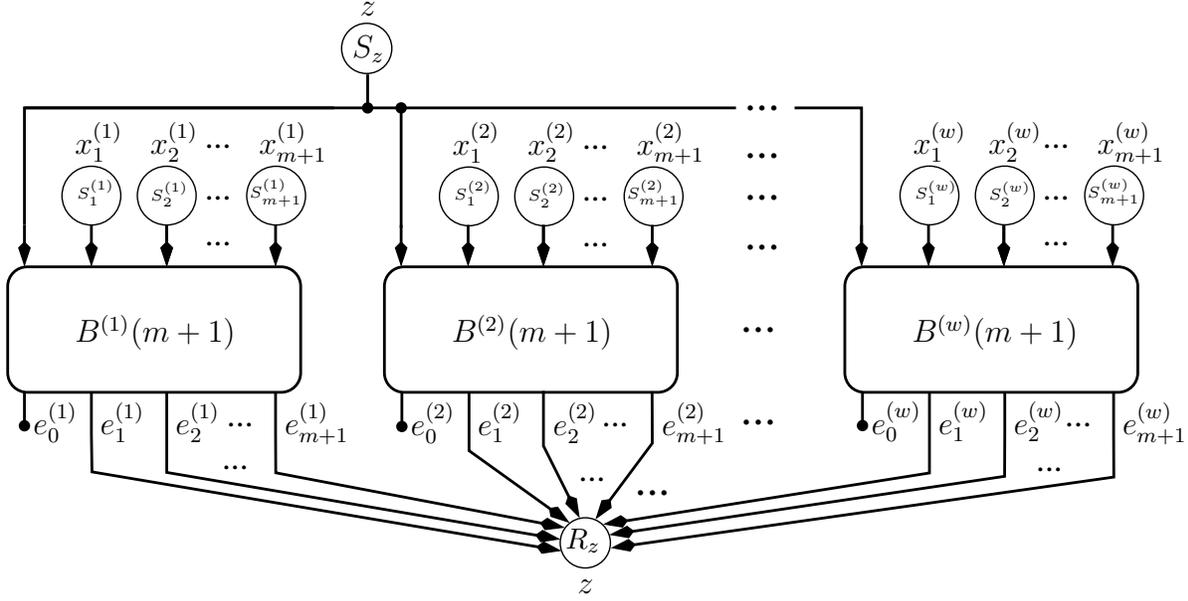}}
  \caption{Network $\Network_2(\m ,\n)$ is constructed from $\n$ blocks of $\B(\m +1)$
  together with
  $\n(\m+1) +1$ source nodes and an additional receiver $R_z$.
  The $l$th block is denoted $\B^{(l)}(\m + 1)$,
  and the nodes and edge symbols within $\B^{(l)}(\m+1)$ are denoted with a superscript $l$.
  For each $l = 1,2,\dots,\n$, the block $\B^{(l)}(\m+1)$ has inputs from source nodes $\Ss{l}{1},\Ss{l}{2},\dots,\Ss{l}{\m+1}$,
  which generate messages $\xx{l}{1},\xx{l}{2}, \dots, \xx{l}{\m+1}$.
  The shared message $z$ is generated by source node $S_z$ and
  is the $0$th input to each $\B^{(l)}(\m+1)$.
  Each of the output edges of $\B^{(l)}(\m+1)$, except the $0$th, is an input to the shared receiver $R_z$,
  which demands the shared message $z$.
  }
  \label{fig:N2}  
\end{figure}
%


For each $\m \geq 2$ and $\n \geq 1$, network $\Network_2(\m,\n)$ is
defined in Figure~\ref{fig:N2}.
We note that $\Network_2(\m,1)$ and $\Network_1(\m+1)$ have similar structure,
but in network $\Network_1(\m+1)$ each of the output edges of $\B(\m+1)$ is connected to $R_x$,
and in network $\Network_2(\m,1)$ all but one of the output edges of $\B(\m+1)$ are connected to $R_z$.
This disconnected edge causes the difference in solvability properties of the two networks.
Lemmas~\ref{lem:N2_non}, \ref{lem:N2_solv}, \ref{lem:N2_lin}, and \ref{lem:N2_cap}
demonstrate that network $\Network_2(\m,\n)$ is:
\begin{enumerate}
\itemsep0em 
  \item non-linear solvable over an alphabet of size $\m \n$, if $\n \geq 2$, 
  \item solvable over alphabet $\A$ only if $|\A|$ is not relatively prime to $\m$,
  \item scalar linear solvable over \niceRModule{} $G$ if and only if $\Char{R}$ divides $\m$,
  \item asymptotically linear solvable over finite field $\F$ if and only if $\Char{\F}$ divides $\m$.
\end{enumerate}

\begin{remark}
For each $\m \geq 2$ and $\n \geq 1$
network $\Network_2(\m,\n)$ has 
$ \n  (\m + 1) + 1$  source nodes,
$ \n  (2 \m + 6)$ intermediate nodes,
and 
$ \n  (\m + 2) + 1$ receiver nodes,
so the total number of nodes in $\Network_2(\m,\n)$
is $4 \m \n + 9 \n + 2$.
  \label{rem:N2_nodes}
\end{remark}

\subsection{Solvability conditions of \texorpdfstring{$\Network_2(\m,\n)$}{N2(m,w)}} \label{ssec:N2_solv}

For each positive integer $\m$, 
we will view the ring $\Z_{\m}$ as the set $\{0,1,\dots,\m-1\}$ together with
addition and multiplication modulo $\m$.
This ring will be used to construct non-linear solutions
in Lemmas~\ref{lem:N2_P}, \ref{lem:N2_non}, \ref{lem:N3_P}, and \ref{lem:N3_non}.

For each $\m, \n \geq 2$ and $a \in \Z_{\m\n}$,
a receiver cannot uniquely determine the symbol $a$ in $\Z_{\m\n}$
from the symbol $\n a \in \Z_{\m\n}$
since $\n$ is not invertible in $\Z_{\m\n}$. 
For example, if a receiver receives $\n a = 0$ in $\Z_{\m\n}$, then the symbol
$a$ could be any element in the set $\{0, \m, 2\m, \dots, (\n-1)\m\}$.
The following lemma describes a technique
for recovering the value of $a$
via a decoding function $\psi$
from the $\n$-tuple $\n \pi_1(a), \n \pi_2(a), \dots, \n \pi_{\n}(a)$,
where each $\pi_i$ is a particular permutation of $\Z_{\m\n}$.
This technique will then be used
to show that network $\Network_2(\m,\n)$ is solvable over an alphabet of size $\m \n$.

\begin{lemma}
  For each $\m \geq 2$ and $\n \geq 1$,
  there exist permutations $\pi_1,\pi_2, \dots,\pi_{\n}$ of $\Z_{\m\n}$
  and a mapping $\psi: \Z_{\m \n}^{\n} \to \Z_{\m \n}$
  such that for all $a \in \Z_{\m\n}$
  $$ \psi\left( \n \pi_1(a), \n \pi_2(a), \,  \dots, \, \n \pi_{\n}(a) \right) = a .$$
  \vspace{-.75cm}
\label{lem:N2_P}
\end{lemma}

\newcommand{\ProofOflemNTwoP}{
\begin{proof}[Proof of Lemma \ref{lem:N2_P}]
Assume $\n = 1$
and let $\pi_1$ and $\psi$ be identity permutations.
For each $a \in \Z_{\m\n}$ we have 
$$\psi(\n \pi_1(a) ) = \psi(a) = a.$$

Assume $\n > 1$.
By the Euclidean Division Theorem,
for each integer $y$,
there exist unique integers $q_y,r_y$
such that $y = q_y  \m + r_y$
and $0 \leq r_y < \m$.
We have 
$\n y = \n (q_y \m + r_y)$,
which implies
\begin{align}
  \n  y  & = \n r_y \; \; \Mod{\m \n}
  \label{eq:N2_P_1} .
\end{align}
For all integers $x,y$ we have
\begin{align}
  \n x = \n y \; \; \; \Mod{\m \n}
    & \Longleftrightarrow 
    \n r_x = \n r_y \; \; \; \Mod{\m \n} 
      & & \Comment{\eqref{eq:N2_P_1}} \notag \\
    & \Longleftrightarrow r_x = r_y
      & & \Comment{$0 \leq r_x,r_y < \m$}.
      \label{eq:N2_P_2}
\end{align}

For each $a = q_a \m + r_a \in \Z_{\m\n}$ such that $r_a \in \{0,1,\dots,\m-1\}$,
let $\hat{r}_a$ be the unique integer in $\{0,1,\dots,\m-1\}$ such that $\hat{r}_a = r_a + 1 \; \; \Mod{\m}$,
and define permutations $\pi_1,\pi_2,\dots,\pi_{\n}$ of $\Z_{\m\n}$ as follows:
\begin{align}
  \pi_l(a) &= \left\{ \begin{array}{lc}
        q_a  \m + \hat{r}_a & \text{if }  q_a = l \\
        q_a  \m + r_a &  \text{ otherwise} 
  \end{array} \right.
  & & (l = 1,2,\dots,\n-1) 
  \label{eq:N2_P_3}\\  
  \pi_{\n}(a) &= a = q_a \m + r_a 
  \label{eq:N2_P_4}.   
\end{align}
Note that for all $l = 1,2,\dots,\n-1$, 
the (non-linear) permutation $\pi_l$
modifies the remainder $r_a$ if $q_a = l$
and otherwise acts as the identity permutation.
Also, $\pi_{\n}$ is the identity permutation.
Since $a \in \Z_{\m\n}$,
we have $0 \leq q_a, < \n$.

For each $a \in \Z_{\m\n}$ 
we will show the mapping
$a \longmapsto (\n \pi_{1}(a), \dots, \n \pi_{\n}(a))$
is injective.
For each $a,b \in \Z_{\m\n}$, suppose
\begin{align}
  \n \pi_l(a) &= \n  \pi_l(b) \; \; \Mod{\m \n}
    & & (l = 1,2,\dots,\n)
    \label{eq:N2_P_5},
\end{align}
where $a = q_a  \m + r_a$ and
$b = q_b  \m + r_b$,
with $0 \leq r_a,r_b < \m$ and $0 \leq q_a,q_b <\n$.
Then we have
\begin{align}
  \n \pi_{\n}(a) &= \n \pi_{\n}(b)
    & & \! \! \! \Mod{\m \n}
    & & \Comment{\eqref{eq:N2_P_5}} 
    \label{eq:N2_P_6}\\
  \n r_a & = \n r_b
    & & \! \! \! \Mod{\m \n}
    & & \Comment{\eqref{eq:N2_P_1}, \eqref{eq:N2_P_4},\eqref{eq:N2_P_6}} 
    \notag \\
\therefore
  r_a &= r_b
    & &  
    & & \Comment{\eqref{eq:N2_P_2}}
    \label{eq:N2_P_7}.
\end{align}

Let $\hat{r}_b$ be the unique integer in $\{0,1,\dots,\m-1\}$ such that
$\hat{r}_b = r_b + 1\; \; \Mod{\m}$.
If $q_a \ne q_b$, then without loss of generality, $q_b \ne 0$, 
so we have:
\begin{align}
  \n \pi_{q_b}(a) &= \n \pi_{q_b}(b)
    & & \! \! \! \Mod{\m \n} 
    & & \Comment{\eqref{eq:N2_P_5}} 
    \label{eq:N2_P_8}\\
\therefore \,
  \n r_a &=  \n \hat{r}_b
      & & \! \! \! \Mod{\m \n}
      & & \Comment{\eqref{eq:N2_P_1}, \eqref{eq:N2_P_3}, \eqref{eq:N2_P_8}} 
      \notag \\
  \therefore \,
   r_a &= r_a + 1
      & & \! \! \! \Mod{\m}
      & & \Comment{\eqref{eq:N2_P_2}, \eqref{eq:N2_P_7}},
      \notag
\end{align}
which is a contradiction,
so we must have $q_a = q_b$.
Thus $a = b$.

We have shown $\n \pi_l(a) = \n \pi_l(b)\; \; \Mod{\m\n}$ for all $l$
if and only if $a = b$.
Thus $a$ can be uniquely determined from
the $\n$-tuple
$(\n \pi_1(a), \n \pi_2(a), \,  \dots, \, \n \pi_{\n}(a))$.
This implies the existence of the claimed mapping.
\end{proof}
}

\begin{example} 
The following table illustrates Lemma~\ref{lem:N2_P}
for the case $\m = 4$ and $\n = 3$.
\begin{align*}
  \begin{array}{|c|c|c||c|c|c|}
    \hline
    a = \pi_3(a)  & \pi_2(a) & \pi_1(a) & 3 \pi_3(a) & 3 \pi_2(a) & 3 \pi_1(a)  \\[0pt] \hline 
    0   & 0   & 0   & 0   & 0   & 0 \\[0pt]
    1   & 1   & 1   & 3   & 3   & 3 \\[0pt]
    2   & 2   & 2   & 6   & 6   & 6 \\[0pt]
    3   & 3   & 3   & 9   & 9   & 9 \\[0pt] \hline
    4   & 4   & 5   & 0   & 0   & 3 \\[0pt]
    5   & 5   & 6   & 3   & 3   & 6 \\[0pt]
    6   & 6   & 7   & 6   & 6   & 9 \\[0pt]
    7   & 7   & 4   & 9   & 9   & 0 \\[0pt]\hline
    8   & 9   & 8   & 0   & 3   & 0 \\[0pt]
    9   & 10  & 9   & 3   & 6   & 3 \\[0pt]
    10  & 11  & 10  & 6   & 9   & 6 \\[0pt]
    11  & 8   & 11  & 9   & 0   & 9 \\[0pt] \hline
  \end{array}
\end{align*}  
For each $a \in \Z_{12}$,
the triple $(3 \pi_3(a), \, 3 \pi_2(a), \, 3 \pi_1(a)) \in \Z_{12}^3$
is distinct.
\label{ex:N2_P}
\end{example}

Lemma~\ref{lem:N2_P} will be used in the proof of Lemma~\ref{lem:N2_non} 
to show that the receiver $R_z$
can recover the message $z$ from the set of edge symbols $\ee{l}{i}$ where $l = 1,2,\dots,\n$
and $i = 1,2,\dots,\m+1$.

\begin{lemma}
  For each $\m \geq 2$ and $\n \geq 1$,
  network $\Network_2(\m,\n)$ is solvable over an alphabet of size $\m \n$. 
\label{lem:N2_non}
\end{lemma}

\newcommand{\ProofOflemNTwoNon}{
\begin{proof}[Proof of Lemma \ref{lem:N2_non}]

Let $\pi_1,\pi_2,\dots,\pi_{\n}$ and $\psi$ be the permutations and mapping, respectively, from Lemma~\ref{lem:N2_P}.
Define a code for network $\Network_2(\m,\n)$ over the ring $\Z_{\m \n}$
for each $l = 1,2,\dots,\n$ by:
\begin{align*}
  \ee{l}{0} &= \sum_{j = 1}^{\m+1} \xx{l}{j} \\     
  \ee{l}{i} & = \pi_l(z) + \sum_{\substack{j = 1 \\ j \ne i}}^{\m+1} \xx{l}{j} 
      & & (i = 1,2,\dots,\m+1) \\
  \ee{l}{} &= \pi_l(z) + \sum_{j = 1}^{\m+1} \xx{l}{j} .    
\end{align*}
For each $l = 1,2,\dots, \n$,
the receivers within each $\B^{(l)}(\m+1)$ block can recover their respective messages as follows:
\begin{align*}
  \RR{l}{0}: \ \ 
    & \pi_l^{-1} \left( \ee{l}{} - \ee{l}{0} \right) = z \\
  \RR{l}{i}:   \ \ 
     & \ee{l}{} - \ee{l}{i} = \xx{l}{i}     
    & &  (i = 1,2,\dots,\m+1).
\end{align*}
We have
\begin{align}
  \n  \sum_{i = 1}^{\m+1} \ee{l}{i} &= \n (\m + 1) \, \pi_l(z) + \m \n \sum_{j = 1}^{\m+1} \xx{l}{j} 
    & & (l = 1,2,\dots,\n)    \notag \\
    & = \n \pi_l(z)
    & & \Comment{$\m \n = 0$ mod $\m \n$} 
    \label{eq:N2_non_1}.
\end{align}
Receiver $R_z$ can recover $z$ from its inputs as follows:
\begin{align*}
  R_z:\ \  & \psi\left(\n \sum_{i = 1}^{\m+1} \ee{1}{i}, \, 
                       \n \sum_{i = 1}^{\m+1} \ee{2}{i}, \, 
                       \dots, \, 
                       \n  \sum_{i = 1}^{\m+1} \ee{\n}{i}   
                \right) \\
     & = \psi\left(\n \pi_1(z), \, \n\pi_2(z), \, \dots, \, \n \pi_{\n}(z) \right) 
        = z
      & & \Comment{\eqref{eq:N2_non_1} and Lemma~\ref{lem:N2_P}}.
\end{align*}  

Thus the network code described above is, in fact, a solution for $\Network_2(\m,\n)$.
\end{proof}
}

In the code given in the proof of Lemma~\ref{lem:N2_non},
if $\n = 1$, then $\pi_1$ and $\psi$ are identity permutations,
so the code is linear.
However if $\n > 1$, then
$\pi_1,\pi_2,\dots,\pi_{\n-1}$ are generally non-linear,
so the code is non-linear.  

\begin{lemma}
  For each $\m \geq 2$ and $\n \geq 1$,
  if network $\Network_2(\m,\n)$ is solvable over alphabet $\A$,
  then $\m$ and $\vert \A \vert$ are not relatively prime.
  \label{lem:N2_solv}
\end{lemma}

\newcommand{\ProofOflemNTwoSolv}{
\begin{proof}[Proof of Lemma \ref{lem:N2_solv}]
Assume $\Network_2(\m,\n)$ is solvable over $\A$.
For each $l = 1,2,\dots, \n$,
the block $\B^{(l)}(\m+1)$ together with source nodes
$S_z, \Ss{l}{1},\Ss{l}{2},\dots,\Ss{l}{\m+1}$
forms a copy of $\Network_0(\m+1)$,
so by Lemma~\ref{lem:N0_P}, 
the edge functions within block $\B^{(l)}(\m+1)$
must satisfy Property $P(\m+1)$.
Thus, for each $l$, 
there exists an Abelian group $\left(\A, \oplus_l \right)$, 
with identity $0_l \in \A$,
and permutations $\pipi{l}{0}, \pipi{l}{1}, \dots, \pipi{l}{\m+1}$ and 
$\sigsig{l}{0}, \sigsig{l}{1}, \dots, \sigsig{l}{\m+1}$ of $\A$,
such that
the edges carry the symbols:
\begin{align}
  \ee{l}{0} &= \sigsig{l}{0} \left( \bigoplus_{j = 1}^{\m+1} \pipi{l}{j}\left( \xx{l}{j} \right)  \right) 
     \notag \\
  \ee{l}{i} &= \sigsig{l}{i} \left( \pipi{l}{0}(z) \oplus_l \bigoplus_{\substack{j = 1 \\ j \ne i}}^{\m+1} \pipi{l}{j}\left( \xx{l}{j} \right)  \right)
    & & (i = 1,2,\dots,\m+1)
      \label{eq:N2_solv_0}  \\  
  \ee{l}{} &= \pipi{l}{0}(z) \oplus_i \bigoplus_{j = 1}^{\m+1} \pipi{l}{j}\left( \xx{l}{j} \right)  \notag,
\end{align}
where $\bigoplus$ in each of the previous three equations denotes $\oplus_l$.

Now suppose to the contrary that $\m$ and $\vert \A \vert$ are relatively prime.
Then by Cauchy's Theorem, 
for each group $(\A,\oplus_l)$
there are no non-identity elements 
whose order divides $\m$.
That is, for each $\oplus_l$ and each $a \in \A$, we have $\RepAdd{a}{\oplus_l}{\m} = 0_l$
if and only if $a = 0_l$.
So for each $l = 1,2,\dots,\n$
let $a, b \in \A$.
We have 
\begin{align*}
  \RepAdd{ a } { \oplus_l } { \m } = \RepAdd{ b } { \oplus_l } { \m } 
    & \Longleftrightarrow \RepAdd{ \left(a \ominus_l b \right) } { \oplus_l }{ \m } = 0_l
    & & \Comment{$\left(\A, \oplus_l \right)$ Abelian} \\
  & \Longleftrightarrow a =  b
    & & \Comment{$\GCD{\m}{\vert\A\vert}=1$}.
\end{align*}
Thus, for each $l$ the mapping
$a \longmapsto \RepAdd{ a }{ \oplus_l }{ \m }$ 
is injective on the finite set $\A$ and therefore is bijective,
and its inverse
$\phi_l: \, \A \to \A$
satisfies
\begin{align}
  \RepAdd{\phi_l(a)} {\oplus_l}{\m}  &= a 
  & & (l = 1,2,\dots,\n) 
  \label{eq:N2_solv_1} .
\end{align}

For each $a \in \A$ such that $a \ne 0_1$, let
\begin{align}
  f_l(a)  &= \pipi{l}{0} \left( \pipiinv{1}{0}(0_1) \right) \, \ominus_l \, \pipi{l}{0} \left( \pipiinv{1}{0}(a) \right)
    && (l = 2,\dots,\n)
      \label{eq:N2_solv_2},
\end{align}
and  define two collections of messages as follows:
\begin{align*}
  \xx{1}{j}   &= \pipiinv{1}{j}(\phi_1(a)))
    & & (j = 1,2,\dots, \m+1)  \\
   z       & = \pipiinv{1}{0}(0_1) \\
  \xx{l}{j}   &= \pipiinv{l}{j}(0_l) 
    & & \Stack{(l = 2,\dots,\n)}{(j = 1,2,\dots, \m+1)} \\
    \\
  \xxh{1}{j} &= \pipiinv{1}{j}(0_1) 
    & & (j = 1,2,\dots, \m+1)\\
   \hat{z}      &=  \pipiinv{1}{0}(a) \\
  \xxh{l}{j} &= \pipiinv{l}{j} \left(  \phi_l( f_l(a) )  \right)
      & & \Stack{(l = 2,\dots,\n)}{(j = 1,2,\dots, \m+1).} 
\end{align*}
Since $a \ne 0_1$
and $\pipi{1}{0}$ is bijective, it follows that $z \ne \hat{z}$.
By Property $P(\m + 1)$
and \eqref{eq:N2_solv_0},
for each $i = 1,2,\dots, \m+1$ we have:
\begin{align*}
  \ee{1}{i} &= \sigsig{1}{i} \left(  \RepAdd{\phi_1(a)}{\oplus_1}{\m}  \right) = \sigsig{1}{i}(a)
        && \Comment{\eqref{eq:N2_solv_1}} \\
  \ee{l}{i} &= \sigsig{l}{i} \left( \pipi{l}{0}\left( \pipiinv{1}{0}(0_1) \right) \right)
    & & (l = 2, \dots, \n)
\end{align*}
for the messages $\xx{l}{j},z$, and
\begin{align*}
  \ee{1}{i} &= \sigsig{1}{i}\left(  a  \right) \\
  \ee{l}{i} &= \sigsig{l}{i} \left( \pipi{l}{0} \left( \pipiinv{1}{0}(a) \right) \oplus_l \RepAdd{ \phi_l(f_l(a)) }{\oplus_l}{\m}  \right)
        & & (l = 2, \dots, \n) \\
      &= \sigsig{l}{i} \left( \pipi{l}{0} \left( \pipiinv{1}{0}(a) \right)\oplus_l f_l(a)  \right)
        & & \Comment{\eqref{eq:N2_solv_1}} \\
      &= \sigsig{l}{i}\left( \pipi{l}{0} \left( \pipiinv{1}{0}(0_1) \right) \right) 
        & & \Comment{\eqref{eq:N2_solv_2}}.
\end{align*}
for the messages $\xxh{l}{j},\hat{z}$.
For both collections of messages,
the edge symbols $\ee{l}{i}$ are the same for all
$l = 1,2,\dots,\n$ and $i = 1,2,\dots,\m+1$,
and therefore the decoded value $z$ at $R_z$ must be the same.
However, this contradicts the fact that $z \ne \hat{z}$.
\end{proof}
}

Lemmas~\ref{lem:N2_non} and \ref{lem:N2_solv} together provide a partial characterization
of the alphabet sizes over which $\Network_2(\m,\n)$ is solvable.
However, these conditions are sufficient for showing our main results.

\subsection{Linear solvability conditions of \texorpdfstring{$\Network_2(\m,\n)$}{N2(m,w)}} \label{ssec:N2_lin}

Lemma~\ref{lem:N2_lin} characterizes a necessary and sufficient condition
for the scalar linear solvability of $\Network_2(\m,\n)$ over \niceRModules.

\begin{lemma}
  Let $\m \geq 2$ and $\n \geq 1$,
  and let $G$ be a \niceRModule.
  Then network $\Network_2(\m,\n)$ is scalar linear solvable over $G$
  if and only if $\Char{R}$ divides $\m$.
  \label{lem:N2_lin}
\end{lemma}

\newcommand{\ProofOflemNTwoLin}{
\begin{proof}[Proof of Lemma \ref{lem:N2_lin}]
For any ring $R$ with multiplicative identity $1_R$,
the characteristic of $R$ divides $\m$ if and only if $\m = \m \, 1_R = 0_R$,
so it suffices to show that
for each $\m, \n$ and each \niceRModule{} $G$, 
network $\Network_2(\m,\n)$ is scalar linear solvable
over $G$ if and only if $\m = 0_R$.

Assume network $\Network_2(\m,\n)$ is scalar linear solvable 
over \niceRModule{} $G$.
The messages are drawn from $G$, and
there exist $\ccc{l}{i}{j}, \ccc{l}{\?}{j} \in R$,
such that for each $l = 1,2,\dots,\n$, 
the edge symbols can be written as:
\begin{align}
\ee{l}{0} &= \bigoplus_{j = 1}^{\m+1} \left( \ccc{l}{0}{j} \act \xx{l}{j} \right)
  \label{eq:N2_lin_1} \\  
\ee{l}{i} &= \left( \ccc{l}{i}{0} \act  z \right)
    \oplus \bigoplus_{\substack{j = 1\\ j\ne i}}^{\m+1} \left( \ccc{l}{i}{j} \act \xx{l}{j} \right)
  & &  (i = 1,2, \dots,\m+1) 
  \label{eq:N2_lin_2} \\
\ee{l}{} &= \left( \ccc{l}{\?}{0} \act  z\right) \oplus \bigoplus_{j = 1}^{\m+1} \left( \ccc{l}{\?}{j} \act \xx{l}{j} \right)
  \label{eq:N2_lin_3}
\end{align}
and there exist $\ddd{l}{i}{e}, \ddd{l}{\?}{i}, \ddd{l}{z}{i} \in R$,
such that each receiver can linearly recover its respective message
from its received edge symbols by:
\begin{align}
\RR{l}{0} : \ \  & 
  \; \; z  \ = \left( \ddd{l}{0}{e} \act \ee{l}{} \right) \oplus \left( \ddd{l}{\?}{0} \act \ee{l}{0}\right)
    & & (l = 1,2,\dots, \n)
    \label{eq:N2_lin_4} \\    
\RR{l}{i} : \ \ & 
  \xx{l}{i}  = \left( \ddd{l}{i}{e} \act \ee{l}{} \right) \oplus \left( \ddd{l}{\?}{i} \act \ee{l}{i} \right)
    & & \Stack{(l = 1,2,\dots, \n)}{(i = 1, 2,\dots,\m+1)}
    \label{eq:N2_lin_5} \\
R_z : \ \   & 
  \; \; z \  =  \bigoplus_{l = 1}^{\n} \bigoplus_{i = 1}^{\m+1} \left( \ddd{l}{z}{i} \act \ee{l}{i} \right)
    \label{eq:N2_lin_6} .
\end{align}
For each $l = 1,2,\dots,\n$,
the block $\B^{(l)}(\m+1)$ together with source nodes
$S_z, \Ss{l}{1},\Ss{l}{2},\dots,\Ss{l}{\m+1}$
forms a copy of $\Network_0(\m+1)$,
so by Lemma~\ref{lem:N0_lin}
and \eqref{eq:N2_lin_1} -- \eqref{eq:N2_lin_5},
each $\ccc{l}{\?}{i}$ and each $\ddd{l}{\?}{i}$
is invertible in $R$,
and
\begin{align}
 \ccc{l}{i}{j} = - \dddinv{l}{\?}{i} \! \ddd{l}{i}{e} \, \ccc{l}{\?}{j}
      & & \Stack{(l = 1,2,\dots, \n)}{(i,j = 0,1, \dots, \m+1 \text{ and } j \ne i) .}
      \label{eq:N2_lin_7}
\end{align}
Equating message components at $R_z$ yields:
\begin{align}
1_R &= \sum_{l = 1}^{\n} \sum_{i = 1}^{\m+1} \ddd{l}{z}{i} \,  \ccc{l}{i}{0}
    & & 
    & &\Comment{\eqref{eq:N2_lin_2}, \eqref{eq:N2_lin_6}} \notag \\
    & = - \sum_{l = 1}^{\n} \sum_{i = 1}^{\m+1} \ddd{l}{z}{i} \,  \dddinv{l}{\?}{i} \! \ddd{l}{i}{e} \, \ccc{l}{\?}{0}
    & & 
    & & \Comment{\eqref{eq:N2_lin_7}}
    \label{eq:N2_lin_8}
\end{align}
and for each $l = 1,2,\dots,\n$,
\begin{align}
  0_R &=  \sum_{\substack{i = 1 \\ i \ne j}}^{\m+1} \ddd{l}{z}{i} \, \ccc{l}{i}{j}
    & &  (j = 1,2,\dots, \m+1)
    & &\Comment{\eqref{eq:N2_lin_2}, \eqref{eq:N2_lin_6}}
    \notag \\
  &= - \left( \sum_{\substack{i = 1 \\ i \ne j}}^{\m+1} \ddd{l}{z}{i} \, \dddinv{l}{\?}{i} \! \ddd{l}{i}{e} \right) \, \ccc{l}{\?}{j}
    & &  (j = 1,2,\dots, \m+1)
    & &\Comment{\eqref{eq:N2_lin_7}}.
    \label{eq:N2_lin_07}
\end{align}
For each $l=1,2,\dots,\n$,
by multiplying \eqref{eq:N2_lin_07} by
$\cccinv{l}{\?}{j} \! \ccc{l}{\?}{0}$,
we have
\begin{align}
0_R &=  \sum_{\substack{i = 1 \\ i \ne j}}^{\m+1} \ddd{l}{z}{i} \, \dddinv{l}{\?}{i} \! \ddd{l}{i}{e} \, \ccc{l}{\?}{0}
    & &  (j = 1,2,\dots, \m+1)
    & & \Comment{\eqref{eq:N2_lin_07}}
    \notag
\end{align}
and by summing over $j = 1,2,\dots,\m+1$ we have
\begin{align}
0_R 
  &= \sum_{j = 1}^{\m+1} \sum_{\substack{i = 1 \\ i \ne j}}^{\m+1} \ddd{l}{z}{i} \, \dddinv{l}{\?}{i} \! \ddd{l}{i}{e} \, \ccc{l}{\?}{0}
       \notag \\
    &= \m  \sum_{i=1}^{\m+1} \ddd{l}{z}{i} \, \dddinv{l}{\?}{i} \! \ddd{l}{i}{e} \, \ccc{l}{\?}{0}
    \label{eq:N2_lin_9} .
\end{align}

By summing \eqref{eq:N2_lin_9}
over $l = 1,2,\dots,\n$, 
we have
\begin{align*}
  0_R &= \m \sum_{i=1}^{\n} \sum_{i=1}^{\m+1} \ddd{l}{z}{i} \, \dddinv{l}{\?}{i} \! \ddd{l}{i}{e} \, \ccc{l}{\?}{0}
    & & \Comment{\eqref{eq:N2_lin_9}}  \\
  \therefore 0_R  &=  \m 
      & & \Comment{\eqref{eq:N2_lin_8}}.
\end{align*}

To prove the converse,
let $G$ be a \niceRModule{}
such that $\m \, 1_R = 0_R$.
Define a scalar linear code over $G$,
for each $l = 1,2,\dots,\n$, by: 
\begin{align*}
\ee{l}{0} &= \bigoplus_{j = 1}^{\m+1}  \xx{l}{j} \\ 
\ee{l}{i} &=  z \oplus \bigoplus_{\substack{j=1\\ j \ne i}}^{\m+1}  \xx{l}{j} 
  & & (i = 1,2, \dots,\m+1) \\
\ee{l}{} &=   z \oplus \bigoplus_{j = 1}^{\m+1} \xx{l}{j} .
\end{align*}

For each $l = 1,2,\dots, \n$,
the receivers within each $\B^{(l)}(\m+1)$ block can linearly recover their respective messages as follows:
\begin{align*}
  \RR{l}{0}: \ \
    & \ee{l}{} \ominus \ee{l}{0} = z\\
  \RR{l}{i}: \ \ 
    & \ee{l}{} \ominus \ee{l}{i} = \xx{l}{i}
      & & (i = 1,2,\dots,\m+1).
\end{align*}
Receiver $R_z$ can linearly recover $z$ as follows:
\begin{align*}
  R_z: \ \ 
    & \bigoplus_{i = 1}^{\m+1} \ee{1}{i} =  z  \oplus  (\m \, z) \oplus \left( \m  \bigoplus_{j = 1}^{\m+1} \xx{1}{j} \right)
      = z 
    & & \Comment{$\m  = 0_R$}.
\end{align*}
Thus the code is a scalar linear solution for $\Network_2(\m,\n)$.
\end{proof}
}

\noindent
By Lemma~\ref{lem:N2_non},
for every $\m,\n \geq 2,$
the network $\Network_2(\m,\n)$ is solvable over the ring $\Z_{\m \n}$,
but $\Char{\Z_{\m\n}} = \NDiv{\m \n}{\m} $
so by Lemma~\ref{lem:N2_lin}, the solution is necessarily non-linear.

\subsection{Capacity and linear capacity of \texorpdfstring{$\Network_2(\m,\n)$}{N2(m,w)}} \label{ssec:N2_cap}

The following lemma provides a partial characterization of the linear capacity of $\Network_2(\m,\n)$ over finite-field alphabets.

\begin{lemma}
  For each $\m \geq 2$ and $\n \geq 1$,
  network $\Network_2(\m,\n)$ has
  \begin{itemize}
  \itemsep0em 
    \item[(a)] capacity equal to $1$,
    \item[(b)] linear capacity equal to $1$ for any finite-field alphabet whose characteristic divides $\m$,
    \item[(c)] linear capacity upper bounded by
      $1 - \frac{1}{2 \m \n + 2 \n + 1}$
      for any finite-field alphabet whose characteristic does not divide $\m$.
  \end{itemize}
\label{lem:N2_cap}
\end{lemma}

\newcommand{\ProofOflemNTwoCap}{
\begin{proof}[Proof of Lemma \ref{lem:N2_cap}]
Since a scalar linear solution over a finite-field alphabet is a special case of a scalar linear solution over a \niceRModule,
by Lemma~\ref{lem:N2_lin}, 
$\Network_2(\m,\n)$ is scalar linear solvable over any finite-field alphabet whose characteristic divides $\m$,
so the linear capacity for such finite-field alphabets is at least $1$.
By Lemma~\ref{lem:N0_cap}, 
network $\Network_0(\m+1)$ has capacity equal to $1$,
and 
the block $\B^{(1)}(\m+1)$ together with the source nodes
$S_z, \Ss{1}{1},\Ss{1}{2},\dots,\Ss{1}{\m+1}$
forms a copy of $\Network_0(\m+1)$,
so the capacity of $\Network_2(\m,\n)$ is at most $1$.
Thus both the capacity of $\Network_2(\m,\n)$
and its linear capacity over any finite-field alphabet whose characteristic divides $\m$
are $1$. 

To prove part (c),
consider a $(k,n)$ fractional linear solution for $\Network_2(\m,\n)$ 
over a finite field $\F$ whose characteristic does not divide $\m$.
Since $\NDiv{\Char{\F}}{\m}$, the integer $m$ is invertible in $\F$.

We have $\xx{l}{j}, z \in \F^k$
and $\ee{l}{i},\ee{l}{} \in \F^n$,
with $n \geq k$, since the capacity is one.
There exist $n \times k$ coding matrices 
$\MMM{l}{\?}{j}, \, \MMM{l}{i}{j}$ 
over $\F$,
such that
for each $l = 1,2,\dots,\n$ 
the edge vectors can be written as:
\begin{align}
  \ee{l}{0} &= \sum_{j = 1}^{\m+1} \MMM{l}{0}{j} \, \xx{l}{j}
    \notag \\
  \ee{l}{i} &= \MMM{l}{i}{0} \, z + \sum_{\substack{j=1 \\ j \ne i}}^{\m+1} \MMM{l}{i}{j}\, \xx{l}{j} 
    & & (i = 1,2,\dots, \m+1)
    \label{eq:N2_cap_1}\\
  \ee{l}{} &= \MMM{l}{\?}{0} \, z + \sum_{j = 1}^{\m+1} \MMM{l}{\?}{j} \, \xx{l}{j} 
    \label{eq:N2_cap_2} 
\end{align}
and there exist $k \times n$ decoding matrices 
$\DDD{l}{i}{e}$ and $\DDD{l}{\?}{i}$ over $\F$,
such that 
for each $l = 1,2,\dots,\n$
the message $\xx{l}{i}$ can be linearly decoded 
at $\RR{l}{i}$ from the $n$-vectors $\ee{l}{i}$ and $\ee{l}{}$ by:
\begin{align} %
  \RR{l}{i}: \ \  
  \xx{l}{i} &= \DDD{l}{i}{e} \, \ee{l}{} + \DDD{l}{\?}{i} \, \ee{l}{i}   
    & & (i = 1,2,\dots, \m+1).
    \label{eq:N2_cap_3}
\end{align}
Since receiver $R_z$ linearly recovers $z$ from 
its incoming edge vectors,
we have
\begin{align}
  & \left\{ \ee{l}{i} \; : \; 
  \Stack{ l = 1,2,\dots, \n }{ i = 1,2,\dots, \m+1} 
  \right\}
    \, \longrightarrow \, z 
  \label{eq:N2_cap_4} .
\end{align}

For each $l = 1,2,\dots,\n$ and $i = 1,2,\dots,\m+1$,
if we set $\xx{l}{i} = 0$ in \eqref{eq:N2_cap_3},
then, since $\ee{l}{i}$ does not depend on $\xx{l}{i}$,
we get the following relationship 
among the remaining messages:
\begin{align}
  0 &=  \DDD{l}{i}{e} \, \left( \MMM{l}{\?}{0} \, z 
      + \sum_{\substack{ j = 1 \\ j \ne i}}^{\m+1} \MMM{l}{\?}{j} \, \xx{l}{j}  \right)
    + \DDD{l}{\?}{i} \, \ee{l}{i}
    & & \Comment{\eqref{eq:N2_cap_1}, \eqref{eq:N2_cap_2}, \eqref{eq:N2_cap_3}} 
    \label{eq:N2_cap_5}
\end{align}
and thus
\begin{align}
  &\ee{l}{i} \, \longrightarrow \, \DDD{l}{i}{e} \, \left( \MMM{l}{\?}{0} \, z 
      + \sum_{\substack{ j = 1 \\ j \ne i}}^{\m+1} \MMM{l}{\?}{j} \, \xx{l}{j}  \right)
    & &  \Stack{(l = 1,2,\dots, \n)}{(i = 1,2,\dots, \m+1)}
    & & \Comment{\eqref{eq:N2_cap_5}}
    \label{eq:N2_cap_6}  .
\end{align}

For each $l = 1,2,\dots,\n$
and $i = 1,2,\dots,\m+1$,
let $\QQQ{l}{i}{e}$ be the matrix $Q$ in Lemma~\ref{lem:mat_3}
corresponding to when $\DDD{l}{i}{e}$ is the matrix $A$
in Lemma~\ref{lem:mat_3}.

For each $l = 1,2,\dots,\n$, let $L^{(l)}$ be the
following list of $2 (\m + 1)$ vector
functions of \\ $z,\xx{l}{1},\xx{l}{2},\dots,\xx{l}{\m+1}$:
\begin{align*}
  &\QQQ{l}{i}{e} \, \left( \MMM{l}{\?}{0} \, z 
      + \sum_{\substack{ j = 1 \\ j \ne i}}^{\m+1} \MMM{l}{\?}{j} \, \xx{l}{j}  \right)
    && (i = 1,2,\dots,\m+1) \\
  &\ee{l}{i}  
    && (i = 1,2,\dots,\m+1) .
\end{align*}

For each $l = 1,2,\dots,\n$ we have
\begin{align}
  & L^{(l)} \, \longrightarrow \, 
    \DDD{l}{i}{e} \, \left( \MMM{l}{\?}{0} \, z 
      + \sum_{\substack{ j = 1 \\ j \ne l}}^{\m+1} \MMM{l}{\?}{j} \, \xx{l}{j}  \right)
    & & (i = 1,2,\dots, \m+1)
    & & \Comment{\eqref{eq:N2_cap_6}} 
    \label{eq:N2_cap_7} \\
  & L^{(l)} \, \longrightarrow \, 
    \MMM{l}{\?}{0} \, z + \sum_{\substack{ j = 1 \\ j \ne i}}^{\m+1} \MMM{l}{\?}{j} \, \xx{l}{j}
    & & (i = 1,2,\dots, \m+1)
    && \Comment{Lemma~\ref{lem:mat_3}, \eqref{eq:N2_cap_7}} ,
    \label{eq:N2_cap_8}
\end{align}
and 
\begin{align}
  &z, 
   \ \ \ \
  \left\{ \MMM{l}{\?}{0} \, z + \sum_{\substack{ j = 1 \\ j \ne i}}^{\m+1} \MMM{l}{\?}{j} \, \xx{l}{j}
    \; : \; i = 1,2,\dots,\m+1 \right\}
    \notag \\
  & \; \; \; \; \longrightarrow 
    \sum_{i=1}^{\m+1} \left( \MMM{l}{\?}{0} \, z + \sum_{\substack{ j = 1 \\ j \ne i}}^{\m+1} \MMM{l}{\?}{j} \, \xx{l}{j} \right) 
    - \MMM{l}{\?}{0} \, z  
    \notag \\
  & \; \; \; \; =  (\m+1) \, \MMM{l}{\?}{0} \, z + \m \, \sum_{j=1}^{\m+1} \MMM{l}{\?}{j} \, \xx{l}{j}
    - \MMM{l}{\?}{0} \, z 
    \notag \\
  & \; \; \; \;  = \m \, \ee{l}{} \, \longrightarrow \, \ee{l}{} 
    & & \Comment{\eqref{eq:N2_cap_2} and $\NDiv{\Char{\F}}{\m}$} .
    \label{eq:N2_cap_9}
\end{align}
We also have
\begin{align}
  L^{(1)} , \dots , L^{(\n)}  \longrightarrow &  z
    & & & & \Comment{\eqref{eq:N2_cap_4}}
    \label{eq:N2_cap_10}
\end{align}
and for each $l = 1,2,\dots, \n$
\begin{align}
  L^{(l)}, \, z \longrightarrow & \ee{l}{}  
    & &
    & &\Comment{\eqref{eq:N2_cap_8}, \eqref{eq:N2_cap_9}} 
    \label{eq:N2_cap_11} \\
  L^{(l)}, z \longrightarrow & \xx{l}{i}
    & &(i = 1,2,\dots, \m+1)
    & & \Comment{\eqref{eq:N2_cap_3}, \eqref{eq:N2_cap_11}} .
    \label{eq:N2_cap_12}
\end{align}
Thus
\begin{align}
  L^{(1)} , \dots , L^{(\n)} \, \longrightarrow \, z, \, \left\{\xx{l}{i} 
    \; : \; 
     \Stack{l = 1,2,\dots, \n}{i = 1,2,\dots, \m+1}
    \right\}
    & & \Comment{\eqref{eq:N2_cap_10}, \eqref{eq:N2_cap_12}}
    \label{eq:N2_cap_13}.
\end{align}
We will now bound the number of independent entries in each list $L^{(l)}$. 

By equating message components in equation \eqref{eq:N2_cap_3}, 
we have:
\begin{align}
 I_k = &\DDD{l}{i}{e} \, \MMM{l}{\?}{i} 
  & & 
    \Stack{(l = 1,2,\dots, \n)}{(i = 1,2,\dots, \m+1)} 
  & & \Comment{\eqref{eq:N2_cap_1}, \eqref{eq:N2_cap_2}, \eqref{eq:N2_cap_3}}  
  \label{eq:N2_cap_14}
\end{align}
Since each $\DDD{l}{i}{e}$ is $k \times n$ and $k \leq n$,
the rank of each matrix is at most $k$, but we also have
\begin{align*}
 \Rank{\DDD{l}{i}{e}} &\geq \Rank{\DDD{l}{i}{e} \, \MMM{l}{\?}{i}} = \Rank{I_k}  = k
  & & \Comment{\eqref{eq:mat_1_2}, \eqref{eq:N2_cap_14}},
\end{align*}
and so $\Rank{\DDD{l}{i}{e}} = k$. 
By Lemma~\eqref{lem:mat_3}, this implies $\Rank{\QQQ{l}{i}{e}} = n - k$.
Therefore each vector function
$$\QQQ{l}{i}{e} \, \left( \MMM{l}{\?}{0} \, z 
      + \sum_{\substack{ j = 1 \\ j \ne i}}^{\m+1} \MMM{l}{\?}{j} \, \xx{l}{j}  \right)
\ \ \ \ \   
    \Stack{(l = 1,2,\dots, \n)} {(i = 1,2,\dots, \m+1)} $$
in the list $L^{(l)}$ has dimension $n - k$.

If we view the messages vectors as random variables,
each of whose $k$ components are independent and uniformly distributed over the field $\F$,
then we have the following entropy 
(using logarithms base $\vert \F \vert$)
upper bounds:
\begin{align}
  H\left(
    \QQQ{l}{i}{e} \, \left( \MMM{l}{\?}{0} \, z 
      + \sum_{\substack{ j = 1 \\ j \ne i}}^{\m+1} \MMM{l}{\?}{j} \, \xx{l}{j}  \right)
    \; : \;  
    \Stack{l = 1,2,\dots, \n}{i = 1,2,\dots, \m+1}
  \right)
  & \leq \n  (\m + 1) \, (n - k) 
\label{eq:N2_cap_15}\\
\notag \\
  H\left(
    \ee{l}{i} \; : \; 
    \Stack{l = 1,2,\dots, \n}{i = 1,2,\dots, \m+1}
  \right)
  & \leq \n (\m + 1) \, n
\label{eq:N2_cap_16}.  
\end{align}

Therefore, the entropy of all of the vector functions in the list of lists $L^{(1)}, \dots , L^{(\n)}$
is bounded by summing the bounds in \eqref{eq:N2_cap_15} and \eqref{eq:N2_cap_16}:
\begin{align}
  H\left(L^{(1)} , \dots , L^{(\n)} \right) & \leq  \n (\m + 1) \, n - \n (\m + 1) \, k
    & & \Comment{\eqref{eq:N2_cap_15}, \eqref{eq:N2_cap_16}} .
    \label{eq:N2_cap_17} 
\end{align}
But then we have:
\begin{align*}
  (\n  (\m + 1) + 1) \, k 
    &= H\left( z, \; \left\{ \xx{l}{i} 
      \; : \;
      \Stack{l = 1,2,\dots,\n}{i = 1,2,\dots,\m+1}
     \right\} \right)
    & & \Comment{$z, \xx{l}{i} \in \F^k$} \\
  & \leq H\left(L^{(1)} , \dots , L^{(\n)} \right)
    & & \Comment{\eqref{eq:N2_cap_13}} \\
  & \leq  2  \n (\m + 1) \, n - \n (\m + 1) \, k
    & & \Comment{\eqref{eq:N2_cap_17}} \\
  \therefore \frac{k}{n} &\leq \frac{ 2 \n (\m + 1) }{2 \n (\m + 1) + 1}  .
\end{align*}
Thus the linear capacity of $\Network_2(\m,\n)$ for finite-field alphabets whose characteristic does not divide $\m$
is upper bounded by
$$1 - \frac{1}{2 \m \n + 2 \n + 1}.$$
\end{proof}
}

Improving these upper-bounds on the linear capacities
and/or finding codes at these rates are left as open problems.
The problems appear to be non-trivial,
and such improvements are unrelated to the main results of this paper.

\clearpage

\section{The network \texorpdfstring{$\Network_3(\m_1,\m_2)$}{N3(m1,m2)}} \label{sec:N3}
\psfrag{N0(a)}{\large$\B^{(1)}(\m_1)$}
\psfrag{N0(b)}{\large$\B^{(2)}(\m_2)$}

\psfrag{x11}{$\xx{1}{1}$}
\psfrag{x12}{$\xx{1}{2}$}
\psfrag{x1a}{$\xx{1}{\m_1}$}

\psfrag{x21}{$\xx{2}{1}$}
\psfrag{x22}{$\xx{2}{2}$}
\psfrag{x2b}{$\xx{2}{\m_2}$}

\psfrag{S11}{\tiny$\Ss{1}{1}$}
\psfrag{S12}{\tiny$\Ss{1}{2}$}
\psfrag{S1a}{\tiny$\Ss{1}{\m_1}$}

\psfrag{S21}{\tiny$\Ss{2}{1}$}
\psfrag{S22}{\tiny$\Ss{2}{2}$}
\psfrag{S2b}{\tiny$\Ss{2}{\m_2}$}

\psfrag{e11}{$\ee{1}{1}$}
\psfrag{e12}{$\ee{1}{2}$}
\psfrag{e1a}{$\ee{1}{\m_1}$}
\psfrag{e1z}{$\ee{1}{0}$}

\psfrag{e21}{$\ee{2}{1}$}
\psfrag{e22}{$\ee{2}{2}$}
\psfrag{e2b}{$\ee{2}{\m_2}$}
\psfrag{e2z}{$\ee{2}{0}$}
\psfrag{Rz}{$R_z$}
\psfrag{z}{$z$}

\begin{figure*}[h]
\begin{center}
\leavevmode
\hbox{\epsfxsize=.68\textwidth\epsffile{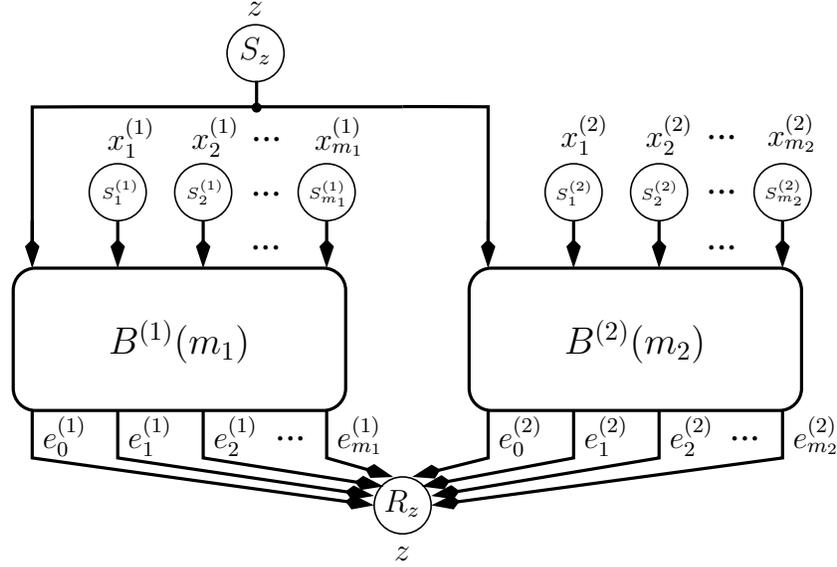}}
\end{center}
\caption{The network $\Network_3(\m_1,\m_2)$ is constructed from $\B(\m_1)$ and $\B(\m_2)$ blocks
  together with
  $\m_1 + \m_2 + 1$ source nodes
  and an additional receiver $R_z$. 
  The blocks are denoted $\B^{(1)}(\m_1)$ and $\B^{(2)}(\m_2)$ respectively,
  and for each $l = 1,2,$ 
  the nodes and edge symbols in $\B^{(l)}(\m_l)$ are denoted with a superscript $l$.
  Each $\B^{(l)}(\m_l)$ block has inputs from source nodes $\Ss{l}{1},\Ss{l}{2},\dots,\Ss{l}{\m_l}$,
  which generate messages $\xx{l}{1}, \xx{l}{2}, \dots, \xx{l}{\m_l}$.
  The shared message $z$ is generated by source node $S_z$ and
  is the $0$th input to $\B^{(l)}(\m_l)$.
  The additional receiver $R_z$ receives all of the output edges of $\B^{(1)}(\m_1)$ and $\B^{(2)}(\m_2)$
  and demands the shared message $z$.
}
\label{fig:N3}
\end{figure*}

For each $\m_1,\m_2 \geq 2,$ network $\Network_3(\m_1,\m_2)$ is
defined in Figure~\ref{fig:N3}.
We note that $\Network_2(m,2)$ and $\Network_3(m+1,m+1)$ have similar structure,
with the exception of the disconnected output edge of each $\B(m+1)$ in $\Network_2(m,2)$.
This disconnected edge causes the difference in solvability properties of the two networks.
Corollary~\ref{cor:N3_non} and Lemmas~\ref{lem:N3_solv}, \ref{lem:N3_lin}, and \ref{lem:N3_cap} 
demonstrate that network $\Network_3(\m_1,\m_2)$ is:
\begin{enumerate} 
\itemsep0em 
  \item non-linear solvable over an alphabet of size $t \m_1^{\alpha+1}$,
  if $\alpha\ge 1$, $\m_2 = s \m_1^{\alpha}$,
  and $s$ and $t$ are relatively prime to $\m_1$,
  
  \item solvable over alphabet $\A$ only if $|\A|$ is relatively prime to $\m_1$ or $|\A|$ does not divide $\m_2,$

  \item scalar linear solvable over \niceRModule{} $G$ if and only if $\GGGCD{\Char{R} \!}{\m_1}{\m_2} = 1$,
  
  \item asymptotically linear solvable over finite field $\F$ if and only if $\Char{\F}$ is relatively prime to $\m_1$ or $\m_2$.
\end{enumerate}

\begin{remark}
For each $\m_1,\m_2 \geq 2,$
the network $\Network_3(\m_1,\m_2)$ has 
$\m_1 + \m_2 + 1$ source nodes,
$2 (\m_1 + \m_2 + 4)$ intermediate nodes,
and $\m_1 + \m_2 + 3$ receiver nodes,
so the total number of nodes in $\Network_3(\m_1,\m_2)$ is
$4 \m_1 + 4\m_2 + 12$.
\label{rem:N3_nodes}
\end{remark}

\subsection{Solvability conditions of \texorpdfstring{$\Network_3(\m_1,\m_2)$}{N3(m1,m2)}} \label{ssec:N3_solv}

The following lemmas demonstrate that $\Network_3(\m_1,\m_2)$
is non-linear solvable when $\m_2 = s \m_1^{\alpha}$, $\alpha \ge 1$, 
and $s$ is relatively prime to $\m_1$.
Consider the ring alphabet $\Z_{\m_1^{\alpha+1}}$.
For every $a \in \Z_{\m_1^{\alpha+1}}$,
a receiver cannot uniquely determine a symbol $a$ in $\Z_{\m_1^{\alpha+1}}$
from the symbols $\m_1 a$ and $s \m_1^{\alpha} a$,
since $\m_1$ is not invertible in $\Z_{\m_1^{\alpha+1}}$.
For example, if a receiver receives 
$\m_1 a = s \m_1^{\alpha} a = 0$ in $\Z_{\m_1^{\alpha+1}}$, then
the symbol $a$ 
could be any element in the set $\{0,\m_1^{\alpha},2\m_1^{\alpha}, \dots, (\m_1-1)\m_1^{\alpha} \}$.
The following lemma describes a technique for recovering the value of $a$
via a decoding function $\psi$
from $\m_1  \pi_1(a)$ and $s \m_1^{\alpha}  \pi_2(a)$, 
where $\pi_1$ and $\pi_2$ are particular permutations of $\Z_{\m_1^{\alpha+1}}$.

\begin{lemma}
  Let $\m \geq 2$
  and $\alpha, s \geq 1$ 
  be integers such that $s$ is relatively prime to $\m$.
  Then there exist permutations
  $\pi_1$ and $\pi_2$
  of $\Z_{\m^{\alpha+1}}$
  and a mapping $\psi: \Z_{\m^{\alpha+1}}^2 \to \Z_{\m^{\alpha+1}}$
  such that for all $a \in \Z_{\m^{\alpha+1}}$,
  $$ \psi\left( \m \pi_1(a), \;  s \m^{\alpha}  \pi_2(a)\right) 
      = a .$$
  \vspace{-.5cm}
  \label{lem:N3_P}
\end{lemma}

\newcommand{\ProofOflemNThreeP}{
\begin{proof}[Proof of Lemma \ref{lem:N3_P}]
Define permutations $\pi_1,\pi_2$ of $\Z_{\m^{\alpha+1}}$ as follows.
For each $a \in \Z_{\m^{\alpha+1}}$, 
let $\sum_{i = 0}^{\alpha} \m^{i} a_i$ denote the base $m$ representation of $a$.
We define
\begin{align}
  \pi_1(a) &=   \m^{\alpha} a_{0} + \sum_{i = 1}^{\alpha} \m^{i-1} a_{i}
    \label{eq:N3_P_1} \\
  \pi_2(a) &= a =  \sum_{i = 0}^{\alpha} \m^{i} a_{i}
    \label{eq:N3_P_2} .
\end{align}
The (non-linear) permutation $\pi_1$ performs a right-cyclic shift of the base-$\m$ digits of $a$,
and
$\pi_2$ is the identity permutation.
For each $a \in \Z_{\m^{\alpha+1}}$,
we will show the mapping
$a \longmapsto (\m \pi_1(a), \; s \m^{\alpha} \pi_2(a))$
is injective.
For each $a,b \in \Z_{\m^{\alpha+1}}$,
suppose
\begin{align}
    \m  \pi_1(a) & = \m  \pi_1(b) & & \! \! \! \Mod {\m^{\alpha+1}}
      \label{eq:N3_P_3}\\
    s \m^{\alpha} \pi_2(a) & = s \m^{\alpha} \pi_2(b) & & \! \! \! \Mod {\m^{\alpha+1}}
      \label{eq:N3_P_4}
\end{align}
where $a = \sum_{i=0}^{\alpha} \m^{i} a_i$
and $b = \sum_{i=0}^{\alpha} \m^{i} b_i$.
Then we have
\begin{align} 
  \sum_{i = 1}^{\alpha} \m^i a_{i}
    &= \sum_{i = 1}^{\alpha} \m^i b_{i}
      & & \! \! \! \Mod {\m^{\alpha+1}}
      & & \Comment{\eqref{eq:N3_P_1}, \eqref{eq:N3_P_3}}   
    \notag\\
  \therefore \,
      a_{i} &= b_{i} 
       && \!  (i = 1,2,\dots, \alpha)
    & & \Comment{$0 \leq a_{i},b_{i} < \m $}
    \notag 
\end{align}
and
\begin{align}
  s \m^{\alpha}  a_{0} 
    &=  s \m^{\alpha}  b_{0}
    & & \! \! \! \Mod{\m^{\alpha+1}}
    & &  \Comment{\eqref{eq:N3_P_2}, \eqref{eq:N3_P_4}} 
    \notag \\
  \therefore \,
  \m^{\alpha}  a_{0} 
    &= \m^{\alpha} b_{0} 
    & & \! \! \! \Mod{\m^{\alpha+1}}
    & & \Comment{$\GCD{\m}{s} = 1$} 
    \notag \\
  \therefore \, a_{0} &= b_{0}
    & & & & \Comment{$0 \leq a_{0},b_{0} < \m$}. \notag
\end{align}
Thus $a = b$.

We have shown that
$\m \pi_1(a) = \m \pi_1(b)$ 
and $s \m^{\alpha} \pi_2(a)  =  s \m^{\alpha} \pi_2(b)$
if and only if $a = b$.
Thus $a$ can be uniquely determined from
$\m \pi_1(a)$ and $s \m^{\alpha} \pi_2(a)$.
This implies the existence of the claimed mapping.
\end{proof}
}

\begin{example}
The table below illustrates Lemma~\ref{lem:N3_P} for the case
$\m = 2$, $s = 3$, and $\alpha = 2$,
and permutations $\pi_1$ and $\pi_2$ of $\Z_8$.
{\small 
$$
  \begin{array}{|c|c||c|c|}
  \hline
    a = \pi_2(a) & \pi_1(a) & 12 \pi_2(a) & 2 \pi_1(a) \\[0pt] \hline 
      0  & 0  & 0    & 0  \\[0pt]
      1  & 4  & 4    & 0  \\[0pt]
      2  & 1  & 0    & 2  \\[0pt]
      3  & 5  & 4    & 2  \\[0pt]
      4  & 2  & 0    & 4  \\[0pt]
      5  & 6  & 4    & 4  \\[0pt]
      6  & 3  & 0    & 6  \\[0pt]
      7  & 7  & 4    & 6  \\[0pt] \hline
  \end{array}
$$}
For each $a \in \Z_{8}$, the pair $(2 \pi_1(a), \, 12 \pi_2(a)) \in \Z_{8}^2$ is distinct.
\label{ex:N3_P}
\end{example}

Lemma~\ref{lem:N3_P}
will be used in the proof of Lemma~\ref{lem:N3_non}
to show that
the receiver $R_z$ can recover the message $z$
from the set of edge symbols $\ee{l}{i}$,
where $l = 1,2$ and $i = 0,1,\dots,\m_l$.
\begin{lemma}
  Let $\m_1, \m_2 \geq 2$ and $\alpha, s \ge 1$ be integers
  such that $\m_2 = s \m_1^{\alpha}$ and
  $s$ is relatively prime to $\m_1$.
  Then network $\Network_3(\m_1, \m_2 )$ is solvable 
  over an alphabet of size $\m_1^{\alpha+1}$.
  \label{lem:N3_non}
\end{lemma}

\newcommand{\ProofOflemNThreeNon}{
\begin{proof}[Proof of Lemma \ref{lem:N3_non}]
Let $\pi_1, \pi_2$ and $\psi$
be the permutations and mapping, respectively, from\\ Lemma~\ref{lem:N3_P}.
Define a code for the network $\Network_3(\m_1,\m_2)$ over the ring $\Z_{\m_1^{\alpha+1}}$,
for each $l = 1,2$, by:
\begin{align*}
  \ee{l}{0} &= \sum_{j = 1}^{\m_l} \xx{l}{j} \\
  \ee{l}{i} &= \pi_l(z) + \sum_{\substack{j = 1\\
                            j \ne i}}^{\m_l} \xx{l}{j} 
          &&  (i =1,2,\dots,\m_l) \\
  \ee{l}{} &= \pi_l(z) + \sum_{j = 1}^{\m_l} \xx{l}{j}.
\end{align*}
For each $l = 1,2,$ the receivers within the block $\B^{(l)}(\m_l)$
can recover their respective messages as follows:
\begin{align*}
\RR{l}{0}: \ \  
  & \pi_l^{-1}\left(\ee{l}{} - \ee{l}{0} \right) = z \\
\RR{l}{i}: \ \  
  & \ee{l}{} - \ee{l}{i} = \xx{l}{i}
    & & (i = 1,2,\dots,\m_l) .
\end{align*}
For each $l = 1,2,$
we have
\begin{align}
-\m_l  \ee{l}{0} + \sum_{i=0}^{\m_l} \ee{l}{i} 
  &  = -\m_l  \sum_{j = 1}^{\m_l} \xx{l}{j} 
    + \m_l  \pi_l(z) 
    + \m_l  \sum_{j = 1}^{\m_l} \xx{l}{j} \notag \\
  &   = \m_l   \pi_l(z). \label{eq:N3_non_1}
\end{align}
The receiver $R_z$ can recover $z$ from its inputs as follows:
\begin{align*}
  &\psi\left( -\m_1  \ee{1}{0} + \sum_{i=0}^{\m_1} \ee{1}{i}, \; -\m_2 \ee{2}{0} + \sum_{i=0}^{\m_2} \ee{2}{i}     \right) \\
  &\ \ \ = \psi\left( \m_1 \pi_1(z), \; \m_2  \pi_2(z) \right)
    & & \Comment{\eqref{eq:N3_non_1}} \\
  &\ \ \ = \psi\left( \m_1 \pi_1(z), \; s \m_1^{\alpha} \pi_2(z) \right) = z
    & & \Comment{$\m_2 = s \m_1^{\alpha}$ and Lemma~\ref{lem:N3_P}}.
\end{align*}
Thus the network code described above is, in fact, a solution for $\Network_3(\m_1,\m_2)$.
\end{proof}
}

In the code given in the proof of Lemma~\ref{lem:N3_non},
the permutation $\pi_1$ is non-linear,
so the code is non-linear.

\begin{lemma}
  Let $\m_1,\m_2 \geq 2$.
  If network $\Network_3(\m_1,\m_2)$ is solvable over alphabet $\A$
  and $\vert \A \vert$ divides $\m_2$, then $\m_1$ and $\vert\A\vert$ are relatively prime.
\label{lem:N3_solv}
\end{lemma}

\newcommand{\ProofOflemNThreeSolv}{
\begin{proof}[Proof of Lemma \ref{lem:N3_solv}]
Assume $\Network_3(\m_1,\m_2)$ is solvable over $\A$.
For each $l =1,2$ the block $\B^{(l)}(\m_l)$ together with the source nodes
$S_z, \Ss{l}{1},\Ss{l}{2},\dots,\Ss{l}{\m_l}$
forms a copy of $\Network_0(\m_l)$,
so by Lemma~\ref{lem:N0_P}, 
the edge functions within $\B^{(1)}(\m_1)$ and $\B^{(2)}(\m_2)$ must satisfy 
Property $P(\m_1)$ and Property $P(\m_2)$, respectively.
Thus there exist Abelian groups $(\A,\oplus_1)$ and $(\A, \oplus_2)$ 
with identity elements $0_1$ and $0_2$
for the left-hand side and right-hand side of the network, respectively,
and permutations $\pipi{l}{0},\pipi{l}{1},\dots,\pipi{l}{\m_l}$
and $\sigsig{l}{0},\sigsig{l}{1},\dots,\sigsig{l}{\m_l}$ of $\A$,
such that for each $l = 1,2$ the edges carry the symbols:
\begin{align}
  \ee{l}{0} &= \sigsig{l}{0} \left( 
      \bigoplus_{ j = 1}^{\m_l} \pipi{l}{j} \left( \xx{l}{j} \right) 
        \right)
    \label{eq:N3_solv_1}\\
  \ee{l}{i} &= \sigsig{l}{i} \left( 
      \pipi{l}{0}(z) \oplus_l \bigoplus_{\substack{ j = 1 \\ j \ne i}}^{\m_l} \pipi{l}{j}\left(\xx{l}{j} \right)  
        \right) 
    &&( i = 1,2,\dots, \m_l)
    \label{eq:N3_solv_2}\\
  \ee{l}{}  &= 
      \pipi{l}{0}(z) \oplus_l \bigoplus_{ j = 1}^{\m_l} \pipi{l}{j} \left( \xx{l}{j} \right)
     \notag
\end{align}
where $\bigoplus$ in each of the previous three equations denotes $\oplus_l$.

Now suppose to the contrary
that $\m_1$ and $\vert\A\vert$ are not relatively prime
and $\vert\A\vert$ divides $\m_2$.
Then, since $(\A,\oplus_2)$ is a finite group,
for all $a \in \A$,
we have
\begin{align}
  \RepAdd{a}{\oplus_2}{\m_2} = 0_2 
    & & \Comment{$\Div{\vert \A \vert}{\m_2}$}.
    \label{eq:N3_solv_3}
\end{align}

Since $\m_1$ and $\vert\A\vert$ are not relatively prime,
$\m_1$ and $\vert\A\vert$ share a common factor $p$.
Since $\Div{p}{\vert\A\vert}$, by Cauchy's Theorem,
there exists $a \in \A\backslash\{0_1\}$
such that the order of $a$ is $p$,
and since $p$ divides $\m_1$
we have $\RepAdd{a}{\oplus_1}{\m_1} = 0_1$.
Define two collections of messages as follows:
\begin{align*}
  \xx{1}{j} &= \pipiinv{1}{j}(0_1) 
    & & ( j = 1,2,\dots, \m_1 ) \\
  \xx{2}{j} & = \pipiinv{2}{j} \left(  \pipi{2}{0}  \left(  \pipiinv{1}{0}(0_1) \right) \right) 
    & & ( j = 1,2,\dots, \m_2 ) \\
  z & = \pipiinv{1}{0}(0_1)
\end{align*}
\begin{align*}
  \xxh{1}{j} &= \pipiinv{1}{j}(a) 
    & & ( j = 1,2,\dots, \m_1 ) \\
  \xxh{2}{j} & = \pipiinv{2}{j} \left(  \pipi{2}{0} \left(   \pipiinv{1}{0}(a)  \right) \right)           
    & & ( j = 1,2,\dots, \m_2 ) \\
  \hat{z} & = \pipiinv{1}{0}(a)  .
\end{align*}
Since $a \ne 0_1$ and $\pipi{1}{0}$ is bijective,
it follows that $z \ne \hat{z}$.
By Properties $P(\m_1)$ and $P(\m_2)$
and \eqref{eq:N3_solv_1} and \eqref{eq:N3_solv_2},
we have
\begin{align*}
\ee{1}{i} &= \sigsig{1}{i} \left( \RepAdd{0_1}{\oplus_1}{\m_1} \right) = \sigsig{1}{i} ( 0_1 )
    & &  ( i = 0,1,\dots, \m_1)
    \\
\ee{2}{i} &= \sigsig{2}{i} \left( \RepAdd{ \pipi{2}{0}  \left(  \pipiinv{1}{0}(0_1) \right) }{\oplus_2}{\m_2} \right)
    & & (i = 0,1,\dots, \m_2) \\
    &= \sigsig{2}{i} \left( 0_2 \right)   
      & & \Comment{\eqref{eq:N3_solv_3}}
\end{align*}
for the messages $\xx{l}{j},z$, and
\begin{align*}
\ee{1}{i} &= \sigsig{1}{i} \left( \RepAdd{a}{\oplus_1}{\m_1} \right) = \sigsig{1}{i} ( 0_1 )
    & &  ( i = 0,1,\dots, \m_1)
    \\
\ee{2}{i} &= \sigsig{2}{i} \left( \RepAdd{ \pipi{2}{0}  \left(  \pipiinv{1}{0}(a) \right) }{\oplus_2}{\m_2} \right)
    & & (i = 0,1,\dots, \m_2)  \\
    &= \sigsig{2}{i} \left( 0_2 \right)   
    & & \Comment{\eqref{eq:N3_solv_3}}
\end{align*}
for the messages $\xxh{l}{j},\hat{z}$.
For both collections of messages, 
the edge symbols $\ee{1}{0},\ee{1}{1},\dots,\ee{1}{\m_1}$
and $\ee{2}{0},\ee{2}{1},\dots,\ee{2}{\m_2}$
are the same, and therefore the decoded value $z$ at $R_z$
must be the same.
However, this contradicts the fact that $z \ne \hat{z}$.
\end{proof}
}

Lemmas~\ref{lem:N3_non} and \ref{lem:N3_solv} together provide a partial characterization
of the alphabet sizes over which $\Network_2(\m,\n)$ is solvable.
However, these conditions are sufficient for showing our main results.

\subsection{Linear solvability conditions of \texorpdfstring{$\Network_3(\m_1,\m_2)$}{N3(m1,m2)}} \label{ssec:N3_lin}

The following lemma characterizes a necessary and sufficient condition 
for the scalar linear solvability
of $\Network_3(\m_1,\m_2)$ over \niceRModules.
\begin{lemma}
  Let $\m_1,\m_2 \geq 2$,
  and let $G$ be a \niceRModule.
  Then network $\Network_3(\m_1,\m_2)$ is scalar linear solvable over $G$
  if and only if $\GGGCD{\Char{R} \!}{\m_1}{\m_2} = 1$.
  \label{lem:N3_lin}
\end{lemma}

\newcommand{\ProofOflemNThreeLin}{
\begin{proof}[Proof of Lemma \ref{lem:N3_lin}]
For any integers $a,b,c \geq 1$,
we have $\GGGCD{a}{b}{c} = \GCD{\GCD{a}{b}}{c}$,
so by Lemma~\ref{lem:RingInv2}
$\GCD{\m_1}{\m_2}$
is invertible in $R$
if and only if $\GGGCD{\m_1}{\m_2}{\Char{R}} = 1$.
Thus it suffices to show that
for each $\m_1,\m_2$ and each \niceRModule{} $G$, 
network $\Network_3(\m_1,\m_2)$ is scalar linear solvable over $G$
if and only if $\GCD{\m_1}{\m_2}$ is invertible in $R$.

Assume network $\Network_3(\m_1,\m_2)$ is scalar linear solvable 
over \niceRModule{} $G$.
The messages are drawn from $G$, and
there exist $\ccc{l}{i}{j}, \ccc{l}{\?}{j} \in R$,
such that for each $l = 1,2$ the edge symbols can be written as:
\begin{align}
  \ee{l}{0} &= \bigoplus_{j = 1}^{\m_l} \left(\ccc{l}{0}{j} \act \xx{l}{j} \right)
    \label{eq:N3_lin_1} \\
  \ee{l}{i} &= \left( \ccc{l}{i}{0} \act z  \right)
      \oplus \bigoplus_{\substack{j=1\\ j\ne i}}^{\m_l} \left( \ccc{l}{i}{j} \act \xx{l}{j} \right)
    &  & (i = 1, \dots, \m_l) 
    \label{eq:N3_lin_2} \\    
  \ee{l}{} &= \left( \ccc{l}{\?}{0} \act z \right)
      \oplus \bigoplus_{j = 1}^{\m_l} \left( \ccc{l}{\?}{j} \act \xx{l}{j} \right)
    \label{eq:N3_lin_3} 
\end{align}
and there exist $\ddd{l}{i}{e}, \ddd{l}{\?}{i}, \ddd{l}{z}{i} \in R$,
such that each receiver can linearly recover its respective message
from its received edge symbols by:
\begin{align}
 \RR{l}{0} : \ \ & 
  \; \; z  \ = \left( \ddd{l}{0}{e} \act \ee{l}{} \right) \oplus  \left( \ddd{l}{\?}{0} \act \ee{l}{0} \right)
    & & (l = 1,2)
    \label{eq:N3_lin_4} \\
 \RR{l}{i} : \ \   &
  \xx{l}{i} = \left( \ddd{l}{i}{e} \act \ee{l}{} \right)  \oplus \left( \ddd{l}{\?}{i} \act  \ee{l}{i} \right)
    & & \Stack{(l = 1,2)}{(i = 1, \dots,\m_l)} 
    \label{eq:N3_lin_5} \\
 R_z :  \ \ &  
  \; \; z  \  = \bigoplus_{l = 1}^{2}  \bigoplus_{i = 0}^{\m_l} \left( \ddd{l}{z}{i} \act \ee{l}{i}  \right)
    \label{eq:N3_lin_6} .
\end{align}

For each $l =1,2$ the block $\B^{(l)}(\m_l)$ together with the source nodes
$S_z, \Ss{l}{1},\Ss{l}{2},\dots,\Ss{l}{\m_l}$
forms a copy of $\Network_0(\m_l)$,
so by Lemma~\ref{lem:N0_lin}
and \eqref{eq:N3_lin_1} -- \eqref{eq:N3_lin_5},
each $\ccc{l}{\?}{i}$ and each $\ddd{l}{\?}{i}$
is invertible in $R$,
and
\begin{align}
  \ccc{l}{i}{j} &= - \dddinv{l}{\?}{i} \! \ddd{l}{i}{e} \, \ccc{l}{\?}{j}
    & & \Stack{(l = 1,2)}{(i,j = 0,1,\dots,\m_l \text{ and } j \ne i).}
    \label{eq:N3_lin_7} 
\end{align}
Equating message components at $R_z$ yields:
\begin{align}
1_R &=  \sum_{l = 1}^2 \sum_{i = 1}^{\m_l} \ddd{l}{z}{i} \, \ccc{l}{i}{0} 
    & & 
    & &\Comment{\eqref{eq:N3_lin_1}, \eqref{eq:N3_lin_2}, \eqref{eq:N3_lin_6}}
    \notag \\
  &= - \sum_{l = 1}^2 \sum_{i = 1}^{\m_l} \ddd{l}{z}{i} \, \dddinv{l}{\?}{i} \! \ddd{l}{i}{e} \, \ccc{l}{\?}{0}
    & & 
    & & \Comment{\eqref{eq:N3_lin_7}}
    \label{eq:N3_lin_8}
\end{align}
and for each $l = 1,2$ we have
\begin{align}
0_R &= \sum_{\substack{ i = 0 \\ i\ne j}}^{\m_l} \ddd{l}{z}{i} \, \ccc{l}{i}{j}
    & & (j = 1,2,\dots, \m_l)
    & & \Comment{\eqref{eq:N3_lin_2}, \eqref{eq:N3_lin_1}, \eqref{eq:N3_lin_6}} 
    \notag \\
  &= - \left( \sum_{\substack{ i = 0 \\ i\ne j}}^{\m_l} \ddd{l}{z}{i} \, \dddinv{l}{\?}{i} \! \ddd{l}{i}{e} \right) \, \ccc{l}{\?}{j}
    & & (j = 1,2,\dots, \m_l)
    & & \Comment{\eqref{eq:N3_lin_7}} 
    \label{eq:N3_lin_07}.
\end{align}
For each $l=1,2,$
by multiplying \eqref{eq:N3_lin_07} by $\cccinv{l}{\?}{j} \! \ccc{l}{\?}{0},$
we have
\begin{align}
  0_R &= \sum_{\substack{ i = 0 \\ i\ne j}}^{\m_l} \ddd{l}{z}{i} \, \dddinv{l}{\?}{i} \! \ddd{l}{i}{e} \, \ccc{l}{\?}{0}
    & & (j = 1,2,\dots, \m_l).
    \label{eq:N3_lin_9} 
\end{align}
Summing
\eqref{eq:N3_lin_9} over $l = 1,2$ and $j = 1,2,\dots,\m_l$
and subtracting \eqref{eq:N3_lin_8}, 
yields
\begin{align}
  - 1_R &= \sum_{l = 1}^{2} \sum_{j = 0}^{\m_l} \sum_{\substack{i = 0 \\ i \ne j}}^{\m_l}
        \ddd{l}{z}{i} \, \dddinv{l}{\?}{i} \! \ddd{l}{i}{e} \, \ccc{l}{\?}{0} \notag \\
    &= \sum_{l = 1}^{2} \m_l \, \sum_{i = 0}^{\m_l} \ddd{l}{z}{i} \, \dddinv{l}{\?}{i} \! \ddd{l}{i}{e} \, \ccc{l}{\?}{0}.
    \label{eq:N3_lin_10}
\end{align}
Equation \eqref{eq:N3_lin_10} implies there exist $r_1,r_2 \in R$ such that
\begin{align}
  1_R &= \m_1 \, r_1 + \m_2 \, r_2 \label{eq:N3_lin_11} .
\end{align}
Since $\GCD{\m_1}{\m_2}$ can be factored out of both terms on the right-hand side of equation \eqref{eq:N3_lin_11},
the ring element $\GCD{\m_1}{\m_2}$ is invertible. 

To prove the converse, let $G$ be a \niceRModule{},
such that $\GCD{\m_1}{\m_2}$ is invertible in $R$.
Define a scalar linear code over $G$ for $\Network_3(\m_1,\m_2)$,
for each $l = 1,2$, by:
\begin{align*}
  \ee{l}{0} &= \bigoplus_{j = 1}^{\m_l} \xx{l}{j}
    \\
  \ee{l}{i} &= z \oplus \bigoplus_{\substack{j=1\\ j\ne i}}^{\m_l} \xx{l}{j} 
    & &(i = 1, \dots, \m_l) \\
  \ee{l}{} &=  z \oplus \bigoplus_{j = 1}^{\m_l}  \xx{l}{j} .
\end{align*}
For each $l = 1,2$, the receivers within $\B^{(l)}(\m_l)$
can linearly recover their respective messages by:
\begin{align*}
\RR{l}{0}: \ \
  & \ee{l}{} \ominus \ee{l}{0} = z \\
\RR{l}{i}: \ \ 
  & \ee{l}{} \ominus \ee{l}{i} = \xx{l}{i}
  & & (i = 1,2,\dots,\m_l).
\end{align*}
Let 
$\m_1' = \m_1/\GCD{\m_1}{\m_2}$ and
$\m_2' = \m_2/\GCD{\m_1}{\m_2}$.
Then $\m_1'$ and $\m_2'$ are relatively prime,
so there exist $n_1,n_2 \in \Z$
such that $n_1 \m_1' + n_2 \m_2' = 1$.
Thus in $R$ we have
$$(n_1 \m_1') \, 1_R + (n_2 \m_2') \, 1_R = 1_R . $$

Receiver $R_z$ can linearly recover message $z$ as follows:
\begin{align*}
R_z: \ \  
    & \bigoplus_{l = 1}^2  \left( \left(n_l \,  \GCD{\m_1}{\m_2}^{-1}  \right) \act
      \left( \bigoplus_{i = 0}^{\m_l} \ee{l}{i}  \ominus \left( \m_l \ee{l}{0} \right) \right) \right) \\
    &= \bigoplus_{l = 1}^2  \left( \left(n_l \,  \GCD{\m_1}{\m_2}^{-1} \right) \act \left( \m_l \, z \right) \right) \\
    &= (n_1 \m_1' \, z) \oplus (n_2 \m_2' \, z) = \left((n_1 \m_1') \, 1_R + (n_2 \m_2') \, 1_R \right) \, z = z.
\end{align*}
Thus the code is a scalar linear solution for $\Network_3(\m_1,\m_2)$.

\end{proof}
}


\begin{corollary}
  Let $\m_1, \m_2 \geq 2$ and $\alpha, s, t \ge 1$ be integers such that
  $\m_2 = s \m_1^{\alpha}$ and 
  $s$ and $t$ are relatively prime to $\m_1$.
  Then network $\Network_3(\m_1, \m_2 )$ is solvable 
  over an alphabet of size $t \m_1^{\alpha+1}$.
  \label{cor:N3_non}
\end{corollary}
\begin{proof}
  By Lemma~\ref{lem:N3_non},
  network $\Network_3(\m_1,\m_2)$ is solvable over an alphabet of size $\m_1^{\alpha+1}$.
  $\Z_t$ is a standard $\Z_t$-module and $\Char{\Z_t} = t$ is relatively prime to $\m_1$,
  so by Lemma~\ref{lem:N3_lin},
  network $\Network_3(\m_1,\m_2)$ is scalar linear solvable over the ring $\Z_t$.

  By taking the Cartesian product code of these solutions,
  network $\Network_3(\m_1,\m_2)$ is solvable over an alphabet of size $t \m_1^{\alpha+1}$.
\end{proof}

For each $\m_1 \geq 2$
  and $\alpha, s \geq 1$ 
  such that $s$ is relatively prime to $\m_1$,
  let $\m_2 = \m_1^{\alpha} s$.
By Lemma~\ref{lem:N3_non}, network $\Network_3(\m_1,\m_2)$ is solvable over $\Z_{\m_1^{\alpha+1}},$
but we have 
$$\GGGCD{\m_1}{\m_2}{\Char{\Z_{\m_1^{\alpha+1}}}} = \GGGCD{\m_1}{\m_1^{\alpha} s}{\m_1^{\alpha+1}} = \m_1 \ne 1,$$ 
in this case,
so by Lemma~\ref{lem:N3_lin} the solution is necessarily non-linear.
This also implies that the Cartesian product code in Corollary~\ref{cor:N3_non}
is necessarily non-linear.

\subsection{Capacity and linear capacity of \texorpdfstring{$\Network_3(\m_1,\m_2)$}{N3(m1,m2)}} \label{ssec:N3_cap}

Since the characteristic of any finite field is prime, 
the conditions of (b) and (c)
of the following lemma are complements of one another.
\begin{lemma}
  For each $\m_1,\m_2 \geq 2,$
  network $\Network_3(\m_1,\m_2)$ has
  \begin{itemize}
  \itemsep0em 
      \item[(a)] capacity equal to $1$,
      \item[(b)] linear capacity equal to $1$ for any finite-field alphabet
          whose characteristic is relatively prime to $\m_1$ or $\m_2,$
      \item[(c)] linear capacity equal to $1 - \frac{1}{2\m_1 + 2\m_2+3}$ for any finite-field alphabet
          whose characteristic divides $\m_1$ and $\m_2$.
  \end{itemize}
  \label{lem:N3_cap}
\end{lemma}

\newcommand{\ProofOflemNThreeCap}{
\begin{proof}[Proof of Lemma \ref{lem:N3_cap}]
By Lemma~\ref{lem:N3_lin}, network $\Network_3(\m_1,\m_2)$
is scalar linear solvable over any finite-field alphabet whose characteristic is relatively prime to $\m_1$ or $\m_2,$
so the network's linear capacity for such finite-field alphabets is at least $1$.
By Lemma~\ref{lem:N0_cap}, 
network $\Network_0(\m_1)$ has capacity equal to $1$,
the block $\B^{(1)}(\m_1)$ together with the source nodes
$S_z, \Ss{1}{1},\Ss{1}{2},\dots,\Ss{1}{\m_1}$
forms a copy of $\Network_0(\m_1)$,
so the capacity of $\Network_3(\m_1,\m_2)$ is at most $1$.
Thus both the capacity of $\Network_3(\m_1,\m_2)$
and its linear capacity over any finite-field alphabet whose characteristic is relatively prime to $\m_1$ or $\m_2$
are $1$. 

To prove part (c),
consider a $(k,n)$ fractional linear solution for $\Network_3(\m_1,\m_2)$
over a finite field $\F$ whose characteristic divides both $\m_1$ and $\m_2$.
Since $\Div{\Char{\F}}{\m_1}$ and $\Div{\Char{\F}}{\m_2}$,
we have $\m_1 = \m_2 = 0$ in $\F$.

We have $\xx{l}{j},z \in \F^k$
and $\ee{l}{i},\ee{l}{} \in \F^n$,
with $n \geq k$, since the capacity is one.
There exist $n \times k$ coding matrices
$\MMM{l}{\?}{j},\MMM{l}{i}{j}$ with entries in $\F$,
such that for each $l = 1,2$ 
the edge vectors can be written as:
\begin{align}
  \ee{l}{0} &= \sum_{j = 1}^{\m_l} \MMM{l}{0}{j} \, \xx{l}{j}
    \label{eq:N3_cap_1} \\
  \ee{l}{i} &= \MMM{l}{\?}{0} \, z + \sum_{\substack{j = 1 \\ j \ne i }}^{\m_l} \MMM{l}{i}{j} \, \xx{l}{j}
    & & (i = 1,2, \dots, \m_l)
    \label{eq:N3_cap_2} \\
  \ee{l}{} &= \MMM{l}{\?}{0} \, z + \sum_{j = 1}^{\m_l} \MMM{l}{\?}{j} \, \xx{l}{j}
    \label{eq:N3_cap_3} 
\end{align}
and there exist $k \times n$ decoding matrices $\DDD{l}{i}{e}, \,  \DDD{l}{\?}{i}$
with entries in $\F$, such that 
for each $l = 1,2$
the receivers within the block $\B^{(l)}(\m_l)$ can recover their respective messages
from their received edge vectors by:
\begin{align}
\RR{l}{0}: \ \ &
  \; \; z \ = \DDD{l}{0}{e} \, \ee{l}{} 
            + \DDD{l}{\?}{0} \, \ee{l}{0}
    \label{eq:N3_cap_4} \\
\RR{l}{i}: \ \  &  
  \xx{l}{i} = \DDD{l}{i}{e} \, \ee{l}{} 
            + \DDD{l}{\?}{i} \, \ee{l}{i}
    & & (i = 1,2,\dots, \m_l).
    \label{eq:N3_cap_5}
\end{align}

Since the receiver $R_z$ recovers message $z$ linearly from its incoming edge vectors, we have
\begin{align}
  \left\{\ee{l}{i} \; : \; 
    \Stack{l = 1,2}{i = 0,1,\dots,\m_l}
    \right\} 
    & \longrightarrow z .
    \label{eq:N3_cap_6}
\end{align}

By setting $z = 0$ in \eqref{eq:N3_cap_4}, for each $l = 1,2$ we have
\begin{align}
  0 &= \DDD{l}{0}{e} \, \sum_{j = 1}^{\m_l} \MMM{l}{\?}{j} \, \xx{l}{j} 
     + \DDD{l}{\?}{0} \, \ee{l}{0} 
    & & \Comment{\eqref{eq:N3_cap_1}, \eqref{eq:N3_cap_3}, \eqref{eq:N3_cap_4}}\notag \\
  \therefore \, &
    \sum_{j = 1}^{\m_l} \MMM{l}{\?}{j} \, \xx{l}{j} 
      \longrightarrow
    \DDD{l}{\?}{0} \, \ee{l}{0}, \label{eq:N3_cap_7}
\end{align}
and similarly, by setting $\xx{l}{i} = 0$ in \eqref{eq:N3_cap_5} for $l = 1,2$ we have
\begin{align}
  0 &= \DDD{l}{i}{e} \, \left( \MMM{l}{\?}{0} \, z + \sum_{\substack{ j = 1 \\ j \ne i}}^{\m_l} \MMM{l}{\?}{j} \, \xx{l}{j}  \right) 
     + \DDD{l}{\?}{i} \, \ee{l}{i}
      & & (i = 1,2,\dots,\m_l) 
      & & \Comment{\eqref{eq:N3_cap_2}, \eqref{eq:N3_cap_3}, \eqref{eq:N3_cap_4}}\notag \\
 \therefore \, &
    \ee{l}{i}
      \longrightarrow
      \DDD{l}{i}{e} \left(\MMM{l}{\?}{0} \, z + \sum_{\substack{ j = 1 \\ j \ne i}}^{\m_l} \MMM{l}{\?}{j} \, \xx{l}{j} \right)
      & & (i = 1,2,\dots, \m_l).
      \label{eq:N3_cap_8}
\end{align}

As in Lemma~\ref{lem:N1_cap},
for each $l = 1,2$ and $i = 1,2,\dots,\m_l$,
let $\QQQ{l}{\?}{0}$ be the matrix $Q$
in Lemma~\ref{lem:mat_3}
corresponding to when $\DDD{l}{\?}{0}$ is the matrix $A$
in the lemma,
and let
$\QQQ{l}{i}{e}$ be the matrix $Q$
corresponding to when $\DDD{l}{i}{e}$ is the matrix $A$.

Let $L^{(1)}$ and $L^{(2)}$ be the lists
from Lemma~\ref{lem:N1_cap} (where $z$ plays the role of $x_0$),
corresponding to the left-hand side and right-hand side
of the network, respectively.
Specifically, for each $l = 1,2,$ 
let $L^{(l)}$ be the list
\begin{align*}
  &\QQQ{l}{\?}{0} \, \ee{l}{0} \\
  &\ee{l}{i}    
    & & ( i = 1,2,\dots,\m_l) \\
  &\QQQ{l}{i}{e} \left( \MMM{l}{\?}{0} \, z + \sum_{\substack{ j = 1 \\ j \ne i}}^{\m_l} \MMM{l}{\?}{j} \, \xx{l}{j} \right)   
    & & (i = 1,2,\dots,\m_l) .
\end{align*}

For each $l = 1,2$ we have
\begin{align}
  L^{(l)} & \longrightarrow 
    \DDD{l}{i}{e} \left(\MMM{l}{\?}{0} \, z + \sum_{\substack{ j = 1 \\ j \ne i}}^{\m_l} \MMM{l}{\?}{j} \, \xx{l}{j} \right)
      & & \Comment{\eqref{eq:N3_cap_8}} 
      \label{eq:N3_cap_9} \\
  L^{(l)} & \longrightarrow
    \MMM{l}{\?}{0} \, z + \sum_{\substack{ j = 1 \\ j \ne i}}^{\m_l} \MMM{l}{\?}{j} \, \xx{l}{j}
      & & \Comment{Lemma~\ref{lem:mat_3}, \eqref{eq:N3_cap_9}} 
      \label{eq:N3_cap_10} .
\end{align}

For each $l = 1,2$ we also have
\begin{align}
  & \left\{ \MMM{l}{\?}{0} \, z + \sum_{\substack{ j = 1 \\ j \ne i}}^{\m_l} \MMM{l}{\?}{j} \, \xx{l}{j} 
      \; : \; 
      i = 1,2,\dots, \m_l\right\} 
    \notag \\
\longrightarrow & \sum_{i = 1}^{\m_l} \left( \MMM{l}{\?}{0} \, z + \sum_{\substack{ j = 1 \\ j \ne i}}^{\m_l} \MMM{l}{\?}{j} \, \xx{l}{j} \right) 
   \notag \\
  &= \m_l \, \MMM{l}{\?}{0} \, z + (\m_1 - 1) \, \sum_{j = 1}^{\m_l} \MMM{l}{\?}{j} \, \xx{l}{j} 
    \notag \\
  &= - \sum_{j = 1}^{\m_l} \MMM{l}{\?}{j} \, \xx{l}{j} 
    && \Comment{$\Div{\Char{\F}}{\m_l}$},
    \label{eq:N3_cap_11}
\end{align}
and so
\begin{align}
  L^{(l)} & \longrightarrow
    \sum_{j = 1}^{\m_l} \MMM{l}{\?}{j} \, \xx{l}{j} 
      & & \Comment{\eqref{eq:N3_cap_11}, \eqref{eq:N3_cap_10}}
      \label{eq:N3_cap_12} \\
  L^{(l)} & \longrightarrow
    \DDD{l}{\?}{0} \, \ee{l}{0} 
      & & \Comment{\eqref{eq:N3_cap_7}, \eqref{eq:N3_cap_12}} 
      \label{eq:N3_cap_13} \\
  L^{(l)} & \longrightarrow
      \ee{l}{0}
      & & \Comment{Lemma~\ref{lem:mat_3}, \eqref{eq:N3_cap_13}}
      \label{eq:N3_cap_14}.
\end{align}

We have
\begin{align}
  L^{(1)} , L^{(2)} \,& \longrightarrow \ z 
    && \Comment{\eqref{eq:N3_cap_6}, \eqref{eq:N3_cap_14}}
    \label{eq:N3_cap_15} .
\end{align}

For each $l = 1,2$ we also have
\begin{align}
  z, \, \sum_{j = 1}^{\m_l} \MMM{l}{\?}{j} \, \xx{l}{j} 
    & \longrightarrow \ee{l}{} 
    & &
    & & \Comment{\eqref{eq:N3_cap_3}}
    \label{eq:N3_cap_16} \\
  L^{(l)}, \,  z & \longrightarrow \ee{l}{}
    & &
    & & \Comment{\eqref{eq:N3_cap_12}, \eqref{eq:N3_cap_16}}
    \label{eq:N3_cap_17} \\
  L^{(l)}, \ z \, &\longrightarrow \ \xx{l}{i}  
  & &  (i = 1,2,\dots, \m_l) 
  & & \Comment{\eqref{eq:N3_cap_5}, \eqref{eq:N3_cap_17}} 
  \label{eq:N3_cap_18}.
\end{align}

Thus 
\begin{align}
L^{(1)} , L^{(2)} \, \longrightarrow \, z, \, \left\{\xx{l}{i} 
  \; : \; 
  \Stack{l = 1,2}{i = 1,2,\dots,\m_l}
  \right\}
  & & \Comment{\eqref{eq:N3_cap_15}, \eqref{eq:N3_cap_18}}
  \label{eq:N3_cap_19}.
\end{align}

We have $L^{(l)}$ corresponding to the same set of vector functions 
as the list $L$ for $\Network_1(\m_l)$ in Lemma~\ref{lem:N1_cap} 
(with a slight change of labeling).
Thus the bound on the entropy of the list $L$ in \eqref{eq:N1_cap_33} in Lemma~\ref{lem:N1_cap}
can be used to bound the entropy of 
the list $L^{(1)} , L^{(2)}$:
\begin{align}
  H\left(L^{(1)} , L^{(2)} \right) 
    &  \leq \ (2 \m_1 + 2 \m_2 +2) \, n - (\m_1 + \m_2 + 2) \, k   
    & & \Comment{\eqref{eq:N1_cap_33}} 
    \label{eq:N3_cap_20}.
\end{align}
But then we have
\begin{align*}
  (\m_1 + \m_2 + 1) \, k &= H\left(z, \, \left\{\xx{l}{i} \; : \; \Stack{l = 1,2}{i = 1,2,\dots,\m_l}  \right\} \right)
        & & \Comment{$z, \xx{l}{i} \in \F^k$} \\
      & \leq H(L_1 , L_2) 
        & & \Comment{\eqref{eq:N3_cap_19}}\\ 
      & \leq (2 \m_1 + 2 \m_2 + 2) \ n - ( \m_1 + \m_2 + 2) \ k 
        & & \Comment{\eqref{eq:N3_cap_20}} \\
  \therefore \frac{k}{n} & \leq \frac{ 2 \m_1 + 2 \m_2 + 2} { 2 \m_1 +2 \m_2 + 3}.
\end{align*}

Thus the linear capacity of $\Network_3(\m_1,\m_2)$
for finite-field alphabets whose characteristic divides both $\m_1$ and $\m_2$ is upper bounded by
$$1 - \frac{1}{2 \m_1 + 2 \m_2 + 3}.$$

Consider a $(2 \m_1 + 2 \m_2 + 2, 2 \m_1 + 2 \m_2 + 3)$ fractional linear code for 
$\Network_3(\m_1,\m_2)$ over any finite-field alphabet whose characteristic divides both $\m_1$ and $\m_2,$ 
described below.

The edges symbols on the left-hand side of $\Network_3(\m_1,\m_2)$ are given by:
\begin{align*}
\vcomp{\ee{1}{0}}{l} & = 
  \left\{ \begin{array}{ll}
    \dsum_{\substack{ j = 1 \\ j \ne l}}^{\m_1}  \vcomp{\xx{1}{j}}{l} 
    \; \; \; \; \; \; \; \; \; \; \; \; \; \; \; 
      &  (l = 1,2, \dots, \m_1)  
      \\ [2em]
    \dsum_{j = 1}^{\m_1} \vcomp{\xx{1}{j}}{l} 
       &  (l = \m_1 + 1,\dots, 2 \m_1 + 2 \m_2 + 2) 
       \\ [2em]
    \dsum_{j =2 }^{\m_1} \vcomp{\xx{1}{j}}{j}
      & (l = 2 \m_1 + 2 \m_2 + 3)
    \end{array} \right. \\
\\
\vcomp{\ee{1}{i}}{l}  & = 
  \left\{ \begin{array}{ll}
    [z]_l + \dsum_{\substack{ j = 1 \\ j \ne i \\ j \ne l}}^{\m_1} \vcomp{\xx{1}{j}}{l} 
      &  (l = 1,2,\dots, \m_1 \text{ and } l \ne i)
      \\ [3em]
    [z]_{\m_1+1} + \dsum_{ \substack{ j  = 1 \\ j \ne i }}^{\m_1} \vcomp{\xx{1}{j}}{j} 
      & (l = i ) 
      \\ [3em]
    [z]_l + \dsum_{\substack{ j = 1 \\ j \ne i}}^{\m_1} \vcomp{\xx{1}{j}}{l} 
      & (l = \m_1 + 1,\dots, 2 \m_1 + 2 \m_2 + 2) 
      \\ [3em]
    [z]_{\m_1 + i + 1} 
      & (l  = 2 \m_1 + 2 \m_2 +3)
    \end{array} \right. 
    & ( i = 1,2,\dots, \m_1)\\
\\
\vcomp{\ee{1}{}}{l} &= 
  \left\{ \begin{array}{ll}
    [z]_l + \dsum_{\substack{j = 1 \\ j \ne l}}^{\m_1} \vcomp{\xx{1}{j}}{l}  
      \; \; \;
       &  (l = 1,2,\dots, \m_1)  
       \\ [2em]
    [z]_l + \dsum_{j= 1}^{\m_1} \vcomp{\xx{1}{j}}{l}    
      &  (l = \m_1+1,\dots, 2 \m_1 + 2 \m_2 + 2) 
      \\ [2em]
    [z]_{\m_1+1} + \dsum_{j = 1}^{\m_1} \vcomp{\xx{1}{j}}{j} 
      &  (l = 2 \m_1 + 2 \m_2 + 3) .
  \end{array} \right.     
\end{align*}

\newpage
For brevity, let $\delta = 2 \m_1 + \m_2 + 2 = n - (\m_2 + 1)$.
The edges symbols on the right-hand side of $\Network_3(\m_1,\m_2)$ are given by:
\begin{align*}
\vcomp{\ee{2}{0}}{l} &= 
  \left\{ \begin{array}{ll}
    \dsum_{j = 1}^{\m_2} \vcomp{\xx{2}{j}}{l} 
     \; \; \; \; \; \; \; \; \; \; \;  \; \; \; \;  
        & \ \ \;\; (l = 1, 2, \dots, \delta)  
        \\ [2em]
      \dsum_{\substack{ j = 1 \\ j \ne l - \delta}}^{\m_2} \vcomp{\xx{2}{j}}{l} 
        & \ \ \;   \; (l = \delta+1 ,  \dots,  \delta + \m_2) 
        \\ [2em] 
      \dsum_{j =2}^{\m_2} \vcomp{\xx{2}{j}}{\delta + j} 
        & \ \ \;  \; (l = \delta + \m_2 + 1)
  \end{array} \right. \\
\\
\vcomp{\ee{2}{i}}{l}  &= 
  \left\{ \begin{array}{ll}
    [z]_{l} + \dsum_{\substack{ j = 1 \\ j \ne i}}^{\m_2} \vcomp{\xx{2}{j}}{l} 
      &\ \ \; \;   (l = 1, 2, \dots,  \delta)  
      \\ [3em] 
    [z]_{\delta} + \dsum_{ \substack{ j  = 1 \\ j \ne i }}^{\m_2} \vcomp{\xx{2}{j}}{\delta + j} 
      & \ \  \; \;(l = \delta + i) 
      \\ [3em]
    [z]_{l} + \dsum_{\substack{ j = 1 \\ j \ne i \\ j \ne l - \delta }}^{\m_2} \vcomp{\xx{2}{j}}{l}  
      & \left(  \begin{array}{l}  
          l = \delta+ 1 ,  \dots  , \delta + \m_2 \\  
          \text{and } l \ne \delta + i  
      \end{array} \right) 
      \\ [3em]
    [z]_{2 \m_1 + 1 + i } 
      & \ \ \; \; (l  = \delta + \m_2 + 1)
    \end{array} \right.
    & (i = 1,2,\dots, \m_2) \\
\\
\vcomp{\ee{2}{}}{l} &=
   \left\{ \begin{array}{ll}
      [z]_{l} + \dsum_{j= 1}^{\m_2} \vcomp{\xx{2}{j}}{l}          
         & \ \ \; \; (l = 1, 2, \dots,  \delta)  
         \\ [2em] 
      [z]_{l} + \dsum_{\substack{j = 1 \\ j \ne l-\delta}}^{\m_2} \vcomp{\xx{2}{j}}{l}  
        & \ \ \; \; (l =  \delta + 1  ,\dots, \delta + \m_2) 
        \\ [2em] 
      [z]_{\delta} + \dsum_{j = 1}^{\m_2} \vcomp{\xx{2}{j}}{\delta + j}  
        & \ \ \;   \; (l = \delta + \m_2 + 1) .
    \end{array} \right.     
\end{align*}

We have 
\begin{align}
  \sum_{\substack{i = 1 \\ i \ne l}}^{\m_1} \vcomp{\ee{1}{i}}{l} 
    &= (\m_1 - 1) \, [z]_l + (\m_1 - 2) \, \sum_{\substack{ j = 1 \\  j \ne l}}^{\m_1} \vcomp{\xx{1}{j}}{l}  
      & & (l = 1,2,\dots,\m_1) \notag \\
  &= -[z]_l - 2 \, \vcomp{\ee{1}{0}}{l}
      & & \Comment{$\Div{\Char{\F}}{\m_1}$}
      \label{eq:N3_cap_21} \\
  \sum_{\substack{i = 1 \\ i \ne l-\delta}}^{\m_2} \vcomp{\ee{2}{i}}{l}
    &=  (\m_2 - 1) \, [z]_l + (\m_2 - 2) \, \sum_{\substack{ j = 1 \\ j \ne l - \delta }}^{\m_2} \vcomp{\xx{2}{j}}{l}  
      & & (l = \delta + 1,\dots,\delta + \m_2) \notag \\
  &= -[z]_l - 2 \, \vcomp{\ee{2}{0}}{l}
      & & \Comment{$\Div{\Char{\F}}{\m_2}$}.
      \label{eq:N3_cap_22}
\end{align}

Each of the receivers can linearly recover 
each of the $2 \m_1 + 2 \m_2 + 2$
components of its demanded message
from its received vectors by:
\begin{align*}
\RR{1}{0}: \ \
   & \vcomp{\ee{1}{}}{l} - \vcomp{\ee{1}{0}}{l} = [z]_{l}
      & &  (l = 1,2,\dots, 2 \m_1 + 2 \m_2 + 2)  \\
  \\
\RR{1}{i}: \ \ 
   &  \vcomp{\ee{1}{}}{2 \m_1 + 2 \m_2 + 3} - \vcomp{\ee{1}{i}}{i} 
        = \vcomp{\xx{1}{i}}{i}
        && (i = 1,2,\dots,\m_1) \\
   &  \vcomp{\ee{1}{}}{l} - \vcomp{\ee{1}{i}}{l} 
        = \vcomp{\xx{1}{i}}{l} 
        & & (l = 1,2,\dots, 2 \m_1 + 2 \m_2 + 2 \text{ and } l \ne i)\\
    \\
\RR{2}{0}: \ \
   &  \vcomp{\ee{2}{}}{l} - \vcomp{\ee{2}{0}}{l} = [z]_{l}
      & & (l = 1,2,\dots, 2 \m_1 + 2 \m_2 + 2) \\
\\
\RR{2}{i}: \ \
  & \vcomp{\ee{2}{}}{\delta + \m_2 + 1} - \vcomp{\ee{2}{i}}{\delta+i} 
        = \vcomp{\xx{2}{i}}{\delta+i}
        & & (i = 1,2,\dots,\m_2) \\
   &  \vcomp{\ee{2}{}}{l} - \vcomp{\ee{2}{i}}{l} 
        = \vcomp{\xx{2}{i}}{l} 
        & & (l = 1,2,\dots, 2\m_1 + 2 \m_2 + 2 \text{ and } l \ne \delta + i) 
\end{align*}

\begin{align*}
R_z: \ \
   & -2 \, \vcomp{\ee{1}{0}}{l} - \sum_{\substack{i = 1 \\ i \ne l}}^{\m_1} \vcomp{\ee{1}{i}}{l} = [z]_l           
      & & \; \; \; \; (l = 1,2,\dots,\m_1) 
      &  & \Comment{\eqref{eq:N3_cap_21}} \\
    \\
   &\vcomp{\ee{1}{1}}{1} - \vcomp{\ee{1}{0}}{2 \m_1 + 2 \m_2 + 3}  
     = [z]_{\m_1+1} \\     
    \\
   & \vcomp{\ee{1}{l - \m_1 - 1}}{2 \m_1 + 2 \m_2 + 3} = [z]_l
    & & \; \; \; \; (l = \m_1 + 2 , \dots, 2 \m_1 + 1) \\ 
    \\
   & \vcomp{\ee{2}{l - 2 \m_1 - 1}}{\delta + \m_1 + 1} = [z]_l
    & & \; \; \; \; (l = 2 \m_1 + 2 , \dots, 2 \m_1 + \m_2 + 1)  \\
    \\  
   & \vcomp{\ee{2}{1}}{\delta+1} - \vcomp{\ee{2}{0}}{2 _1 + 2 \m_2 + 3}  
    = [z]_{\delta} & & \; \; \; \; (\delta = 2 \m_1 + \m_2 + 2) \\   
    \\
   & -2 \, \vcomp{\ee{2}{0}}{l} - \sum_{\substack{i = 1 \\ i \ne l-\delta}}^{\m_2} \vcomp{\ee{2}{i}}{l} = [z]_l
        & & \; \; \; \;  (l = \delta + 1,\dots, \delta + \m_2)
        & & \Comment{\eqref{eq:N3_cap_22}} .
\end{align*}
  Thus the code is in fact a linear solution for $\Network_3(\m_1,\m_2)$.
\end{proof}
}

\clearpage

\section{The network \texorpdfstring{$\Network_4(\m)$}{N4(m)}} \label{sec:N4}

A \textit{disjoint union} of networks refers to a new network formed by combining existing networks 
with disjoint sets of nodes, edges, sources, and receivers.
Specifically, the nodes/edges/sources/receivers in the resulting network
are the disjoint union of the nodes/edges/sources/receivers in the smaller networks.

\begin{remark}
  The \textit{disjoint union} of networks
  $\Network_1, \dots, \Network_\n$,
  has a $(k,n)$ solution over alphabet $\A$
  if and only if $\Network_1, \dots, \Network_\n$
  each has a $(k,n)$ solution over $\A$.
  \label{rem:DJ}
\end{remark}

For any integer $\m \ge 2,$ 
let $\omega(\m)$ denote the number of distinct prime factors of $\m$.
Denote the prime factorization of $\m$ by
$$\m = \PrimeFact{\m}$$ 
where $\gamma_1,\dots,\gamma_{\omega(\m)} \geq 1$ and
$p_1,\dots,p_{\omega(\m)}$ are distinct primes.
We define the following functions of $\m$ and its prime divisors,
which will be used throughout this section:
\begin{align}
  f(\m) &= p_1^{\gamma_1 - 1} \dots p_{\omega(\m)}^{\gamma_{\omega(m)} -1 }
     \label{eq:N4_f}\\
  \mu(\m,i) &=\min \; \left\{ \alpha \ge 0 \; : \; p_i^{\alpha} \ge f(\m) \right\}
    & & ( i = 1,\dots, \omega(\m)) 
    \label{eq:N4_u}\\
  g(\m,i) &= p_i^{\gamma_i-1} \prod_{\substack{j = 1 \\ j \ne i}}^{\omega(\m)} p_j^{\mu(\m,j)}
    & & ( i = 1,\dots, \omega(\m)) 
    \label{eq:N4_g}.
\end{align}

\noindent For each $\m \geq 2$ with prime factorization $\m = \PrimeFact{\m}$,
we construct network $\Network_4(\m)$ from the following \textit{disjoint union}%
\footnote{When node (respectively, edge and message) labels are repeated (e.g. $\Network_1(\m_1)$ and $\Network_1(\m_2)$
both have receiver $R_x$), add additional superscripts to each node (respectively, edge and message) to avoid repeated labels.
Each disjoint network has a set of messages, nodes, and edges
which is disjoint to every other network's set in the union.
The messages, nodes, and edges are not directly referenced in this section,
so the additional level of labeling is arbitrary
so long as the networks are disjoint.}
of networks:
\begin{align}
  \Network_4(\m)
    & = 
    \left( 
      \bigcup_{\substack{\text{prime } q \\ 
                    \NDiv{q}{\m} \\ 
                    q < f(\m)}} 
          \Network_1(q) 
    \right)
    \, \cup \,
    \left( 
        \bigcup_{i = 1}^{\omega(\m)} 
          \Network_2\left(p_i^{\gamma_i}, \, (\m / p_i^{\gamma_i}) \right) 
    \right)
    \, \cup \,
    \left( 
        \bigcup_{\substack{i = 1 \\ 
                  \gamma_i > 1}}^{\omega(\m)} 
          \Network_3\left(p_i, \ g(\m,i) \right) 
    \right).
    \label{eq:N4_N4}
\end{align}

\begin{theorem}
  For each $\m \geq 2,$
  the network $\Network_4(\m)$ is:
  \begin{enumerate}
  \itemsep0em 
    \item solvable over an alphabet of size $\m$,
    \item not solvable over any alphabet whose size is less than $\m$,
    \item scalar linear solvable over $\GF{\m}$, if $\m$ is prime,
    \item neither vector linear solvable over any $R$-module alphabet
        nor asymptotically linear solvable over any finite-field alphabet
        if $\m$ is composite.
  \end{enumerate}
  \label{thm:N4_results}
\end{theorem}
\begin{proof}
The theorem follows immediately from
Theorems~\ref{thm:N4_m}, \ref{thm:N4_solv}, \ref{thm:N4_prime}, \ref{thm:N4_R},
and Corollary~\ref{cor:N4_asymp}.
\end{proof}

\begin{example}
Consider the special cases of the square-free integer%
\footnote{An integer is \textit{square-free} if it is not divisible by the square of any prime.}
$6$,
the prime power $27$,
and the integer $100$ which is neither square-free nor a prime power.

\begin{itemize}
\item $\m = 6 = 2^1 3^1$.
We have $\gamma_1 = \gamma_2 = 1$ and $f(\m) = 2^{(1-1)} 3^{(1-1)} = 1$,
so $\Network_4(6)$ has neither $\Network_1$ nor $\Network_3$ components.
Thus by \eqref{eq:N4_N4}, 
network $\Network_4(6)$ is the disjoint union of networks: 
$$\Network_2(2,3) \, \cup \, \Network_2(3,2).$$

\item $\m = 27 = 3^3$.
We have $f(27) = 3^{(3-1)} = 9$, $g(27,1) = 3^{(3-1)} = 9$,
and the primes less than $f(27)$ which do not divide $27$ are $2,5,$ and $7$.
Thus by \eqref{eq:N4_N4}, 
network $\Network_4(6)$ is the disjoint union of networks:
$$\Network_1(2) \, \cup \, \Network_1(5) \, \cup \, \Network_1(7) \, \cup \, \Network_2(27,1) \, \cup \, \Network_3(3,9) .$$

\item $\m = 100 = 2^2 5^2$.
We have $f(100) = 2^{(2-1)} 5^{(2-1)} = 10$.
Then $\mu(100,1) = 4$, since $2^4 > f(100) > 2^3$,
and $\mu(100,2) = 2,$ since $5^2 > f(100) > 5^1$.
So $g(100,1) = 2^1 5^2,$ $g(100,2) = 5^1 2^4$,
and the primes less than $f(100)$ which do not divide $100$ are $3$ and $7$.
Thus by \eqref{eq:N4_N4}, 
network $\Network_4(100)$ is the disjoint union of networks:
$$\Network_1(3) \, \cup \, \Network_1(7) \, \cup \, \Network_2(4,25) \, \cup \, \Network_2(25,4) \, \cup \, \Network_3(2,50) \, \cup \,  \Network_3(5,80).$$

We will use these networks as running examples throughout this section and will refer back to these constructions.
\label{ex:N4_construction}
\end{itemize}

\end{example}

\subsection{Solvability conditions of \texorpdfstring{$\Network_4(\m)$}{N4(m)}} \label{ssec:N4_solv}

The following lemma shows that each disjoint component of $\Network_4(\m)$
is solvable over an alphabet of size $\m$,
and therefore $\Network_4(\m)$ is solvable over an alphabet of size $\m$.
The proofs of Theorems~\ref{thm:N4_m} and \ref{thm:N4_solv} make use of the functions
$f,\mu,\text{ and } g$ defined in \eqref{eq:N4_f}, \eqref{eq:N4_u}, and \eqref{eq:N4_g}, respectively.

\begin{theorem}
  For each $\m \geq 2$,
  network $\Network_4(\m)$ is solvable 
  over an alphabet of size $m$.
  \label{thm:N4_m}
\end{theorem}
\begin{proof}

Let $m$ have prime factorization $\m = \PrimeFact{\m}$.

For each prime $q < f(\m)$ such that $\NDiv{q}{\m}$,
by \eqref{eq:N4_N4},
network $\Network_4(\m)$ contains
a copy of $\Network_1(q)$.
$\Z_{\m}$ is a standard $\Z_{\m}$-module
and $\Char{\Z_{\m}} = \m$ is relatively prime to $q$,
so by Lemma~\ref{lem:N1_lin},
network $\Network_1(q)$
is scalar linear solvable over the ring $\Z_{\m}$.

For each $i = 1,\dots, \omega(\m)$,
by \eqref{eq:N4_N4},
network $\Network_4(\m)$ contains
a copy of $\Network_2\left(p_i^{\gamma_i}, \, (\m / p_i^{\gamma_i}) \right)$.
By Lemma~\ref{lem:N2_non},
network $\Network_2\left(p_i^{\gamma_i}, \, (\m / p_i^{\gamma_i}) \right)$
is solvable over an alphabet of size $\m$.

For each $i = 1,\dots,\omega(\m)$
such that $\gamma_i > 1$,
by \eqref{eq:N4_N4},
network $\Network_4(\m)$ contains
a copy of $\Network_3(p_i, \, g(\m,i))$.
Also, $p_i$ and $\m / p_i^{\gamma_i}$ are relatively prime,
and by \eqref{eq:N4_g}, 
$g(\m,i)$ is the product of $p_i^{\gamma_i - 1}$ and a term which is relatively prime to $p_i$,
so by Corollary~\ref{cor:N3_non}, 
network $\Network_3\left(p_i, \ g(\m,i) \right)$
is solvable over an alphabet of size $\m$.

Thus each disjoint component of $\Network_4(\m)$ is solvable over an alphabet of size $\m$,
so $\Network_4(\m)$ is solvable over an alphabet of size $\m$.
\end{proof}

Each network $\Network_1,\Network_2,$ and $\Network_3$ 
requires the alphabet size to meet some divisibility condition in order to have a solution over that alphabet.
The following lemma shows that
because of these conditions, there does not exist an alphabet whose size is less than $\m$
over which each component of $\Network_4(\m)$ is solvable.
\begin{theorem}
  For each $\m \geq 2$,
  if network $\Network_4(\m)$ is solvable over alphabet $\A$,
  then $\vert\A\vert \geq \m$.
  \label{thm:N4_solv}
\end{theorem}
\begin{proof}

Assume to the contrary that $\Network_4(\m)$ is solvable over an alphabet $\A$
such that $\vert\A\vert < \m$.
Then each disjoint component of $\Network_4(\m)$ must be solvable over $\A$.

Let $m$ have prime factorization $\m = \PrimeFact{\m}$.

For each $i = 1,\dots, \omega(\m)$,
by \eqref{eq:N4_N4},
network $\Network_4(\m)$ contains a copy of 
$\Network_2\left(p_i^{\gamma_i}, \, (\m / p_i^{\gamma_i}) \right)$.
Since network $\Network_2\left(p_i^{\gamma_i}, \, (\m / p_i^{\gamma_i}) \right)$ is solvable over $\A$,
then by Lemma~\ref{lem:N2_solv}, $p_i$ is not relatively prime to $\vert\A\vert$.
Since $p_i$ is prime, we have $\Div{p_i}{\vert\A\vert}$,
and thus $\Div{p_1 \cdots p_{\omega(\m)}}{\vert \A \vert}$.
Let 
$$\delta = \frac{\vert \A \vert}{p_1 \cdots p_{\omega(\m)}}.$$

If $\m = p_1 \cdots p_{\omega(\m)}$ (i.e. $\m$ is square-free),
then we contradict the assumption that $\vert \A \vert < \m$.

So we may assume $\m > p_1 \cdots p_{\omega(\m)}$,
which implies $\delta \geq 2$.
If $\delta \geq f(\m),$
then 
\begin{align*}
  \vert\A\vert &= \delta \, p_1 \dots p_{\omega(\m)} \geq f(\m) \, p_1 \dots p_{\omega(\m)} 
  = \PrimeFact{\m} = \m 
    & & \Comment{\eqref{eq:N4_f}},
\end{align*}
which again contradicts the assumption $\vert\A\vert < \m$,
so we must have $\delta < f(\m)$.

In order to write the prime factorization of $\vert\A\vert$,
let $\{q_1,\dots,q_{\rho}\}$ denote the set of 
primes which are less than $f(\m)$ and do not divide $\m$.
Each prime less than $f(\m)$ 
either divides $\m$ and is in the set $\{p_1,\dots,p_{\omega(\m)}\}$
or it does not divide $\m$ and is in the set $\{q_1,\dots,q_{\rho}\}$.
Thus $\delta$ must be a product of $q_1,\dots,q_{\rho}$ and $p_1,\dots,p_{\omega(\m)}$ terms, 
so there exist
$\alpha_1,\dots,\alpha_{\omega(\m)} \geq 1$
and $\beta_1,\dots,\beta_{\rho} \geq 0$
such that we can write $\vert \A \vert$ as
\begin{align}
  \vert \A \vert = p_1^{\alpha_1} \dots p_{\omega(\m)}^{\alpha_{\omega(\m)}} \, q_1^{\beta_1} \dots q_{\rho}^{\beta_{\rho}}.
  \label{eq:N4_solv_0}
\end{align}

For each prime $q < f(\m)$
such that $\NDiv{q}{\m}$,
by \eqref{eq:N4_N4},
network $\Network_4(\m)$ contains a copy of $\Network_1(q)$.
Since network $\Network_1(q)$ is solvable over $\A$,
then by Lemma~\ref{lem:N1_solv}, we have $\GCD{q}{\vert \A \vert} = 1$.
Thus in \eqref{eq:N4_solv_0} we have $\beta_1 = \dots = \beta_{\rho} = 0$.

For each $i = 1,\dots,\omega(\m)$ such that $\gamma_i > 1$,
by \eqref{eq:N4_N4},
network $\Network_4(\m)$ contains a copy of $\Network_3(p_i, \, g(\m,i))$.
Since network $\Network_3(p_i, \, g(\m,i) )$ is solvable over $\A$
and $\Div{p_i}{\vert \A \vert}$,
then by Lemma~\ref{lem:N3_solv},
$\vert \A \vert$ does not divide $g(m,i)$.
Expressing $\vert\A\vert$ and $g(m,i)$ as their prime factorizations yields:
\begin{align*}  
  & p_1^{\alpha_1}  \dots p_{\omega(\m)}^{\alpha_{\omega(\m)}} 
  \, \not{\hspace{-.225mm}\Big\vert} \; \;
   p_i^{\gamma_i-1} \prod_{\substack{j = 1 \\ j \ne i}}^{\omega(\m)} p_j^{\mu(\m,j)}  
    & & \Comment{\eqref{eq:N4_g}, \eqref{eq:N4_solv_0}}.
\end{align*}
This implies that for each $i\in \{1, \dots, \omega(\m)\}$
such that $\gamma_i > 1$, 
either $\alpha_i \geq \gamma_i$
or 
$\alpha_j \geq \mu(\m,j) + 1$ for some $j \ne i$.

If there exists $j \in \{1,\dots,\omega(\m)\}$
such that that $\alpha_j \geq \mu(\m,j) + 1$,
then we have
\begin{align*}
  \vert\A\vert &= p_1^{\alpha_1}  \cdots  p_{\omega(\m)}^{\alpha_{\omega(\m)}} 
        & & \Comment{\eqref{eq:N4_solv_0}}\\
      & \geq p_j^{\alpha_j-1} \left(p_1\cdots p_{\omega(\m)} \right) 
        & &  \Comment{ $\alpha_l \geq 1$} \\
      & \geq p_j^{\mu(\m,j)} \left(p_1 \cdots p_{\omega(\m)} \right)  \\
      & \geq  f(\m) \left(p_1\cdots p_{\omega(\m)} \right) = \m
        & &  \Comment{\eqref{eq:N4_f}, \eqref{eq:N4_u}},
\end{align*}
which contradicts the assumption that $\vert\A\vert < \m$.
So if each component of network $\Network_4(\m)$ is solvable over $\A$ and $\vert \A \vert < \m$, 
it must be the case that
$\alpha_i \geq \gamma_i$,
for each $i$ such that $\gamma_i > 1$.
If $\gamma_i = 1$, then $\alpha_i \geq 1 = \gamma_i$.
So we have $\alpha_i \geq \gamma_i$ for all $i$,
but this implies
\begin{align*}
\vert\A\vert &= p_1^{\alpha_1}  \cdots p_{\omega(\m)}^{\alpha_{\omega(\m)}} 
      & & \Comment{\eqref{eq:N4_solv_0}} \\
    & \geq \PrimeFact{\m} = \m,
\end{align*}
which again contradicts the assumption that $\vert \A \vert < \m$.

Thus there does not exist an alphabet $\A$ whose size is less than $\m$
such that each disjoint component of $\Network_4(\m)$ is solvable over $\A$.
\end{proof}

\begin{example}
We continue our example networks $\Network_4(6), \Network_4(27),$ and $\Network_4(100)$.
\begin{itemize}
  \item Suppose $\Network_4(6)$ is solvable over an alphabet $\A$.
  Since $\Network_2(2,3)$ is solvable over $\A$, we have $2$ divides $\vert\A\vert$.
  Similarly for $\Network_2(3,2)$, we have that $3$ divides $\vert\A\vert$.
  Since $6$ is the smallest positive integer that is divisible by $2$ and $3$,
  we have $\vert\A\vert \geq 6$.

  \item Suppose $\Network_4(27)$ is solvable over an alphabet $\A$
  whose size is less than $27$.
  Then
  \begin{itemize}
    \item $\Network_2(27,1)$ requires $\Div{3}{\vert\A\vert}$, 
  so $\vert\A\vert \in \{3,6,9,12,15,18,21,24\}$.

    \item $\Network_1(2)$, $\Network_1(5)$, and $\Network_1(7)$ require
  $\vert\A\vert$ be relatively prime to $2,$ $5$, and $7$,

  so $\vert\A\vert \not\in \{6,12,15,18,21,24\}$.

    \item $\Network_3(3,9)$ requires $\NDiv{\vert\A\vert}{9}$,
  so $\vert\A\vert \not\in \{3,9\}$.
\end{itemize}
  Therefore $\Network_4(27)$ is not solvable over any alphabet whose size is less than $27$.

  \item Suppose $\Network_4(100)$ is solvable over an alphabet $\A$
    whose size is less than $100$.
    Then
    \begin{itemize} 
    \item $\Network_2(4,25)$ and $\Network_2(25,4)$
    require $\Div{10}{\vert \A \vert}$,
    so $\vert \A \vert \in \{10,20,\dots,90\}$.

    \item $\Network_1(3)$ and $\Network_1(7)$ require $\vert\A\vert$ to be relatively prime to $3$ and $7$,
    so $\vert \A \vert \not\in \{30,60,70,90\}$.

    \item $\Network_3(2,50)$ requires $\NDiv{\vert\A\vert}{50}$,
    so $\vert \A \vert \not\in \{10,50\}$.

    \item $\Network_3(5,80)$ requires $\NDiv{\vert \A \vert}{80}$,
    so $\vert \A \vert \not\in \{10,20,40,80\}$.
    \end{itemize}
    Therefore $\Network_4(100)$ is not solvable over any alphabet whose size is less than $100$.

\end{itemize}
\label{ex:N4_solv}
\end{example}

\subsection{Linear solvability conditions of \texorpdfstring{$\Network_4(\m)$}{N4(m)}} \label{ssec:N4_lin}

The following theorems show that $\Network_4(\m)$
is linear solvable if and only if $\m$ is prime.
\begin{theorem}
  For each prime $p$, 
  network $\Network_4(p)$ is scalar linear solvable 
  over $\GF{p}$.
  \label{thm:N4_prime}
\end{theorem}
\begin{proof}
  If $p$ is a prime number,
  then $f(p) = 1$ and the power of $p$ is one,
  so by \eqref{eq:N4_N4}, network $\Network_4(p)$ consists solely of a copy of network $\Network_2(p,1)$.
  By Lemma~\ref{lem:N2_lin}, network $\Network_2(p,1)$
  has a scalar linear solution over every finite-field alphabet with characteristic $p$. 
\end{proof}
  
\begin{theorem}
  For each composite number $\m$,
  network $\Network_4(\m)$ is not vector linear solvable over any $R$-module.
  \label{thm:N4_R}
\end{theorem} 
\begin{proof}
Let $G$ be a \niceRModule{},
and assume a scalar linear solution for $\Network_4(\m)$ exists over $G$.
Since $\Network_4(\m)$ is scalar linear solvable over $G$,
each disjoint component of $\Network_4(\m)$ is scalar linear solvable over $G$.
Suppose $\m$ is a composite number.
Then $\m$ is a product of two or more (possibly distinct) primes. 
We will separately consider the cases of prime powers
and non-power-of-prime composite numbers.

For each prime $p$ and integer $\gamma \geq 2,$
by \eqref{eq:N4_N4},
network $\Network_4(p^{\gamma})$ contains
copies of $\Network_2(p^{\gamma},1)$ 
and $\Network_3\left(p, \, p^{\gamma-1} \right)$.
Since network $\Network_2(p^{\gamma},1)$ 
is scalar linear solvable over $G$,
by Lemma~\ref{lem:N2_lin},
the characteristic of $R$ divides $p^{\gamma}$.
Since network $\Network_3\left(p, \, p^{\gamma-1} \right)$ 
is scalar linear solvable over $G$,
by Lemma~\ref{lem:N3_lin},
the characteristic of $R$ is relatively prime to $p$.
If the characteristic of $R$ both divides $p^{\gamma}$ and is relatively prime to $p$,
then the characteristic of $R$ is $1$,
which only occurs in the trivial ring (of size one).
Thus there is no \niceRModule{} over which all components 
of network $\Network_4(p^{\gamma})$ are scalar linear solvable. 

Now suppose $\omega(\m) \geq 2$. Then
$\m$ has prime factorization $\m = \PrimeFact{\m}$,
and by \eqref{eq:N4_N4},
network $\Network_4(\m)$ contains copies of
$\Network_2\left(p_1^{\gamma_1},\, (\m /p_1^{\gamma_1}) \right)$
and network $\Network_2\left(p_2^{\gamma_2},\,  (\m / p_2^{\gamma_2} ) \right)$.
Since network $\Network_2\left(p_i^{\gamma_i},\, (\m / p_i^{\gamma_i}) \right)$ 
is scalar linear solvable over $G$,
by Lemma~\ref{lem:N2_lin},
the characteristic of $R$ divides $p_i^{\gamma_i}$.
For primes $p_1 \ne p_2,$
if the characteristic of $R$ divides both $p_1^{\gamma_1}$ and $p_2^{\gamma_2}$
then the characteristic of $R$ is $1$,
which only occurs in the trivial ring.
Thus there is no \niceRModule{} over which all components 
of network $\Network_4(\m)$ are scalar linear solvable. 

If $\m$ is a composite number,
then there are no scalar linear solutions for $\Network_4(\m)$ over any \niceRModule{},
which, by Lemmas~\ref{lem:mod_2} and \ref{lem:mod_0}
implies there are
no vector linear solutions for $\Network_4(\m)$ over any $R$-module.
\end{proof}

\subsection{Capacity and linear capacity of \texorpdfstring{$\Network_4(\m)$}{N4(m)}} \label{ssec:N4_cap}
\begin{theorem}
  For each $\m \geq 2$
  network $\Network_4(\m)$
  has:
  \begin{itemize}
  \itemsep0em 
  \item[(a)] capacity equal to $1$,
  \item[(b)] linear capacity bounded away from $1$ 
  over all finite-field alphabets, if $\m$ is composite.
  \end{itemize}
  \label{thm:N4_cap}
\end{theorem}
\begin{proof}

For each $m \geq 2,$ by Theorem~\ref{thm:N4_m}, network $\Network_4(\m)$ is solvable over an alphabet of size $\m$, 
so its capacity is at least $1$.
Each network $\Network_1,\Network_2,\text{ and }\Network_3$ has capacity equal to $1$,
and $\Network_4(\m)$ consists of disjoint copies of $\Network_1,\Network_2, \text{ and } \Network_3$,
so its capacity is at most $1$.
Thus the capacity of $\Network_4(\m)$ is equal to $1$.

For composite $\m$,
we will again separately consider the cases of prime powers
and non-power-of-prime composite numbers. 

For each prime $p$ and integer $\gamma \geq 2,$
by \eqref{eq:N4_N4},
network $\Network_4(p^{\gamma})$ contains copies of  
$\Network_2(p^{\gamma},1)$ 
and $\Network_3\left(p, \, p^{\gamma-1} \right)$.
By Lemma~\ref{lem:N2_cap},
network $\Network_2(p^{\gamma},1)$ has linear capacity 
upper bounded by 
$$1-\frac{1}{2 p^{\gamma} + 3}$$ 
for finite-field alphabets with characteristic other than $p$.
By Lemma~\ref{lem:N3_cap},
network $\Network_3\left(p, \, p^{\gamma-1} \right)$ has linear capacity
equal to 
$$1 - \frac{1}{2 p^{\gamma-1} + 2 p + 3}$$
for finite-field alphabets with characteristic $p$. 
Whether we select a finite-field alphabet with characteristic $p$
or characteristic other than $p$,
the linear capacity of $\Network_4(p^{\gamma})$ is bounded away from $1$,
for fixed $p$ and $\gamma$.

Now suppose $\omega(\m) \geq 2$.
Then $\m$ has prime factorization $\m= \PrimeFact{\m}$,
and by \eqref{eq:N4_N4},
network $\Network_4(\m)$ contains copies of  
$\Network_2\left(p_1^{\gamma_1},\, (\m / p_1^{\gamma_1}) \right)$
and $\Network_2\left(p_2^{\gamma_2},\, (\m / p_2^{\gamma_2}) \right)$.
By Lemma~\ref{lem:N2_cap},
network $\Network_2\left(p_i^{\gamma_i},\, (\m / p_i^{\gamma_i}) \right)$ has linear capacity 
upper bounded by 
$$1 - \frac{1}{2 \m  + 2 (\m / p_i^{\gamma_i}) + 1}$$ 
for finite-field alphabets with characteristic other than $p_i$.
Since $p_1 \ne p_2,$
whether we select a finite-field alphabet with characteristic $p_1,p_2,$
or neither $p_1$ nor $p_2,$
the linear capacity is bounded away from $1$,
for fixed $\m$.

Thus for any fixed composite number $\m$,
the linear capacity of network $\Network_4(\m)$ 
is bounded away from $1$
over all finite-field alphabets.
\end{proof}
Calculating the exact linear capacity of $\Network_4(\m)$ over every finite-field alphabet
is left as an open problem.

\begin{corollary}
  For each composite $\m$,
  network $\Network_4(\m)$ is not asymptotically linear solvable
  over any finite-field alphabet.
  \label{cor:N4_asymp}
\end{corollary}
\begin{proof}
  This follows directly from the fact
  that for any fixed composite number $\m$, 
  by Theorem~\ref{thm:N4_cap},
  the linear capacity of $\Network_4(\m)$ is bounded
  away from one over all finite-field alphabets.
\end{proof}

\subsection{Size of \texorpdfstring{$\Network_4(\m)$}{N4(m)} } \label{ssec:N4_size}
Depending on the prime divisors of $\m$,
the number of nodes in $\Network_4(\m)$
can be dominated by nodes from $\Network_1$ networks, $\Network_2$ networks, or $\Network_3$ networks.
The following theorem makes use of the functions $f(\m),$ $\mu(\m,i),$ and $g(\m,i)$
defined in \eqref{eq:N4_f}, \eqref{eq:N4_u}, \eqref{eq:N4_g}.

\begin{theorem}
  For each $\m \geq 2,$
  the number of nodes in network $\Network_4(\m)$ is asymptotically
  \begin{itemize}   
  \itemsep0em 
    \item[(a)] $\Omega(\m) $,
    \item[(b)] $O(\m)$, when $\m$ is prime,
    \item[(c)] $O\left(\frac{ \m \log{\m}}{\log{\log{\m}}} \right)$,  
    when $\m$ is square-free, 
    \item[(d)] $O\left( \m^2 / \log{\m} \right)$, when $\m$ is a prime-power,
    \item[(e)] $O\left(  \m^{\frac{\log{\m}}{\log{\log{\m}}}} \right)$,
    when $\m$ is neither square-free nor a prime-power.
  \end{itemize}
  \label{thm:N4_nodes}
\end{theorem}
\begin{proof}

By Remark~\ref{rem:N1_nodes},
the number of nodes in $\Network_1(q)$ is $4 q+ 7$.

By Remark~\ref{rem:N2_nodes},
the number of nodes in $\Network_2(\m,\n)$ is $4 \m \n + 9 \n + 2$.

By Remark~\ref{rem:N3_nodes},
the number of nodes in $\Network_3(\m_1,\m_2)$ is $4 \m_1 + 4 \m_2 + 12$.

By the construction of $\Network_4(\m)$ given in \eqref{eq:N4_N4},
the total number of nodes in $\Network_4(\m)$ is:
\begin{align}
& \left( \sum_{\substack{\text{prime q} \\ \NDiv{q}{\m} \\  q < f(\m)}}  
    (4q + 7) \right)
  + \left( \sum_{i = 1}^{\omega(\m)} 
    (4\m   + 9 (\m /p_i^{\gamma_i}) + 2) \right)
  + \left( \sum_{\substack{i = 1 \\ \gamma_i > 1}}^{\omega(\m)} 
    (4 g(\m,i) + 4 p_i + 12)  \right)
    \label{eq:N4_nodes_1}
\end{align}
where the first, second, and third terms 
are the number of nodes from $\Network_1$, $\Network_2,$ and $\Network_3$ networks,
respectively.
In order to find upper and lower bounds  
on the total number of nodes in $\Network_4(\m)$,
we will first find upper and lower bounds 
on the number of nodes from $\Network_1,\Network_2,$ and $\Network_3$
networks within $\Network_4(\m)$.

It is known \cite[VII.27a]{Sandor-NumberTheory}
that 
\begin{align} 
  \sum_{\substack{\text{prime } q \\ q \leq \m}} q = O\left( \frac{\m^2}{\log{\m}} \right) .
    \label{eq:N4_nodes_2}
\end{align}

If $\m$ is a square-free number,
then we have $f(\m) = 1$,
so in this case, there are no nodes in $\Network_4(\m)$ from $\Network_1$ networks.
Thus for general $\m$, we have
\begin{align}
  \sum_{\substack{\text{prime q} \\ \NDiv{q}{\m} \\  q < f(\m)}} (4q + 7) 
    &\geq 0 \label{eq:N4_nodes_3}
\end{align}
and
\begin{align} 
    \sum_{\substack{\text{prime q} \\ \NDiv{q}{\m} \\  q < f(\m)}} (4q + 7) 
      & < \sum_{\substack{\text{prime } q \\ q \leq \m}} (4q +7) 
     = O\left( \frac{\m^2}{\log{\m}} \right)
     && \Comment{\eqref{eq:N4_nodes_2}}.
     \label{eq:N4_nodes_4}
\end{align}

The total number of nodes in $\Network_4(\m)$ from $\Network_2$ networks is
\begin{align}
  \sum_{i = 1}^{\omega(\m)} (4 \m  + 9 (\m/ p_i^{\gamma_i}) + 2)
    & > \sum_{i = 1}^{\omega(\m)} 4 \m = \Omega\left(\omega(\m) \, \m \right)
\label{eq:N4_nodes_5} 
\end{align}
and
\begin{align}
  \sum_{i = 1}^{\omega(\m)} (4 \m  + 9 (\m/ p_i^{\gamma_i}) + 2)
    & < \sum_{i = 1}^{\omega(\m)} (13 \m + 2) = O\left(\omega(\m) \, \m \right).
\label{eq:N4_nodes_6}
\end{align}

For each $i = 1,\dots,\omega(\m)$ we have
\begin{align}
    p_i^{\mu(\m,i)} & < p_i \, f(\m)
    & & \Comment{\eqref{eq:N4_u}}
    \label{eq:N4_nodes_7}\\
    g(\m,i) &= p_i^{\gamma_i - 1} \, \prod_{\substack{j = 1 \\ j \ne i}}^{\omega(\m)} p_j^{\mu(\m,j)}
        & & \Comment{\eqref{eq:N4_g}}  \notag \\        
      & < p_i^{\gamma_i - 1} \, \prod_{\substack{j = 1 \\ j \ne i }}^{\omega(\m)} p_j f(\m) 
        & & \Comment{\eqref{eq:N4_nodes_7}} \notag\\
      & < p_i^{\gamma_i} \, f(\m)^{\omega(\m) - 1} \, \prod_{ j = 1}^{\omega(\m)} p_j  \notag \\
       & = p_i^{\gamma_i} \, f(\m)^{\omega(\m) - 2} \, \m
        & & \Comment{\eqref{eq:N4_f}}.
        \label{eq:N4_nodes_8}
\end{align}

If $\m$ is square-free,
then $\gamma_i = 1$ for all $i$,
so in this case,
there are no nodes in $\Network_4(\m)$ from $\Network_3$ networks.
Thus for general $\m$, we have
\begin{align}
  \sum_{\substack{i = 1 \\ \gamma_i > 1}}^{\omega(\m)} 
    (4 g(\m,i) + 4 p_i + 12) & \geq 0.
    \label{eq:N4_nodes_9}
\end{align}
and
\begin{align}
 \sum_{\substack{i = 1 \\ \gamma_i > 1}}^{\omega(\m)} 
    (4 g(\m,i) + 4 p_i + 12)
       &\leq \sum_{i = 1}^{\omega(\m)} 20 g(\m,i)
          & & \Comment{\eqref{eq:N4_g}}\notag \\
       &< 20 \m \, f(\m)^{\omega(\m) - 2} \, \sum_{i = 1}^{\omega(\m)} p_i^{\gamma_i} 
        & & \Comment{\eqref{eq:N4_nodes_8}} \notag\\
       &< 20 \m \, f(\m)^{\omega(\m) - 2} \, \prod_{i = 1}^{\omega(\m)} p_i^{\gamma_i}
        & & \Comment{$a b \geq a +b$ for all $a,b \geq 2$}\notag \\
       &= 20 \m^2 \, f(\m)^{\omega(\m) - 2} \notag \\
       & < 20 \m^{\omega(\m)} = O\left(\m^{\omega(\m)} \right) 
        & & \Comment{\eqref{eq:N4_f}}.
 \label{eq:N4_nodes_10}
\end{align}

To prove part (a), consider the lower bounds of each term of \eqref{eq:N4_nodes_1}.
The total number of nodes in $\Network_4(\m)$ is lower bounded by:
\begin{align*}
0 + \Omega(\omega(\m) \, \m) + 0
  &= \Omega( \omega(\m) \, \m) = \Omega(\m)
    && \Comment{\eqref{eq:N4_nodes_1}, \eqref{eq:N4_nodes_3}, \eqref{eq:N4_nodes_5}, \eqref{eq:N4_nodes_9}},
\end{align*}
where the final equality comes from the fact 
$\omega(\m) = \Omega(1)$,
since $\omega(\m) = 1$ when $\m$ is prime.

It follows from \cite[Theorem~11]{Robin} that 
\begin{align}
  \omega(\m) &= O\left( \frac{\log{\m}}{\log{\log{\m}}} \right)
  \label{eq:N4_nodes_11}.
\end{align}

To prove parts (b)-(e),
we will consider the upper bounds on the number of nodes of each term of \eqref{eq:N4_nodes_1}.
However, each term  dominates in different cases,
depending on the prime factors of $\m$.

To prove parts (b) and (c),
consider a square-free integer $\m = p_1 \cdots p_{\omega(\m)}$.
Since $\gamma_i = 1$ for all $i$,
we have $f(\m) = 1$,
so there are neither $\Network_1$ nor $\Network_3$ components in $\Network_4(\m)$.
Thus there are $0$ nodes from $\Network_1$ and $\Network_3$ components.
Then by \eqref{eq:N4_nodes_1} and \eqref{eq:N4_nodes_6}, 
the number of nodes in $\Network_4(\m)$ is $O( \omega(\m) \, \m )$.
If $\m$ is prime, then $\omega(\m) = 1$, so we have the desired bound.
If $\m$ is not prime, then the number of nodes is $O(\omega(\m) \, \m)$,
which, along with \eqref{eq:N4_nodes_11}, yields the desired bound.

To prove part (d),
consider a prime power $\m = p^{\gamma}$,
where $\gamma \geq 2$.
We have $\omega\left(p^{\gamma}\right) = 1$,
so by \eqref{eq:N4_nodes_6},
the number of nodes from $\Network_2$ components is $O(\m)$,
and, by \eqref{eq:N4_nodes_10}, 
the number of nodes from $\Network_3$ components is $O(\m)$.
By \eqref{eq:N4_nodes_4},
the number of nodes from $\Network_1$ components is
$O(\m^2/ \log{\m}) $.
Thus the number of nodes in $\Network_4(\m)$ is $O(\m^2 / \log{\m} )$.

To prove part (e),
consider $\m$ which is neither a prime power (so $\omega(\m) \geq 2$)
nor square-free (so there are $\Network_3$ components in $\Network_4(\m)$).
The number of nodes in $\Network_4(\m)$ is
\begin{align*}
  &O\left( \frac{\m^2}{\log{\m}} \right) 
    + O\left(\omega(\m) \, \m \right) 
    + O\left( \m^{\omega(\m)} \right)
    & & \Comment{\eqref{eq:N4_nodes_1}, \eqref{eq:N4_nodes_4}, \eqref{eq:N4_nodes_6}, \eqref{eq:N4_nodes_10}} \\
  &= O\left( \m^{\omega(\m)} \right)
    & & \Comment{$\omega(\m) \geq 2$},
\end{align*}
which, along with \eqref{eq:N4_nodes_11}, yields the desired bound.
\end{proof}

\begin{example}
We continue our example networks $\Network_4(6), \Network_4(27),$ and $\Network_4(100)$.
\begin{itemize}
  \item $\Network_4(6)$ has $97$ nodes:
$53$ from $\Network_2(2,3)$
and $44$ from $\Network_2(3,2)$.

  \item $\Network_4(27)$ has $256$ nodes:
$15$ from $\Network_1(2)$, 
$27$ from $\Network_1(5)$,
$35$ from $\Network_1(7)$,
$119$ from $\Network_2(27,1)$,
and $60$ from $\Network_3(3,9)$.

  \item $\Network_4(100)$ has $1691$ nodes:
 $19$ from $\Network_1(3)$, 
 $35$ from $\Network_1(7)$, 
 $627$ from $\Network_2(4,25)$,
 $438$ from $\Network_2(25,4)$, 
 $220$ from $\Network_3(2,50)$, and 
 $352$ from $\Network_3(5,80)$.
\end{itemize}
\label{ex:N4_nodes}
\end{example}

\section{Open Questions} \label{sec:Q}

Below are some remaining open questions regarding
linear and non-linear solvability:
\begin{enumerate}
  
  \item In \cite{Dougherty-Freiling-Zeger04-Insufficiency} it was shown that
  there exists a network which is not vector linear solvable over any $R$-module
  yet is non-linear solvable over an alphabet of size $4$.
  We have shown that for each composite number $\m$,
  there exists a network which is not vector linear solvable over any $R$-module
  yet is non-linear solvable over an alphabet of size $\m$.
  Do there exist networks which are not vector linear solvable over $R$-modules
  but are non-linear solvable over some alphabet of prime size?

  \item There are examples \cite{ChenHaiBin-Characterization}, \cite{Lehman-Complexity}
  in the literature of solvable networks
  which are not solvable over any alphabet whose size is less than some $\m$.
  For each $\m \geq 2$, we have demonstrated a network which is solvable over an alphabet of size $\m$
  but is not solvable over any alphabet whose size is less than $\m$.
  For each $\m \geq 2$ does there exist a network which is solvable over alphabet $\A$
  if and only if $\vert\A\vert \geq \m$?
  Which other ``interesting'' sets $S \subset \N$
  have the property that there exists a network which is solvable over $\A$
  if and only if $\vert\A\vert \in S$?

  \item It is not currently known whether there can exist an algorithm
  which determines whether a network is solvable.
  We have demonstrated a class of solvable networks with no vector linear solutions (i.e. diabolical networks).
  Can there exist an algorithm which detects whether a network is diabolical?

\end{enumerate}

\clearpage

\makeatletter\renewcommand{\@seccntformat}[1]{}\makeatother
\appendix
\section{Appendix - Proofs of Lemmas}

\subsection{Proofs of Lemmas in Section~\ref{sec:intro}}
\ProofOflemModTwo
\ProofOflemModZero
\ProofOflemRingInvOne
\ProofOflemRingInvTwo

\clearpage

\subsection{Proofs of Lemmas in Section~\ref{sec:N0}}
\ProofOflemNZeroP
\ProofOflemNZeroLin
\ProofOflemNZeroCap

\clearpage

\subsection{Proofs of Lemmas in Section~\ref{sec:N1}}
\ProofOflemNOneSolv
\ProofOflemNOneLin
\ProofOflemMatTwo
\ProofOflemMatThree
\ProofOflemNOnecap

\clearpage

\subsection{Proofs of Lemmas in Section~\ref{sec:N2}}
\ProofOflemNTwoP
\ProofOflemNTwoNon
\ProofOflemNTwoSolv
\ProofOflemNTwoLin
\ProofOflemNTwoCap

\clearpage

\subsection{Proofs of Lemmas in Section~\ref{sec:N3}}
\ProofOflemNThreeP
\ProofOflemNThreeNon
\ProofOflemNThreeSolv
\ProofOflemNThreeLin
\ProofOflemNThreeCap

\clearpage

\renewcommand{\baselinestretch}{1.0}

\end{document}